\documentclass{amsart}
\usepackage{amssymb,amsmath,amsthm,mathtools, geometry, amssymb, verbatim, esint, soul}
\usepackage{a4wide}
\usepackage{float}
\usepackage{graphicx}
\usepackage[utf8]{inputenc} 
\usepackage{a4wide}
\usepackage{tikz}

\usepackage{tkz-euclide}
\usetikzlibrary{calc}
\usepackage{ulem}
\usepackage{stackengine,scalerel,graphicx}
\stackMath
\usepackage{hyperref}
\usepackage{amsfonts} 
\usepackage{latexsym}
\usepackage[font=small,format=hang,labelfont={sf,bf}]{caption}
\usepackage{epsfig}
\usepackage{enumitem}
\usepackage{subfig}
\usepackage{url}
\usepackage{varioref}

\usepackage{bm}
\usepackage{color}
\usepackage{mathrsfs}
\usepackage{latexsym}
\usepackage{bbm}
\usepackage{faktor}
\usepackage{yhmath}
\usepackage{empheq}
\usepackage{mathtools}
\usepackage{color}
\usepackage{marginnote}
\usetikzlibrary{patterns}
\usepackage{comment}
\usepackage{cancel}
\usetikzlibrary{arrows.meta}

\allowdisplaybreaks[1]

\newcommand{\dist}{\mathrm{dist}\,}          

\DeclareMathOperator{\p}{\partial}

\DeclareMathOperator{\id}{Id}

\DeclareMathOperator{\R}{\mathbb{R}}

\DeclareMathOperator{\Z}{\mathbb{Z}}
\DeclareMathOperator{\N}{\mathbb{N}}

\newcommand*{\mc}[1]{\mathcal{#1}}
\DeclareMathOperator{\lra}{\Leftrightarrow}
\newcommand{\vertiii}[1]{{\left\vert\kern-0.25ex\left\vert\kern-0.25ex\left\vert #1 
		\right\vert\kern-0.25ex\right\vert\kern-0.25ex\right\vert}}
\def\avint{\mathop{\,\rlap{-}\kern-1.1ex\int}\nolimits}

\newtheorem{Theorem}{Theorem}[section]
\newtheorem{Def}[Theorem]{Definition}

\newtheorem{prop}[Theorem]{Proposition}
\newtheorem{lem}[Theorem]{Lemma}
\newtheorem{cor}[Theorem]{Corollary}
\newtheorem{rem}[Theorem]{Remark}

\newcommand{\eps}{\varepsilon}

\newcommand{\EEE}{\color{black}}

\makeatletter
\renewcommand{\theequation}{\thesection.\arabic{equation}}
\@addtoreset{equation}{section}
\makeatother

\newtheorem{theorem}{Theorem}[section]
\newtheorem*{theorem*}{Theorem}

\theoremstyle{definition}
\newtheorem{definition}[theorem]{Definition}
\newtheorem{remark}[theorem]{Remark}

\title[Emergence of Polycrystals]{Emergence of rigid polycrystals from atomistic systems with general interactions}

\author[L.~Kreutz]{Leonard~Kreutz}
\address[Leonard Kreutz]{Technische Universit\"at M\"unchen, Boltzmannstrasse 3, 85748 Garching, Germany	}
\email[]{kleo@cit.tum.de}
\author[T.~Ziereis]{Timo~Ziereis}
\address[Timo Ziereis]{Technische Universit\"at M\"unchen, Boltzmannstrasse 3, 85748 Garching, Germany	}
\email[]{timo.ziereis@tum.de}

\date{\today}	

\newcounter{fig}
\newcounter{thm}
\begin{document}

\begin{abstract}
We investigate the formation of polycrystalline structures in a class of particle systems. The atomistic energy is modeled as a sum of particle energies that favor atoms being locally isometric to a reference lattice. The discrete frame invariant energy allows for particle configurations in which no underlying lattice is assumed a priori.

We prove a discrete-to-continuum limit for configurations with finite surface-energy scaling by means of $\Gamma$-convergence. The resulting continuum theory is described by piecewise constant fields encoding the local orientation of the configuration. The limiting energy is concentrated on grain boundaries, corresponding to the interfaces between regions where the microscopic configuration has constant orientation. The associated energy density depends on the orientations of the two grains as well as on the normal to the interface.

Due to our assumptions on the rigid interactions, solid–solid phase transitions with interpolating boundary layers are not energetically favorable; the energy density therefore decomposes into twice the energy density for solid–vacuum transitions.
\end{abstract}

\maketitle

	\vskip5pt
	\noindent
	\textsc{Keywords: Polycrystals, Crystallization, atomic interaction potentials,\\ $\Gamma$-convergence} 
	\vskip5pt
	\noindent
	\textsc{AMS subject classifications: 49J45, 82D25, 82B24}

\tableofcontents

\section{Introduction}

Crystallization \cite{BlancLewin:15} is a fundamental phenomenon in condensed matter physics and materials science, describing the spontaneous organization of atoms into ordered lattice structures. At the microscopic level, this process is governed by quantum-mechanical interactions between atoms, which determine effective interatomic potentials \cite{Molecular}. Understanding how these interactions give rise to specific atomic arrangements can be formulated mathematically as an energy minimization problem: stable configurations correspond to arrangements of atoms that minimize the total interaction energy. From a mathematical perspective, the study of such minimization problems provides insight into how microscopic interaction laws generate macroscopic structures, such as crystalline lattices and polycrystalline materials. In particular, analyzing low-energy configurations and their associated surface energies is essential for understanding the emergence of defects, grain boundaries, and other large-scale features observed in crystalline solids.

\medskip

At zero temperature, thermal fluctuations are absent, and the system’s behavior is determined entirely by the geometry of the atomic configuration. In this regime, stable states are expected to correspond to configurations that minimize the total interaction energy. From a mathematical standpoint, given a configuration of $n$ atoms with positions $\{x_1,\ldots,x_n\} \subset \mathbb{R}^d$ (typically $d=3$, corresponding to the physical space dimension) and a suitable configurational energy $\mathcal{E}$, the crystallization problem can be formulated as the minimization problem
\[
\min_{x_1,\ldots,x_n} \mathcal{E}(x_1,\ldots,x_n)\,.
\]

\medskip

Several rigorous results address crystallization in atomistic interaction models by analyzing the asymptotic behavior of ground states as the number of particles grows. For two-dimensional systems with pair potentials, \cite{Theil:06} and \cite{BeterminDeLucaPetrache} show that the ground-state energy per particle converges to that of a periodic lattice configuration, triangular and square, respectively. Similar results have been obtained for models including angular or multi-body interactions, where the limiting optimal structure is the hexagonal lattice, see e.g. \cite{ELi:09,Farmer-Esedoglu-Smereka}.  

Regarding finite crystallization, the rigorous results available in the literature are so far limited to two dimensions and typically require rather rigid interaction potentials. For instance, crystallization of two-dimensional hard disks on the triangular lattice was proved in \cite{HeitmannRadin:80} (see also \cite{DeLuca-DelNin,Radin:81} for generalizations), and recently revisited in \cite{Lucia} using an approach based on discrete differential geometry. For energies involving both two-body pair interactions and three-body angular terms, crystallization on the square and honeycomb lattices has been established in \cite{MaininiPiovanoStefanelli:14,Mainini}. The analysis has also been extended to ionic compounds, where crystallization on the square and honeycomb lattices was proved in \cite{FriedrichKreutz:19,FriedrichKreutz:20}. More recently, a new approach based on stratification has been developed, leading to crystallization results for the square lattice \cite{FriedrichKreutz:23}, as well as to results on relative crystallization \cite{FriedrichKreutzStefanelli}.

\medskip

In this article, we are interested in \textbf{next-order effects}. More precisely, assuming interatomic interactions that favor crystallization on a lattice $\mathcal{L}$ with energy per atom $\mathrm{m}$, we study configurations whose energy has a given excess relative to the ground state, of the form
\begin{align}\label{intro-ineq:excess-scaling}
\mathcal{E}(\{x_1,\ldots,x_n\}) - n \cdot \mathrm{m} \leq C n^{\frac{d-1}{d}}\,.
\end{align}
This \textbf{surface-scaling regime} allows for singularities along $(d-1)$-dimensional surfaces and is particularly relevant for the formation of \textbf{polycrystalline structures}, consisting of many microscopic crystallites or grains of varying size and orientation. Within each grain, atoms are periodically arranged in a crystalline pattern, while different grains meet at \textbf{grain boundaries}, i.e., $(d-1)$-dimensional interfaces across which the lattice orientation is discontinuous.

\medskip

We now describe the \textbf{model under investigation} and the main goals. Let $X := \{x_1,\ldots,x_n\} \subset \mathbb{R}^d$. We describe the normalized (i.e., relative to the ground-state energy) energy in terms of \textbf{cell energies}, a convenient way to encode local interactions, including possible multi-body terms such as triple-point angular interactions. The excess energy is given by
\begin{align*}
E_1(X) := \sum_{x \in X} E_{\mathrm{cell}}(x,X)\,,
\end{align*}
where the cell energy is normalized so that $E_1(\mathcal{L}) = 0$ and accounts for the local interactions of each point and possibly its neighbors. The precise assumptions on the cell energy are collected in Subsection~\ref{subsec:lattices}, and examples of possible interactions are given in the Appendix.  

To derive a \textbf{continuum description} of low-energy configurations, we introduce the scaled energies
\begin{align*}
E_\varepsilon(X) = \varepsilon^{d-1} \sum_{x \in X} E_{\mathrm{cell}}^\varepsilon(x,X)\,,
\end{align*}
with $E_{\mathrm{cell}}^\varepsilon(x,X) = E_{\mathrm{cell}}(\frac{1}{\varepsilon}x, \frac{1}{\varepsilon} X)$ and $\varepsilon = n^{-1/d}$ representing the typical interparticle distance. This scaling ensures consistency with the surface-scaling energy regime \eqref{intro-ineq:excess-scaling}. We are then interested in the limit $\varepsilon \to 0$ (equivalently $n \to \infty$) and in describing the continuum excess energy via \textbf{\(\Gamma\)-convergence}, see \cite{Braides:02,DalMaso:93}. Our goal is to understand the formation of polycrystals for general lattices in this regime. As a starting point, we follow the strategy of \cite{FriedrichKreutzSchmidt}, which provides a rigorous description of polycrystals on the triangular lattice. Motivated by a similar analysis for the honeycomb lattice, in which well-separated crystallites arise for strong angle potentials, we extract the main features of that approach to extend it to general lattices and energy systems.

\medskip

Our main results include a full \textbf{\(\Gamma\)-convergence proof} for the functionals $E_\varepsilon$ to a surface energy (Theorem~\ref{thm:Gamma-convergence}), and a detailed analysis of the limiting surface energy density (Proposition~\ref{prop:density}, Theorem~\ref{thm:properties-of-phi}) expressed via a \textbf{cell formula}. We also prove a \textbf{compactness result} for sequences with equi-bounded energy (Theorem~\ref{thm:compactness}). The lower bound and explicit characterization of the energy density rely on a structural result for grain boundaries under our specific, very rigid interaction potentials. The continuum description tracks both the orientation and (micro-)translations of grains. However, because interpolation between grains is energetically unfavorable, solid-solid boundaries effectively decompose into two solid-vacuum boundaries, and the surface energy ultimately depends only on the rotational mismatch, not on (micro-)translations.

\medskip

Finally, we comment on the \textbf{proof strategy}. Following standard variational techniques for interfacial energies, the limiting density $\varphi$ is expressed via a \textbf{cell formula} minimizing the asymptotic surface energy between two grains separated by a flat boundary. To establish the $\Gamma$-liminf and $\Gamma$-limsup inequalities, it is crucial to replace $L^1$-converging boundary data with fixed boundary values. This is achieved in three main steps:

\begin{enumerate}
\item Using the \textbf{fundamental estimate} (Section~\ref{sec:L1-to-converging-data}), we perform a cut-off construction to pass from $L^1$-convergence to converging boundary values. Due to the rigidity of the setup, small modifications must be handled carefully to preserve admissibility, see \ref{H1}.  
\item We analyze the minimization problem for interpolating between two grains separated by a flat boundary (Section~\ref{sec:reduction-two-lattices}) and show that, under additional rigidity assumptions, it suffices to consider only two well-separated lattices.  
\item Finally, using the reduction to well-separated lattices, we analyze the \textbf{cell problem for solid-vacuum grain boundaries} (Section~\ref{sec:vacuum-energy}). Combining this with the previous reduction allows us to treat general solid-solid grain boundaries by reducing them to the solid-vacuum case.
\end{enumerate}

\medskip

This paper is organized as follows. Section~\ref{Sec:Setting} introduces the model and states the main results precisely. Section~\ref{sec:3} is devoted to the proofs of compactness and $\Gamma$-convergence. Section~\ref{subsec:auxiliary-results} collects auxiliary results used throughout the paper. Section~\ref{sec:L1-to-converging-data} addresses the fundamental estimate, Section~\ref{sec:reduction-two-lattices} contains the reduction to two well-separated lattices, and Section~\ref{sec:vacuum-energy} proves the passage from converging to fixed boundary values for solid-vacuum grain boundaries. Finally, Section~\ref{sec:converging-fixed-bdryvalues} treats solid-solid grain boundaries and establishes properties of the limiting surface energy density $\varphi$.

\section{Setting and Main Results}\label{Sec:Setting}
In this section, we introduce our model, provide the basic definitions, and present our main results.

\noindent {\bf Notation.}
Throughout the article, we fix an integer $d \ge 2$.
We denote by $e_1, \ldots, e_d$ the standard orthonormal basis of $\R^d$.
For $x, y \in \R^d$, we write $\langle x, y \rangle$ for their scalar product and $|x|$ for the Euclidean norm.
The unit sphere is denoted by
\[
\mathbb{S}^{d-1} := \{ x \in \R^d \mid |x| = 1 \}\,.
\]
For $\nu \in \mathbb{S}^{d-1}$ we write $H^\nu :=\{x \in \mathbb{R}^d\mid \langle x,\nu\rangle=0\}$ and $H^\nu_\pm:=\{x\in\R^d\mid\pm\langle x,\nu\rangle\geq0\}$. For $t \in \R$, we write $\lfloor t \rfloor = \max\{ k \in \Z \mid k \le t \}$.
We denote by $\mathcal{L}^d$ the $d$-dimensional Lebesgue measure and by $\mathcal{H}^k$ (for $k \in \N$) the $k$-dimensional Hausdorff measure.
For a measurable set $E \subset \R^d$, we write $\chi_E$ for its characteristic function, equal to $1$ on $E$ and $0$ elsewhere.
If $E$ is a set of finite perimeter, $\partial^* E$ denotes its {\it essential boundary}, see \cite[Definition 3.60]{Ambrosio-Fusco-Pallara:2000}. We denote by $\mathrm{d}_H$ the Hausdorff distance of sets.

For $r > 0$ and $x \in \R^d$, we denote by $B_r(x)$ the open ball of radius $r$ centered at $x$; when $x = 0$ we simply write $B_r$.
We denote by $\omega_d$ the Lebesgue measure of the unit ball in $\R^d$.

Given $A \subset \R^d$, $\tau \in \R^d$, and $\lambda \in \R$, we set
\begin{align*}
A + \tau &:= \{ x + \tau \mid x \in A \}\,, \quad 
\lambda A &:= \{ \lambda x \mid x \in A \}\,, \quad
(A)_\varepsilon &:= A + B_\varepsilon = \{ x + y \mid x \in A,\, y \in B_\varepsilon \}\,.
\end{align*}
\subsection{Lattices and configurational energy} \label{subsec:lattices}  In this subsection, we introduce our configurations, the configurational energy, and its (local) optimal configurations. \\
\noindent {\bf Lattices and isometries.} In the following we will consider ideal crystals, see \cite{Dolbilin-Lagarias-Senechal}, which are sets $\mc{L}\subset\R^d$ such that there exists $l\in\N,\{x_1,\dots,x_l\}\subset\R^d$ and $L\in \mathrm{GL}(d,\R)$ such that
\begin{align}\label{def:lattice}
\mc{L}:=\bigcup_{i=1}^l(x_i+L\Z^n)\,.
\end{align}
Note that this implies, see \cite{Dolbilin-Lagarias-Senechal, Vince} for reference,
\begin{enumerate}[label=(L\arabic*)]
\item\label{L1} (Discreteness) There exists $r>0$ such that for all $x\in\R^d$ it holds
$$\#\{B_{r}(x)\cap\mc{L}\}\leq 1\,.$$
\item\label{L2} (No large empty regions) There exists $R>0$ such that for all $x\in\R^d$ it holds
$$\#\{B_{R}(x)\cap\mc{L}\}\geq1\,.$$
\item\label{L3} (Periodicity) Defining $v_i=Le_i$  ($i=1,\dots,d$) it holds
$$\mc{L}+v_i=\mc{L}\quad\text{ for all } i=1,\dots,d\,.$$
\item\label{L4} (Symmetry) There exists a finite subgroup $G\subset SO(d)$ such that for all $Q\in G$ it holds
$$Q\mc{L}=\mc{L}\,.$$
\end{enumerate} 
Given $x \in \mathcal{L}$ we denote the Voronoi cell with respect to $\mathcal{L}$ by
\begin{align}\label{def:Voronoi-cell}
V_{\mc{L}}(x):=\{y\in\R^d\mid \vert x-y\vert\leq\vert z-y\vert\ \forall z\in\mc{L}\}\,.
\end{align}
Note, that by \ref{L1}-\ref{L2} for every $x\in\mc{L}$ there exists $0 <r_V <R_V$ such that
\begin{align}\label{eq:ball-inclusion}
\overline{B}_{r_V}(x)\subseteq V_{\mc{L}}(x)\subseteq\overline{B}_{R_V}(x)\,.
\end{align} 
Here, $r_V>0$ is the inradius of $V_{\mathcal{L}}(x)$ and $R_V>0$ is the outerradius of $V_{\mathcal{L}}(x)$ (which we assume to be independent of $x \in \mathcal{L}$). Further, we define the Voronoi relevant distance as $S_V:=\sup\{\vert x-y\vert\mid x,y\in\mc{L},\mc{H}^{d-1}(V_\mc{L}(x)\cap V_\mc{L}(y))\neq0\}$. Note that this definition includes Bravais lattices ($l=1,x_1=0$), such as the triangular lattice in the plane, the integer lattice $\mathbb{Z}^d$, or the face-centered cubic lattice (FCC) in three dimensions, but also multi-lattice structures ($l\geq2$) such as the hexagonal closed-packed lattice (HCP) in three dimensions. Furthermore, note that for a rotated, translated and scaled lattice the Voronoi cell coincides with the rotated, translated and scaled Voronoi cell of our base lattice $\mc{L}$ and the above relations still hold up to scaling all three sets with the same parameter $\varepsilon$, i.e., it holds for all $x\in\mc{L}$:
\begin{align}\label{eq:rotated-Voronoi}
V_{\varepsilon R(\mc{L}+\tau)}(\varepsilon R(x+\tau))=\varepsilon R(V_{\mc{L}}(x)+\tau)\,, \quad \text{ and } \quad
\overline{B}_{r_V\varepsilon}(\varepsilon x)\subseteq V_{\varepsilon\mc{L}}(\varepsilon x)\subseteq \overline{B}_{R_V\varepsilon}(\varepsilon x)\,.
\end{align}
The vectors $v_i$ in \ref{L2} are called lattice vectors. For a fixed lattice $\mathcal{L}$, its lattice isometries can be characterized as follows.
Each distinct translation of $\mathcal{L}$ corresponds to a translation in 
\begin{align*}
\mathbb{T}:=\R^d/\mc{L}=\left\{\sum_{i=1}^d\lambda_iv_i\mid 0\leq\lambda_i<1\text{ for }i=1,\dots,d\right\}
\end{align*}
 while each rotation is determined by an element of
 \begin{align*}
 \mathbb{A}:=SO(d)/G\,.
 \end{align*}
The spaces $\mathbb{A}:=SO(d)/G$ and $\mathbb{T}$ are smooth manifolds by \cite[Theorem 21.29]{LeeIntroToManifolds}, which are also compact as they are given as the image of a compact set under their respective quotient maps. As the quotient maps are also smooth covering maps by \cite[Theorem 21.29]{LeeIntroToManifolds}, they are also local diffeomorphisms, which allows us to lift converging sequences. This will be needed for the proof of Lemma \ref{lem:fixed-bdryvalues}. Further, by the Whitney Embedding Theorem (see \cite[Theorem 6.15]{LeeIntroToManifolds}), both manifolds can be compactly embedded into some $\R^N$ for some $N\in \mathbb{N}$. Defining all lattice isometries as the action of $SO(d)\times\R^n$ on $\mc{L}$ via $(R,\tau)\mapsto R(\mc{L}+\tau)$, then there is a unique element $[R]\in\mathbb{A},[\tau]\in\mathbb{T}$ such that $R(\mc{L}+\tau)=[R](\mc{L}+[\tau])$. From now on, we will omit the equivalence classes and just write $R$ and $\tau$ instead of $[R]$ and $[\tau]$. We now introduce the set of lattice isometries
\begin{align}\label{def:isometryspace}
\mc{Z}:=(\mathbb{A}\times\mathbb{T}\times\{1\})\cup\{\textbf{0}\}\,, 
\end{align}
which can be compactly embedded into $\R^N$ for some $N\in\N$ with the product topology, that is $z_j=(R_j,\tau_j,1)\to z=(R,\tau,1)$ if and only if $R_j\to R$ and $\tau_j\to\tau$. Moreover, $z_j\to\textbf{0}$ if and only if $z_j=\textbf{0}$ for all $j$ large enough. Here, the $1$ as the third argument encodes that a lattice is present, whereas $\textbf{0}:=(0,0,0)$ represents vacuum. In particular for a given scale $\varepsilon$ and orientation $z=(R,\tau,1)$ we define the scaled, rotated and translated lattice as
\begin{align}
\mc{L}_\varepsilon(z):=\mc{L}_\varepsilon(R,\tau,1):=\varepsilon R(\mc{L}+\tau)\,,
\end{align}
whereas the vacuum corresponding to $\textbf{0}$ is defined as
\begin{align*}
\mc{L}_\varepsilon(\textbf{0}):=\emptyset\,.
\end{align*} 
If $\varepsilon=1$ we omit the dependence on $\varepsilon$ and write $\mathcal{L}(z) = \mathcal{L}_1(z)$.

\noindent {\bf Configurational energy.} We call a finite subset $X \subset \mathbb{R}^d$
 a configuration, and we denote by
 \begin{align*}
 \mathcal{X}:= \{ X \subset \mathbb{R}^d \mid X \text{ is finite}\}
 \end{align*}
 the set of all configurations. We then define a cell energy $E_{\mathrm{cell}} \colon \mathbb{R}^d \times \mathcal{X} \to [0,+\infty]$ satisfying the following properties: 
\begin{enumerate}[label=(E\arabic*)]
\item\label{H1}  (Discreteness) For every $X \in \mathcal{X}$ and $x \in X$ we have
\begin{align*}
E_{\mathrm{cell}}(x,X)<+\infty\quad \iff \quad\dist(\{x\},X\setminus\{x\})\geq1\,.
\end{align*}
\item\label{H2} (Boundedness) There exists $C>0$ such that  for every $X \in \mathcal{X}$ and every   $x\in X$ with $E_{\mathrm{cell}}(x,X)<+\infty$ we have
$
0\leq E_{\mathrm{cell}}(x,X)\leq C.
$ 
\item\label{H3} (Crystallization) There exists a radius  $r_{\mathrm{crys}} \geq\max\{R_V,S_V\}> 0$ with the following property:  For every $X \in \mathcal{X}$ and every $x \in X$, we have
\begin{align*}
E_{\mathrm{cell}}(x,X) = 0
\quad \iff \quad \begin{split}
&\text{there exists } (R,\tau) \in SO(d) \times \mathbb{R}^d  \\
&\text{such that } 
X \cap \overline{B}_{r_{\mathrm{crys}}}(x) = R(\mathcal{L} + \tau) \cap \overline{B}_{r_{\mathrm{crys}}}(x)\,.
\end{split}
\end{align*}
\item\label{H4} (Locality) There exists $r_{\mathrm{int}}  \geq r_{\mathrm{crys}} > 0$ with the following property. For every $X,Y \in \mathcal{X}$ with $X\cap\overline{B}_{r_{\mathrm{int}}}(x)=Y\cap\overline{B}_{r_{\mathrm{int}}}(x)$ we have $E_{\mathrm{cell}}(x,X)=E_{\mathrm{cell}}(x,Y)$.
\item\label{H5} (Isometry invariance) For every $X \in \mathcal{X}$, $x \in X$ we have
\begin{align*}
E_{\mathrm{cell}}(x,X)=E_{\mathrm{cell}}(Rx+\tau,RX+\tau) \quad \text{for all } R\in SO(d) \text{ and } \tau\in\R^d\,.
\end{align*}
\item\label{H6} (Coercivity) There exists $c>0$ such that for all $X \in \mathcal{X}$ and $x\in X$ we have
\begin{align*}
E_{\mathrm{cell}}(x,X)>0 \quad\implies \quad E_{\mathrm{cell}}(x,X)\geq c\,.
\end{align*}
\item\label{H7} (Local rigidity) Let $x,y\in X$ satisfy $\vert x-y\vert\leq r_{\mathrm{crys}}$ and $E_{\mathrm{cell}}(x,X)+E_{\mathrm{cell}}(y,X)=0$. Then $R_x=QR_y$ for some $Q \in G$  and $\tau_x =Q^{-1}\tau_y +R L z$ for some $z \in \mathbb{Z}^d$, where $L \in \mathrm{GL}(d,\mathbb{R})$ is given in \eqref{def:lattice} and  $(R_x,\tau_x), (R_y,\tau_y) \in SO(d)\times \mathbb{R}^d$ are  given in \ref{H3}. In particular, we can choose $(R_x,\tau_x) =(R_y,\tau_y)$.
To describe the remaining conditions, we define
\begin{align*}
\mathcal{N}(x) = \{y \in X \setminus \{x\} \mid |x-y|\leq r_{\mathrm{int}}\}\,.
\end{align*}
 Recall \eqref{eq:scaled-energy-local} for $\eps=1$.
\item\label{H8} (Unique interpolation) There exists $d_{\mathrm{unique}} \in \mathbb{N}$ such that if $\{x\}\cup\mc{N}(x)\subset\mc{L}(z)$ and $\#\mathcal{N}(x)\geq d_{\mathrm{unique}}$, then $z$ is unique. Furthermore, for all $ x\in X$ with $\#\mc{N}(x)<d_{\mathrm{unique}}$ we have
\begin{align*}
E_1(X\setminus\{x\})\leq E_1(X)\,.
\end{align*}
\item\label{H9} (Non-crystallized neighborhoods) For all $ x\in X$ such that for all $ z\in\mc{Z}$ it holds $\{x\}\cup\mc{N}(x)\nsubseteq\mc{L}(z)$ we have
\begin{align*}
E_1(X\setminus\{x\})\leq E_1(X)\,.
\end{align*} 
\item\label{H10} (Highly coordinated neighbors) For all $x,y\in X$ with $y\in\mc{N}(x)$, $\min(\#\mc{N}(x),\#\mc{N}(y))\geq d_{\mathrm{unique}}$, and $\mc{L}(z(x))\neq\mc{L}(z(y))$ we have
\begin{align*}
E_1(X\setminus\{x\})\leq E_1(X)\,.
\end{align*}
\end{enumerate}
We denote the constant $C > 0$ from \ref{H2} the upper bound of the cell energy, and by $c > 0$ (obtained from \ref{H6} its lower bound on non-optimal configurations. Introducing a scaling parameter $\varepsilon > 0$, 
we define the rescaled cell energy by
\begin{align}\label{eq:cell-energy-scaled}
E_{\mathrm{cell}}^\varepsilon(x, X)
:= E_{\mathrm{cell}}\!\left(\frac{x}{\varepsilon}, \frac{X}{\varepsilon}\right)\,.
\end{align}
The configurational energy in a Borel set $A \subset \mathbb{R}^d$ is
\begin{align}\label{eq:scaled-energy-local}
E_\varepsilon(X, A)
:= \varepsilon^{d-1} \sum_{x \in X \cap A} E_{\mathrm{cell}}^\varepsilon(x, X)\,,
\end{align}
and we write $E_\varepsilon(X) := E_\varepsilon(X, \mathbb{R}^d)$ for the total energy. We gather some elementary properties of the cell and configurational energies in Lemma~\ref{lem:elementary-properties}. Note that the upper and lower bounds in \ref{H2} and \ref{H6} are independent of $\varepsilon$, whereas the bounds in \ref{H3} and \ref{H4} transfer to $E_{\mathrm{cell}}^\varepsilon$ scale to $\varepsilon r_{\mathrm{crys}}, \varepsilon r_{\mathrm{int}}$ instead of $ r_{\mathrm{crys}}, r_{\mathrm{int}}$. 
We comment on the assumptions. The first assumption ensures a minimal distance at finite energy and excludes clustering of atoms for finite energy configurations. The second condition ensures non-negativity, which is necessary for our analysis, as well as the upper bound, which is required but usually given in applications. The third assumption ensures that the lattice $\mathcal{L}$ is a (local) ground state. The fourth assumption is a locality assumption, and we will call $r_{\mathrm{int}}$ the interaction radius. The fifth assumption accounts for a frame invariance in our discrete setting. The sixth assumption states that non-crystallized atoms must overcome a minimal energy barrier. Combined with the upper bound, this allows one to derive counting estimates for the number of non-crystallized atoms. The seventh assumption allows us to show that it is not possible to pass from one lattice to another without paying energy. We will call $x\in X$ crystallized if its cell energy vanishes. In Section~\ref{sec:reduction-two-lattices} we need to impose the additional assumptions \ref{H8}-\ref{H10}, which are crucial for relating the solid–solid transition energy to the solid–vacuum transition energy. This also allows us to derive a cell formula that is useful for matching upper and lower bounds for the $\Gamma$-limit. These conditions ensure that atoms with a {\it sparse} or non-lattice-like neighborhood are energetically not favorable and can be removed without increasing the energy. The last condition ensures that well-connected configurations are rigid or energetically inconvenient.

\begin{remark}
From now on, we will work with the stronger assumption $r_{\mathrm{crys}}\geq2R_V$ (Notice that in general $2R_V\geq S_V$). This greatly simplifies the proofs of the statements, as fewer technical arguments are needed. We will quickly comment on the changes needed to recover the results in the general case. First, one defines the set of interpolation points as 
\begin{align*}
X_{\varepsilon,\mathrm{int}}:=\{x\in X\mid E_{\mathrm{cell}}^\varepsilon(x,X)=0,\exists(R,\tau): X\cap\overline{B}_{4R_V\varepsilon}(x)=\mc{L}_\varepsilon(R,\tau,1)\cap\overline{B}_{4R_V\varepsilon}(x)\},
\end{align*}
which entails that the Voronoi cells of interpolation points in the configuration coincide with the Voronoi cells of the lattice.
Next, one shows that for $r\geq R_V$ one has the following: If $x\in X$ and $E_{\mathrm{cell}}^\varepsilon(z,X)=0$ for all $z\in X\cap\overline{B}_{4r\varepsilon}(x)$, then there is a unique $z\in\mc{Z}$ such that $X\cap\overline{B}_{(4r-4R_V)\eps}(x)=\mc{L}_\varepsilon(z)\cap\overline{B}_{(4r-4R_V)\eps}(x)$. This shows that in a region with vanishing energy in a smaller patch, the configuration coincides with the lattice. Lastly, one shows that if $x$ is crystallized but not an interpolation point, then there is $y\in\overline{B}_{8R_V\eps}(x)\cap X$ with $E_{\mathrm{cell}}^\varepsilon(y,X)>0$.
\end{remark}

\subsection{Identification of configurations with piecewise constant functions} We will now define functions that capture the local orientation of a grain.
\begin{Def}[Local orientation and associated lattice]\label{def:orientation} Let $E_{\mathrm{cell}}$ satisfy {\rm\ref{H1}}--{\rm\ref{H5}}.
 Let $X$ be a configuration with $E_\varepsilon(X)<+\infty$ and $x\in X$ such that $E_{\mathrm{cell}}^\varepsilon(x,X)=0$ and let $R_x\in\mathbb{A},\tau_x\in\mathbb{T}$ be the unique arguments that exist due to \rm{\ref{H3}}.  Then, its local orientation is defined as 
$$z(x):=(R_x,\tau_x,1)$$
and its associated lattice is given by $\mc{L}_\varepsilon(z(x))$.
We also write $X=\mc{L}_\varepsilon(z)$ on $A\subset\R^d$ if $X\cap A=\mc{L}_\varepsilon(z)\cap A$. 
\end{Def}
To each configuration $X \in \mathcal{X}$ with $E_\varepsilon(X)<+\infty$ we associate a piecewise constant function $u_\varepsilon^X \in PC(\R^d;\mc{Z})$. Namely, we define
\begin{align} \label{def:interpolation-points}
X_{\varepsilon,\mathrm{int}}:&=\{x\in X\mid  E_{\mathrm{cell}}^\varepsilon(x,X)=0\}\,.
\end{align}
 We define a piecewise constant orientation field associated to $X$ via
\begin{align}\label{def:interpolation}
u_{\varepsilon}^X(y):=\begin{cases}
z(x) & \text{if } y\in V_{\mc{L}_\varepsilon(z(x))}(x) \text{ and } x\in X_{\varepsilon,\mathrm{int}}\,,\\
\textbf{0} & \text{else.}
\end{cases}
\end{align} 
In the following, if no confusion arises, we will write $u_{\varepsilon}$ instead of $u_{\varepsilon}^X$.

The following lemma states that the functions defined in \eqref{def:interpolation} are well defined as $L^1$-functions.

\begin{lem}[Relation between Voronoi cells]\label{lem:Voronoi-relation}
Let $E_{\mathrm{cell}}$ satisfy {\rm \ref{H1}}-{\rm \ref{H5}}. Let $X \in \mathcal{X}$ be such that $E_\varepsilon(X)<+\infty$. Then the function defined in \eqref{def:interpolation} is well-defined.
\end{lem}
\begin{proof} We need to ensure that if there are two points $x_1,x_2 \in X_{\varepsilon,\mathrm{int}}$ such that $z(x_1) \neq z(x_2)$. Then,  $|V_{\mc{L}_\varepsilon(z(x_1))}(x_1)  \cap V_{\mc{L}_\varepsilon(z(x_2))}(x_2)|=0$. To this end, note that if $x \in X_{\varepsilon,\mathrm{int}}$, i.e., $E_\mathrm{cell}^\varepsilon(x,X)=0$. then $X \cap \overline{B}_{r_{\mathrm{crys}}\varepsilon} =  V_{\mc{L}_\varepsilon(z(x))}(x) \cap  \overline{B}_{r_{\mathrm{crys}}\varepsilon}$. As $r_{\mathrm{crys}} \geq 2R_V$ we in particular have that $V_X(x) =  V_{\mc{L}_\varepsilon(z(x))}(x)$, where 
\begin{align*}
V_X(x):=\{y\in\R^d\mid\vert x-y\vert\leq\vert z-y\vert\ \forall z\in X\}\,.
\end{align*}
Now, it is a well-known fact that for any discrete set we have that $|V_X(x) \cap V_X(y)| = 0$ for all $x,y \in X$ with $x \neq y$ as the Voronoi-cells intersect at at most $(d-1)$-dimensional hyperplanes (that are of Lebesgue measure zero).
\end{proof}
\noindent {\bf The state space.} For $A\subset \mathbb{R}^d$ open, we introduce the space of piecewise constant functions with values in $\mathcal{Z}$ by
\begin{align}\label{def:limitspace}
PC(A;\mc{Z}):=\{u\in SBV(A;\mc{Z})\mid\nabla u=0\,,\mc{L}^d(\{u\neq\textbf{0}\})<+\infty\,,\mc{H}^{d-1}(J_u)<+\infty\}\,.
\end{align}
Notice that we can write $u\in PC(A;\mc{Z})$ as a function of the form
\begin{align}\label{eq:representationu}
u=\sum_{j=1}^\infty z_j\chi_{G_j}
\end{align}
for pairwise distinct $z_j\in\mc{Z}\setminus\{\textbf{0}\}$ and pairwise disjoint $G_j\subset A$ satisfying $\mc{L}^d(\bigcup_{j=1}^\infty G_j)<+\infty$ and
\begin{align}\label{eq:finitesurfaceu}
\sum_{j=1}^\infty\mc{H}^{d-1}(\p^*\!G_j)<+\infty\,.
\end{align}
In this case, $\{G_j\}_{j\in\N}$ represents the grains of the polycrystal, whereas $\{z_j\}_{j\in\N}$ corresponds to the local orientation and translation of the lattice on each grain.



\begin{Def}(Convergence) \label{def:convergence}  Let $\{X_\varepsilon\}_\varepsilon$ be a sequence of configurations. We say that $X_\varepsilon\to u$ in $L^1_{\mathrm{loc}}(\R^d)$ as $\varepsilon \to 0$ if $u_\varepsilon\to u$ in $L^1_{\mathrm{loc}}(\R^d)$ as $\varepsilon \to 0$, where $u_\varepsilon=u_\varepsilon^{X_\varepsilon}$ is the function associated to $X_\varepsilon$ by \eqref{def:interpolation}. 
\end{Def}

\subsection{Cubes and boundary regions}\label{subsec:cubes} Let $\nu \in \mathbb{S}^{d-1}$. Then $R_\nu$ is the orthogonal matrix induced by the linear mapping 
\begin{align*}
x \mapsto \begin{cases}\displaystyle 2\frac{(\langle x, \nu \rangle +x_d)}{|\nu+e_d|^2}(\nu+e_d)-x &\text{if } \nu \in \mathbb{S}^{d-1} \setminus \{-e_d\}\,, \\
-x &\text{otherwise.}
\end{cases} 
\end{align*}
In this way $R_\nu e_d = \nu$ and the set $ \{R_\nu e_j \colon j=1,\ldots,d\}$ forms an orthonormal basis. We define the half-open unit cube with center $0$ and orientation $\nu$ by 
\begin{align*}
Q^\nu&:=R_\nu Q \quad \text{ where } \quad Q:=\left\{y\in\R^d\mid -\frac12\leq\langle y,e_i\rangle<\frac12\text{ for }i=1,\dots,d\right\}
\,.
\end{align*} 
and
\begin{align*}
Q^{\nu,\pm} =\{x \in Q^\nu \mid \pm\langle x,\nu\rangle \geq 0\}\,.
\end{align*}
 For $0<\lambda<\rho$ and $\varepsilon>0$ we define
\begin{align}\label{def:discrete-boundary}
\begin{split}
&\partial_{\lambda\eps} Q^\nu_\rho(x):=Q^\nu_{\rho +\lambda\eps}(x) \setminus  Q^\nu_{\rho- \lambda\eps}(x)\,,\\&  \partial^\pm_{\lambda\eps} Q^\nu_\rho(x) = \partial_{\lambda\eps} Q^\nu_\rho(x) \cap \{ z \in \mathbb{R}^d \colon \pm\langle z-x,\nu\rangle \geq r_{\mathrm{int}}\eps\}\,, \\& \partial^c_{\lambda\eps} Q^\nu_\rho(x) = \partial_{\lambda\eps} Q^\nu_\rho(x) \setminus (\partial_{\lambda\eps}^+ Q^\nu_\rho(x) \cup \partial_{\lambda\eps}^- Q^\nu_\rho(x))\,.
\end{split}
\end{align}
Similarly, we can define the half-open rectangle with length $l>0$ and height $h>0$ by
\begin{align*}
R_{l,h}&:=\{y\in\R^d\mid -\frac{h}{2}\leq\langle y,e_d\rangle<\frac{h}{2}, -\frac{l}{2}\leq\langle y,e_i\rangle<\frac{l}{2}\text{ for }i=1,\dots,d-1\}\,,\\
R^\nu_{l,h}&:=R_\nu R_{l,h}\,.
\end{align*}
Now, for $x\in\R^d$ and $\rho>0$, we can define the scaled and translated version of cubes, half cubes, and rectangles via
\begin{align*}
Q^\nu_{\rho}(x):= x+\rho Q^\nu\,,\quad
Q^{\nu,\pm}_{\rho}(x):= x+\rho Q^{\nu,\pm}\,,\quad
R^\nu_{l,h}(x):= x+R^\nu_{l,h}\,.
\end{align*}
For simplicity, if $\rho=1$ we write $Q^\nu(x)$ instead of $Q^\nu_\rho(x)$.\\
Throughout the paper, we will assume that our cell energy satisfies \ref{H1}-\ref{H10}, but also indicate which specific assumptions are needed for the separate lemmas and theorems. Notice that \ref{H8}-\ref{H10} are needed for the separation argument in our cell problem, see Section~\ref{sec:reduction-two-lattices}.

\subsection{Main Results}\label{subsec:Main-results}
We now state our main results. We start with a compactness result for sequences of configurations with bounded energy. Recall the definition of convergence for configurations given in Definition~\ref{def:convergence}.
\begin{Theorem}[Compactness]\label{thm:compactness}
Let $E_{\mathrm{cell}}$ satisfy {\rm\ref{H1}}--{\rm\ref{H7}}. Let $\{X_\varepsilon\}_\varepsilon$ be a sequence of configurations, such that 
$$\sup_\varepsilon E_\varepsilon(X_\varepsilon)<+\infty\,.$$ Then, there exists a subsequence $\{\varepsilon_k\}_{k\in\N}$ with $\varepsilon_k\to0$ as $k\to\infty$ and a function $u\in PC(\R^d;\mc{Z})$ such that $X_{\varepsilon_k}\to u$ in $L^1_{\mathrm{loc}}(\R^d)$ as $k\to\infty$.
\end{Theorem}

\begin{definition}\label{def:admissible-set}
Given $z^+,z^- \in \mathcal{Z}$, $\varepsilon, \rho >0$, $x \in \mathbb{R}^d$ we say that  $X \in \mathrm{Adm}_{\varepsilon,\lambda}^{(z^+,z^-)}(Q^\nu_\rho(x))$ if it satisfies the following:
\begin{itemize}
\item[(i)] $E_\varepsilon(X) <+\infty\,,$
\item[(ii)] $X = \mathcal{L}_\varepsilon(z^\pm)$  on $\partial_{\lambda \varepsilon}^\pm Q^\nu_\rho(x)$
\item[(iii)] $X = \emptyset$  on $\partial_{\lambda\varepsilon}^c Q^\nu_\rho(x)\,.$
\end{itemize}
\end{definition}

Recall Definition~\ref{def:admissible-set}. The following proposition introduces the energy density, which appears in our continuum limiting functional, see Figure~\ref{fig:cell-formula}.
\begin{prop}[Density]\label{prop:density}
Let $E_{\mathrm{cell}}$ satisfy {\rm \ref{H1}}--{\rm\ref{H10}}. For every $z^+,z^-\in\mc{Z},\nu\in\mathbb{S}^{d-1},x_0\in\R^d$, $\lambda > 6r_{\mathrm{int}}$, and $\rho>0$, there exists 
\begin{align}\label{eq:propdensity}
\varphi(z^+,z^-,\nu)=\lim_{\varepsilon\to0}\frac{1}{\rho^{d-1}}\inf\{ E_\varepsilon(&X,Q^\nu_\rho(x_0))\mid X \in \mathrm{Adm}_{\varepsilon,\lambda}^{(z^+,z^-)}(Q^\nu_\rho(x_0)) \}
\end{align}
 and is independent of $x_0$, $\lambda$, and $\rho$.
\end{prop}
\begin{figure}[htp]
\center
\begin{tikzpicture}[scale=0.6]

\tikzset{>={Latex[width=1mm,length=1mm]}};

\draw (6.5,5) node[above]{$Q^\nu_{\rho}$};
\draw(-5.5,2)--(-6.5,3) node[above]{$\partial^+_{\lambda\varepsilon} Q^\nu_{\rho}$};
\draw(-5.5,-2)--(-6.5,-3) node[below]{$\partial^-_{\lambda\varepsilon} Q^\nu_{\rho}$};
\draw(-5,-5) rectangle(5,5);

\draw[<->](-5,6)-- node[above]{$\rho$}(5,6);
\draw[black!20,pattern=north west lines,pattern color=black!20](-5.5,1)--(-5.5,5.5)--(5.5,5.5)--(5.5,1)--(4.5,1)--(4.5,4.5)--(-4.5,4.5)--(-4.5,1)--(-5.5,1);
\draw[black!20,pattern=north west lines,pattern color=black!20](-5.5,-1)--(-5.5,-5.5)--(5.5,-5.5)--(5.5,-1)--(4.5,-1)--(4.5,-4.5)--(-4.5,-4.5)--(-4.5,-1)--(-5.5,-1);

\begin{scope}[scale=0.4,shift={(-13,2)}]
\foreach\i in {0,...,6}{
\foreach\k in {0,...,8}{
\foreach\j in {0,...,5}{
\fill[black](0+\k*3,0+\i*1.732)++(60*\j:1) circle(.1);
\fill[black](1.5+\k*3,0.866+\i*1.732)++(60*\j:1) circle(.1);
}
}
}
\end{scope}

\draw[fill=white,white](6,1)--(4.5,1)--(4.5,-1)--(6,-1)--(6,1);

\draw[fill=white,white](-6,1)--(-4.5,1)--(-4.5,-1)--(-6,-1)--(-6,1);
\draw[->](5.5,0)--(5.5,0.8) node[right]{$\nu$};
\draw[dashed](-6,0)--(6,0);

\fill(0,0) circle(.04);
\fill(3,0) circle(.04);
\fill(2,0.1) circle(.04);
\fill(1,-0.1) circle(.04);
\fill(-4,0.1) circle(.04);
\fill(-2,-0.1) circle(.04);
\fill(-1,-0.2) circle(.04);
\fill(-1,0.2) circle(.04);
\clip(-5.5,-0.5)--(-5.5,-6)--(5.5,-6)--(5.5,-0.5)--(-5.5,-0.5);
\begin{scope}[scale=0.4,shift={(-16,-4)}]
\foreach\i in {0,...,11}{
\foreach\k in {0,...,10}{
\foreach\j in {0,...,5}{
\fill[black,rotate=20](0+\k*3,0-\i*1.732)++(60*\j:1) circle(.1);
\fill[black,rotate=20](1.5+\k*3,0.866-\i*1.732)++(60*\j:1) circle(.1);
}
}
}
\end{scope}

\draw[fill=white,white](6,1)--(4.5,1)--(4.5,-1)--(6,-1)--(6,1);

\draw[fill=white,white](-6,1)--(-4.5,1)--(-4.5,-1)--(-6,-1)--(-6,1);
\end{tikzpicture}
\caption{Schematic Illustration of a competitor $X$ for the cell problem on $Q^\nu_\rho$ in the definition of $\varphi$. Also, this illustrates a configuration $X \in \mathrm{Adm}_{\varepsilon,\lambda}^{(z^+,z^-)}(Q^\nu_\rho)$. Note that the scaling factor between the lattice spacing and the boundary thickness does not match the definition.} 
\label{fig:cell-formula}
\end{figure}
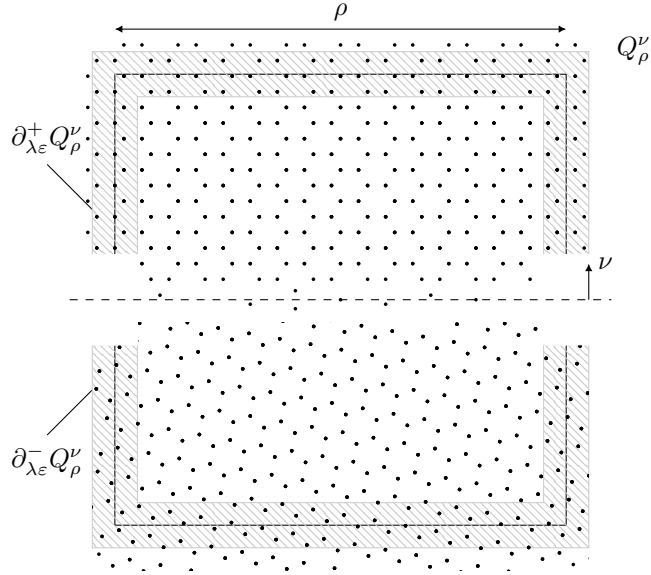
The limiting functional $E: PC(\R^d;\mc{Z})\to[0,+\infty)$ is defined as
\begin{align}\label{def:limiting energy}
E(u):=\int_{J_u}\varphi(u^+,u^-,\nu_u)\,\mathrm{d}\mc{H}^{d-1}\,.
\end{align}
In view of \eqref{def:limitspace}, functions in $PC(\R^d;\mc{Z})$ lie in $SBV(\R^d;\mc{Z})$ and the quantities $u^+,u^-$ and $\nu_u$ are well defined. The following Theorem shows that $E$ can be interpreted as the continuum limit energy of our atomistic energies $E_\varepsilon$ in the sense of $\Gamma$-convergence.
\begin{Theorem}[$\Gamma$-convergence]\label{thm:Gamma-convergence}
Let $E_{\mathrm{cell}}$ satisfy {\rm \ref{H1}}--{\rm\ref{H10}}. Then, the following holds
\begin{enumerate}
\item[{\rm (i)}] \label{thmitem:Gamma-convergence1}$(\Gamma$-$\liminf$ inequality$)$ For each $u\in PC(\R^d;\mc{Z})$ and each sequence $\{X_\varepsilon\}_\varepsilon$ with $X_\varepsilon\to u$ in $L^1_{\mathrm{loc}}(\R^d)$ as $\varepsilon\to0$ it holds
$$\liminf_{\varepsilon\to0}E_\varepsilon(X_\varepsilon)\geq E(u)\,.$$
\item[{\rm (ii)}] \label{thmitem:Gamma-convergence2}$(\Gamma$-$\limsup$ inequality$)$  For each $u\in PC(\R^d;\mc{Z})$ there is a sequence of configurations $\{X_\varepsilon\}_\varepsilon$ such that $X_\varepsilon\to u$ in $L^1_{\mathrm{loc}}(\R^d)$ as $\varepsilon\to0$ and
$$\lim_{\varepsilon\to0}E_\varepsilon(X_\varepsilon)=E(u)\,.$$
\end{enumerate}
\end{Theorem}

\begin{rem}[Extension to $L^1$]\label{rem:extension}
Notice that by defining $E_\varepsilon:L^1(\R^d;\mc{Z})\to [0,+\infty]$ by
\begin{align*}
E_\varepsilon(u):=\begin{cases}
E_\varepsilon(X)&\text{if there exists }X\text{ such that }u=u_\varepsilon^X\,,\\
+\infty&\text{otherwise,}
\end{cases}
\end{align*}
one can view $E=\Gamma(L^1_{\mathrm{loc}})$-$\lim_{\varepsilon\to0}E_\varepsilon$ by Theorem~\ref{thm:Gamma-convergence}.
\end{rem}
To end the section, we gather some properties of the limiting energy density. We would like to mention that due to the rigidity of our cell energies, it is not favorable to introduce interpolating boundary layers to lower the energy. Instead, the energy for a solid-solid transition is the sum of the energies of two solid-vacuum transitions. 

\begin{definition}\label{def:phivac} We define the solid-vacuum density by $\varphi_{\mathrm{vac}} \colon \mathcal{Z} \times \mathbb{R}^d \to [0,+\infty)$ for $z \in \mathcal{Z}$ and $\nu \in \mathbb{R}^d$ by
\begin{align}\label{def:phivac}
\varphi_{\mathrm{vac}}(z,\nu) = \varphi(z,{\bf 0},\nu)\,.
\end{align}
\end{definition}

\begin{Theorem}[Properties of $\varphi$]\label{thm:properties-of-phi}
Let $E_{\mathrm{cell}}$ satisfy {\rm \ref{H1}}--{\rm \ref{H10}}. Let $\varphi$ be the density in Proposition~\ref{prop:density} extended to a function defined on $\mc{Z}\times\mc{Z}\times\R^d$, which is positively $1$-homogeneous in the third variable. Let $\varphi_{\mathrm{vac}}$ be defined as above. Then $\varphi$ and  $\varphi_{\mathrm{vac}}$ satisfy the following properties:
\begin{enumerate}
\item[{\rm (i)}]\label{item:propthm1}(Solid-vacuum energy) It holds $\varphi(z,{\bf 0},\nu)=\varphi({\bf 0},z,-\nu)$ for all $z\in\mc{Z}\setminus\{{\bf 0}\}$ and $\nu\in\mathbb{S}^{d-1}$.
\item[{\rm (ii)}]\label{item:propthm2}(Solid-solid energy) For all $z^+,z^-\in\mc{Z}\setminus\{{\bf 0}\}, z^+\neq z^-$ and $\nu\in\mathbb{S}^{d-1}$ it holds
$$\varphi(z^+,z^-,\nu)=\varphi_{\mathrm{vac}}(z^+,\nu)+\varphi_{\mathrm{vac}}(z^-,-\nu)\,.$$
\item[{\rm (iii)}]\label{item:propthm3}(Convexity) The mapping $\nu\mapsto\varphi_{\mathrm{vac}}(z,\nu)$ is convex for all $z\in\mc{Z}$.
\item[{\rm (iv)}]\label{item:propthm4}(Boundedness \& Continuity) There exists $C>0$ such that 
\begin{align*}
\varphi_{\mathrm{vac}}(z,\nu) \leq C|\nu| \quad \text{for all } \nu \in \mathbb{R}^d\,.
\end{align*}
In particular, $\nu \mapsto \varphi_{\mathrm{vac}}(z,\nu)$ is Lipschitz-continuous.
\item[{\rm (v)}]\label{item:propthm5}(Translational invariance) For all $z=(R,\tau,1) \in \mathcal{Z},\nu\in\mathbb{S}^{d-1}$ there holds
$$\varphi_{\mathrm{vac}}((R,\tau,1),\nu)=\varphi_{\mathrm{vac}}((R,0,1),\nu)\,.$$
\item[{\rm (vi)}]\label{item:propthm6}(Rotational invariance) For all $z=(R,\tau, 1),\nu\in\mathbb{S}^{d-1}$ and $O\in SO(d)$ there holds
$$\varphi_{\mathrm{vac}}((OR,0, 1),O\nu)=\varphi_{\mathrm{vac}}((R,0,1),\nu)\,.$$
\end{enumerate}
\end{Theorem}

\section{Proof of the Main Results}\setcounter{thm}{1}\label{sec:3}

\subsection{Proof of Compactness}
\begin{proof}[Proof of Theorem \ref{thm:compactness}]
Let $\{X_\varepsilon\}_\varepsilon$ be given as in the statement and let $\{u_\varepsilon\}_\varepsilon$ denote their corresponding PC-functions given by \eqref{def:interpolation}. We note that, due to our assumptions on $\mc{Z}$ (see the passage below \eqref{eq:rotated-Voronoi}), it can be embedded into $\R^N$ and is closed and bounded via this embedding. This implies that the $L^\infty$-norm of $\{u_\varepsilon\}_\varepsilon$ can be uniformly bounded by some finite constant. For $r>0$, let $B_r$ be the ball of radius $r$ centered at the origin, then by the coercivity result Lemma~\ref{lem:coercivity} and the assumption $\sup_{\varepsilon>0}E_\varepsilon(X_\varepsilon)<+\infty$ we can uniformly bound $\mc{H}^{d-1}(J_{u_\varepsilon}\cap B_r)$. Thus, for each $r\in\N$ we can apply the compactness theorem for piecewise constant functions, see \cite[Theorem 4.25]{Ambrosio-Fusco-Pallara:2000}. Therefore, by a standard diagonal argument, we obtain that (up to subsequences) $u_\varepsilon \to u$ in $L^1_{\mathrm{loc}}(\mathbb{R}^d;\mathbb{Z})$. Furthermore, due to the lower semicontinuity of $\mc{H}^{d-1}(J_{u} \cap B_n)$ with respect to $L^1$-convergence of partitions, see \cite[Theorem 4.7]{Ambrosio-Fusco-Pallara:2000}, we obtain that $\mc{H}^{d-1}(J_u \cap B_n)\leq C$ for some $C>0$ independent of $n$. In order to conclude that $u\in PC(\R^d;\mc{Z})$ we need to verify that $\mc{L}^d(\{u\neq\textbf{0}\})<+\infty$. Using Lemma~\ref{lem:coercivity} with $A=\R^d$, the isoperimetric inequality and lower semicontinuity of $\mc{L}^d(\{u\neq\textbf{0}\})$ with respect to strong $L^1_{loc}$-convergence, we see that
\begin{align*}
(\mc{L}^d(\{u\neq\textbf{0}\}))^\frac{d-1}{d}&\leq\liminf_{k\to\infty}(\mc{L}^d(\{u_{\varepsilon_k}\neq\textbf{0}\}))^\frac{d-1}{d}\leq\liminf_{k\to\infty}C\mc{H}^{d-1}(\p^*\{u_{\varepsilon_k}\neq\textbf{0}\})\\
&\leq\liminf_{k\to\infty}C\mc{H}^{d-1}(J_{u_{\varepsilon_k}})\leq\liminf_{k\to\infty}CE_{\varepsilon_k}(X_{\varepsilon_k})<+\infty\,.
\end{align*}
Thus, we have $u\in PC(\R^d;\mc{Z})$, which concludes the proof.
\end{proof}

\subsection{Lower Bound}\label{subsec:lowerbound}

In this subsection, we will prove Theorem~\ref{thm:Gamma-convergence}{\rm (i)}. For the proof, it is instrumental to use a different cell formula. In contrast to imposing boundary conditions as in \eqref{eq:propdensity} we require $L^1$-convergence to the function $u^\nu_{z^+,z^-}\in PC_{\mathrm{loc}}(\R^2;\mc{Z})$ defined as
\begin{align}\label{eq:purejump}
u^\nu_{z^+,z^-}(x):=\begin{cases}z^+ & \text{if }\langle x,\nu\rangle\geq0\,,\\ z^- &\text{if }\langle x,\nu\rangle<0\,,\end{cases}
\end{align}
for $x\in\R^d,z^\pm\in\mc{Z}$ and $\nu\in\mathbb{S}^{d-1}$. More precisely, for $z^+,z^- \in \mathcal{Z}, \nu \in \mathbb{S}^{d-1}$ we introduce
\begin{align}\label{eq:cell-formulaL1}
\begin{split}
\psi(z^+,z^-,\nu):=\inf\Big\{\liminf_{\varepsilon\to0}&E_\varepsilon(X_\varepsilon, Q^\nu(y_\varepsilon))\mid y_\varepsilon\in\R^d\,,\\
&\lim_{\varepsilon\to0}\int_{Q^\nu}\vert u_\varepsilon(x+y_\varepsilon)-u^\nu_{z^+,z^-}(x)\vert\,\mathrm{d}x=0\Big\}\,,
\end{split}
\end{align}
 where $u_\varepsilon$ is the function associated to $X_\varepsilon$ as defined in \eqref{def:interpolation}. The density $\psi$ is related to $\varphi$ in the following way:
\begin{prop}[Relation of $\psi$ and $\varphi$]
\label{prop:relationpsiphi}
For all $z^+,z^-\in\mc{Z}$ and $\nu\in\mathbb{S}^{d-1}$ there holds
\begin{align*}
\psi(z^+,z^-,\nu)\geq\varphi(z^+,z^-,\nu)\,.
\end{align*}
\end{prop}
We postpone the proof of Proposition~\ref{prop:relationpsiphi} to Sections~\ref{sec:L1-to-converging-data}--\ref{sec:converging-fixed-bdryvalues}.

\begin{proof}[Proof of Theorem~\ref{thm:Gamma-convergence}{\rm (i)}]
Let $\{X_\varepsilon\}_\varepsilon$ be a sequence with $X_\varepsilon\to u$ in $L^1_{\mathrm{loc}}(\mathbb{R}^d)$. Clearly, it suffices to treat the case
\begin{align}\label{ineq:uniformbound-compactness}
\sup_{\varepsilon>0}E_\varepsilon(X_\varepsilon)<+\infty\,.
\end{align}
We proceed in two steps. In the first step, we will identify a limiting measure associated to the energies of the discrete configurations. In the second step, we proceed by a blow-up procedure for the jump part of this measure.\\ 
\noindent {\bf Step 1:} {\it Identification of a limiting measure.} Consider the family of positive measures $\{\mu_\varepsilon\}_\varepsilon$ given as
\begin{align*}
\mu_\varepsilon:=\varepsilon^{d-1}\sum_{x\in X_\varepsilon}E_{\mathrm{cell}}^\varepsilon(x,X_\varepsilon)\delta_x\,.
\end{align*}
By \eqref{eq:scaled-energy-local}, for all $ A\subset \mathbb{R}^d$ there holds
\begin{align}\label{eq:measure-local-energy}
\vert\mu_\varepsilon\vert(A)=\mu_\varepsilon(A)=E_\varepsilon(X_\varepsilon,A)\,.
\end{align}
Thus, by \eqref{ineq:uniformbound-compactness} we have $\sup_{\varepsilon>0}\vert\mu_\varepsilon\vert(\R^d)<+\infty$ and as $\R^d$ is locally compact, there exists a non-negative finite Radon measure $\mu$ such that up to a (non relabeled) subsequence we have
\begin{align}\label{eq:convergence-mueps}
\mu_\varepsilon\overset{*}{\rightharpoonup}\mu\,.
\end{align}
Now, using the Radon-Nykodym Theorem and Lebesgue Decomposition Theorem for $\mc{H}^{d-1}|_{J_u}$, we get a decomposition into two non-negative mutually singular measures
\begin{align*}
\mu=\xi\mc{H}^{d-1}|_{J_u}+\mu_s\,.
\end{align*}
Using a blow-up procedure, we will show that
\begin{align}\label{ineq:xidensity}
\xi(x_0)\geq\psi(z^+,z^-,\nu)\quad\text{for }\mc{H}^{d-1}\text{-almost every }x_0\in J_u\,,
\end{align}
where $z^+,z^-$ denote the one-sided limits of $u$ at $x_0$ and $\nu$ denotes the corresponding normal (here we omit the explicit dependence on $u$ for notational convenience). We postpone the proof of \eqref{ineq:xidensity} to Step 2 and first argue how we can conclude the proof once \eqref{ineq:xidensity} is proven. We use \eqref{def:limiting energy},\eqref{eq:measure-local-energy}--\eqref{ineq:xidensity},  and Proposition~\ref{prop:relationpsiphi} to show that
\begin{align*}
\liminf_{\varepsilon\to0}E_\varepsilon(X_\varepsilon)=\liminf_{\varepsilon\to0}\mu_\varepsilon(\R^d)\geq\mu(\R^d)\geq\int_{J_u}\xi\,\mathrm{d}\mc{H}^{d-1}
\geq\int_{J_u}\varphi(z^+,z^-,\nu)\,\mathrm{d}\mc{H}^{d-1}=E(u)\,.
\end{align*}
\noindent {\bf Step 2:} {\it Blow-up procedure.}  It remains to prove \eqref{ineq:xidensity}. Using properties of SBV-functions and Radon measures (see, for example, \cite[Theorem 2.63, Theorem 3.78, and Remark 3.79]{Ambrosio-Fusco-Pallara:2000}), we know that for $\mc{H}^{d-1}$- almost every $x_0\in J_u$ it holds that
\begin{enumerate}
\item[(a)] 
$\displaystyle
\lim_{\rho\to0}\frac{1}{\rho^d}\int_{Q_\rho^\nu(x_0)}\vert u(x)-u^\nu_{z^+,z^-}(x-x_0)\vert\,\mathrm{d}x=0\,,
$
\item[(b)]
$\displaystyle
\lim_{\rho\to0}\frac{1}{\rho^{d-1}}\mc{H}^{d-1}(J_u\cap Q_\rho^\nu(x_0))=1\,,
$
\item[(c)]
$ \displaystyle
\xi(x_0)=\lim_{\rho\to0}\frac{\mu(Q^\nu_\rho(x_0))}{\mc{H}^{d-1}(J_u\cap Q_\rho^\nu(x_0))}\,.
$
\end{enumerate}
Here, $u^\nu_{z^+,z^-}$ is defined as in \eqref{eq:purejump}. To show \eqref{ineq:xidensity} it suffices to prove it for all $x_0\in J_u$ such that (a)-(c) hold. To do this, let $x_0\in J_u$ be such a point. As $\mu$ is a locally finite measure we can fix a sequence $\rho_n\to0$ as $n\to\infty$ such that $\vert\mu\vert(\p Q^\nu_{\rho_n}(x_0))=0$ for all $n\in\N$. Then as (b)-(c) hold, using the Portmanteau Theorem as well as \eqref{eq:measure-local-energy}, and \eqref{eq:convergence-mueps}, we get:
\begin{align*}
\xi(x_0)&=\lim_{\rho\to0}\frac{\mu(Q_\rho^\nu(x_0))}{\mc{H}^{d-1}(J_u\cap Q_\rho^\nu(x_0))}=\lim_{\rho\to0}\frac{\mu(Q_\rho^\nu(x_0))}{\rho^{d-1}}
=\lim_{n\to\infty}\frac{1}{\rho_n^{d-1}}\lim_{\varepsilon\to0}\mu_\varepsilon(Q_{\rho_n}^\nu(x_0))\\&=\lim_{n\to\infty}\frac{1}{\rho_n^{d-1}}\lim_{\varepsilon\to0}E_\varepsilon(X_\varepsilon,Q^\nu_{\rho_n}(x_0))\,.
\end{align*}
We introduce the configurations $X_\varepsilon^n:=\rho_n^{-1}X_\varepsilon$ and apply Lemma \ref{lem:elementary-properties}(vii) (for $\lambda=\frac{1}{\rho_n}$) to obtain
\begin{align}\label{eq:xi-energy-cube}
\xi(x_0)=\lim_{n\to\infty}\lim_{\varepsilon\to0}E_{\frac{\varepsilon}{\rho_n}}(X^n_\varepsilon,Q^\nu(\rho_n^{-1}x_0))\,.
\end{align}
Recall that $X_\varepsilon\to u$ in $L^1_{\mathrm{loc}}(\R^d)$ implies, by Definition~\ref{def:convergence}, that $u_\varepsilon\to u$ in $L^1_{\mathrm{loc}}(\R^d)$. Let $u_\varepsilon^n$ denote the function corresponding to $X_\varepsilon^n$. Then Lemma~\ref{lem:scaling} implies that $u^n_\varepsilon(x)=u_\varepsilon(\rho_nx)$ for all $x\in\R^d$. In particular it also yields $u^n_\varepsilon\to u^n$ on $Q^\nu(\rho^{-1}_nx_0)$, where $u^n(x):=u(\rho_nx)$ for $x\in\R^d$. Using (a), the change of variables $y=\rho_nx+x_0$ and $u^n(x+\rho_n^{-1}x_0)=u(x_0+\rho_nx)$ as well as $u^\nu_{z^+,z^-}(x)=u^\nu_{z^+,z^-}(\rho_nx)$ for $x\in\R^d$, we obtain
\begin{align*}
\lim_{n\to\infty}\int_{Q^\nu}\vert u^n(x+\rho_n^{-1}x_0)-u^\nu_{z^+,z^-}(x)\vert\,\mathrm{d}x&=\lim_{n\to\infty}\int_{Q^\nu}\vert u(x_0+\rho_nx)-u^\nu_{z^+,z^-}(\rho_nx)\vert\,\mathrm{d}x
\\&=\lim_{n\to\infty}\frac{1}{\rho_n^d}\int_{Q^\nu_{\rho_n}(x_0)}\vert u(y)-u^\nu_{z^+,z^-}(y-x_0)\vert\,\mathrm{d}y=0\,.
\end{align*}
Now, by \eqref{eq:xi-energy-cube} and $u^n_\varepsilon\to u^n$ on $Q^\nu(\rho_n^{-1}x_0)$ as $\varepsilon\to0$, we use a diagonal argument to find a null sequence $\{\varepsilon(n)\}_n$ such that for $X^n:=X^n_{\varepsilon(n)}$ and $u^n:=u^n_{\varepsilon(n)}$ it holds
\begin{align}\label{eq:convergencexiEnergy2}
\xi(x_0)=\lim_{n\to\infty}E_{\varepsilon_n}(X^n,Q^\nu(y^n))\,,
\end{align}
and
\begin{align*}
\lim_{n\to\infty}\int_{Q^\nu}\vert u^n(x+y^n)-u^\nu_{z^+,z^-}(x)\vert\,\mathrm{d}x=0\,,
\end{align*}
where $\varepsilon_n:=\frac{\varepsilon(n)}{\rho_n}$ and $y^n=\rho_n^{-1}x_0$. As the above constructed sequence is admissible for the definition of $\psi$, see \eqref{eq:cell-formulaL1}, \eqref{eq:convergencexiEnergy2} implies $\xi(x_0)\geq\psi(z^+,z^-,\nu)$. This shows \eqref{ineq:xidensity} and concludes the proof.
\end{proof}
\subsection{Upper bound}\label{subsec:upperbpund}
This subsection is devoted to the proof of Theorem \ref{thm:Gamma-convergence}{\rm (ii)}. The following density result will be crucial.
\begin{lem}
\label{lem:density}
Let $u\in PC(\R^d;\mc{Z})$. Then there exists a sequence $(u_n)_n\subset PC(\R^d;\mc{Z})$ with $u_n\to u$ in $L^1(\R^d)$ and $\limsup_{n\to\infty}E(u_n)\leq E(u)$ such that $u_n$ attains only finitely many values and has polyhedral jump set, that is, $J_{u_n}$ is given by the union of $(d-1)$-dimensional simplices up to a $\mc{H}^{d-1}$-negligible set.
\end{lem}

\begin{proof} The proof follows exactly as the proof of \cite[Lemma~3.6]{FriedrichKreutzSchmidt} where we observe that the density result \cite[Theorem~2.1 and Corollary~2.4]{Braides-Conti-Garroni:2017} holds true for any dimension $d\geq 2$. Here, observe that only continuity of the map $\nu \mapsto \varphi(z_1,z_2,\nu)$ for all $z_1,z_2 \in \mathcal{Z}$ is needed.
\end{proof}

\begin{figure}
\begin{tikzpicture}[scale=0.7]

\tikzset{>={Latex[width=1mm,length=1mm]}};

\tkzDefPoint(0,0){A}
\tkzDefPoint(10,0){B}
\tkzDefSquare(A,B)
\tkzGetPoints{C}{D}
\tkzDrawPolygon(A,B,C,D)
\draw[dashed](-0.5,5)--(11,5) node[right]{$H^\nu$};
\draw[->](11,5)--(11,4) node[right]{$\nu$};
\draw[fill=black!20] (0,5.05) rectangle ++(10,4.95);
\draw[fill=black!50] (0,0) rectangle ++(10,4.95);
\draw(1.55,5) circle(.2);
\draw(-3,5) circle(2);
\draw(-3,5)++(60:2)--(1.55,5.2);
\draw(-3,5)++(-60:2)--(1.55,4.8);
\draw(4.3,2.5) node[right]{$\{u=z\}$};
\draw(4.3,7.5) node[right]{$\{u=0\}$};
\draw[{<[scale=0.5]}-{>[scale=0.5]}](2.1,4.82)--node[below,scale=0.5]{$\rho$}(2.31,4.82) ;
\begin{scope}
\clip(-3,5) circle(2);
\draw(-3,5) node{$X^x_{\varepsilon,\rho}$};
\fill[black!20] (-7,5.7) rectangle (10,8);
\fill[black!20] (-4.3,6) rectangle (-3.7,5.2);
\fill[black!20] (-2.3,6) rectangle (-1.7,5.2);
\fill[black!50] (-7,4.3) rectangle(10,0);
\fill[black!50] (-4.3,4.3) rectangle(-3.7,4.8);
\fill[black!50] (-2.3,4.3) rectangle(-1.7,4.8);
\fill[pattern=north west lines] (-4.3,4.8) rectangle (-3.7,5.2);
\fill[pattern=north west lines] (-2.3,4.8) rectangle (-1.7,5.2);
\foreach \j in {0,...,2}{
\tkzDefPoint(-4+\j*2,4){A}
\tkzDefPoint(-4+\j*2,6){B}
\tkzDefSquare(A,B)
\tkzGetPoints{C}{D}
\tkzDrawPolygon(A,B,C,D)
}
\begin{scope}
\clip(-7,4.8) rectangle(5,0);
\begin{scope}[shift={(-6,4.25)},scale=0.2]
\foreach\i in {0,...,10}{
\foreach\k in {0,...,20}{
\foreach\j in {0,...,5}{
\fill[black,rotate=30](0+\k*3,0-\i*1.732)++(60*\j:1) circle(.1);
\fill[black,rotate=30](1.5+\k*3,0.866-\i*1.732)++(60*\j:1) circle(.1);
}
}
}
\end{scope}
\end{scope}
\end{scope}
\foreach \j in {1,...,46}{
\tkzDefPoint(0+\j*0.21,4.9){E}
\tkzDefPoint(0.21+\j*0.21,4.9){F}
\tkzDefSquare(E,F)
\tkzGetPoints{G}{H}
\tkzDrawPolygon(E,F,G,H)
}
\draw(-3,4)--(-3,2.5) node[below,scale=0.7]{$Q^{\nu}_\rho(x)$};
\end{tikzpicture}
\caption{This is a schematic picture of a transition from two constant values considered in the construction of the upper bound for $d=2$. The dark and light grey regions are regions where the configuration coincides with the lattice and vacuum, respectively.}
\label{fig:\thefig}\stepcounter{fig}
\end{figure}
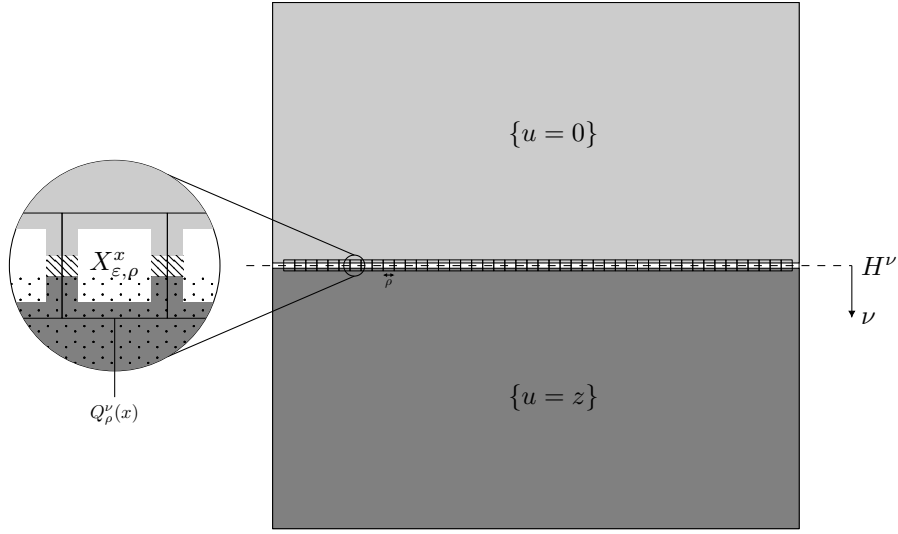

We are now in a position to prove Theorem~\ref{thm:Gamma-convergence}{\rm (ii)}. 

\begin{proof}[Proof of Theorem~\ref{thm:Gamma-convergence}{\rm (ii)}] We split the proof into several steps.  \\
\noindent {\bf Step 1:} {\it Energy representation via $\varphi_{\mathrm{vac}}$.} Let $u \in \mathrm{PC}(\mathbb{R}^d;\mathcal{Z})$. We then have
\begin{align}\label{eq:E-phivac}
E(u) = \sum_{i=1}^\infty \int_{\partial^* G_i} \varphi_{\mathrm{vac}}(z_i, \nu_i)\,\mathrm{d}\mathcal{H}^{d-1}\,.
\end{align}
In fact, due to Theorem~\ref{thm:properties-of-phi}(ii) and the fact that $\nu_i =-\nu_j$ for $\mathcal{H}^{d-1}$-a.e.~$x\in \partial^*G_i\cap \partial^*G_j$ (see \cite[Definition 3.67 and Theorem 4.17]{Ambrosio-Fusco-Pallara:2000}), we have
\begin{align*}
E(u) &= \int_{J_u} \varphi(z^+,z^-,\nu)\,\mathrm{d}\mathcal{H}^{d-1} = \frac{1}{2}\sum_{i=0}^\infty\sum_{j\neq i} \int_{\partial^*G_i\cap \partial^*G_j} \varphi(z_i,z_j,\nu_i)\,\mathrm{d}\mathcal{H}^{d-1} \\&=\frac{1}{2}\sum_{i=0}^\infty\sum_{j\neq i} \int_{\partial^*G_i\cap \partial^*G_j} \varphi_{\mathrm{vac}}(z_i,\nu_i)+ \varphi_{\mathrm{vac}}(z_j,-\nu_i)\,\mathrm{d}\mathcal{H}^{d-1}\\&=\frac{1}{2}\sum_{i=0}^\infty\sum_{j\neq i} \int_{\partial^*G_i\cap \partial^*G_j} \varphi_{\mathrm{vac}}(z_i,\nu_i)\,\mathrm{d}\mathcal{H}^{d-1} + \frac{1}{2}\sum_{i=0}^\infty\sum_{j\neq i} \int_{\partial^*G_i\cap \partial^*G_j} \varphi_{\mathrm{vac}}(z_j,-\nu_i)\,\mathrm{d}\mathcal{H}^{d-1} \\&= \sum_{i=0}^\infty\sum_{j\neq i} \int_{\partial^*G_i\cap \partial^*G_j} \varphi_{\mathrm{vac}}(z_i,\nu_i)\,\mathrm{d}\mathcal{H}^{d-1}  = \sum_{i=1}^\infty \int_{\partial^*G_i} \varphi_{\mathrm{vac}}(z_i,\nu_i)\,\mathrm{d}\mathcal{H}^{d-1} \,,
\end{align*}
where in the last equality we used that $\varphi_\mathrm{vac}(\textbf{0},\nu)=0$ for all $\nu \in \mathbb{S}^{d-1}$ and $\mathcal{H}^{d-1}(\partial^* G_i \Delta\cup_{j\neq i} (\partial^* G_i  \cap \partial^* G_j))=0$ (see, \cite[Theorem 4.17]{Ambrosio-Fusco-Pallara:2000}). This shows \eqref{eq:E-phivac} and concludes Step~1. \\ 
\noindent {\bf Step 2:} {\it Strict approximation of each of the grains.} By a general density argument in the theory of $\Gamma$-convergence and Lemma~\ref{lem:density}, it suffices to construct a recovery sequence only for $u\in PC(\R^d;\mc{Z})$ with bounded support, attaining finitely many values and having polyhedral jump set, i.e., 
\begin{align}\label{eq:u-finite-card}
u=\sum_{j=1}^N z_j \chi_{G_j}\,,
\end{align}
where $G_j \subset \mathbb{R}^d$ are pairwise disjoint,  bounded, and open sets of finite perimeter with polyhedral boundary.  In particular, there holds $\mathcal{H}^{d-1}(\partial G_j) = \mathcal{H}^{d-1} (\partial^* G_j)$ for all $j=1,\ldots,N$. We claim that we  can furthermore assume that 
\begin{align}\label{ineq:distance-Grains-pos}
\min_{i\neq j}\mathrm{dist}(G_i,G_j) >0\,.
\end{align}
To see this, by \eqref{eq:E-phivac} and \eqref{eq:u-finite-card}, it suffices to approximate each grain $G_i$ strictly from the interior as then, due to Reshetnyak's Continuity Theorem, see \cite[Theorem~2.39]{Ambrosio-Fusco-Pallara:2000} and recall that $\nu \mapsto \varphi_{\mathrm{vac}}(z,\nu)$ is continuous (Remark~\ref{rem:continuity}) and bounded (Theorem~\ref{thm:properties-of-phi}(iv)), we can assume \eqref{ineq:distance-Grains-pos}. The strict interior approximation result is a consequence of \cite[Theorem~1.1]{Schmidt}.  \\
\noindent {\bf Step 3:} {\it Discrete-to-continuum approximation.} As the jump set is polyhedral, we only perform a local construction on each $(d-1)$-dimensional face of the jump set. To this end, due to \eqref{ineq:distance-Grains-pos}, up to translation, we may assume that for some $D \subset \mathbb{R}^d$ open, bounded and smooth set and for some $z\in \mathcal{Z}$ we have that $u(x) = u^\nu_{z,\textbf{0}}(x)$ for $x \in D$. It suffices to construct $\{X_\varepsilon\}_\varepsilon$ such that $X_\varepsilon \to u$ in $\mathrm{L}^1(D)$ and
\begin{align}\label{eq:goal-Step3}
\lim_{\varepsilon\to 0} E_\varepsilon(X_\varepsilon,D) =E(u,D)\,.
\end{align}
 Fix $\rho >0$ such that $\rho < \frac{1}{2\sqrt{d}}\min_{i\neq j}\mathrm{dist}(G_i,G_j)$ and $\lambda > 6r_{\mathrm{int}}$. We set 
\begin{align*}
\mathcal{Z}^\nu_\rho = \{  R_\nu k\mid k \in \rho\mathbb{Z}^{d-1} \cap H^{e_d}\}\cap\{x\in D\mid\dist(x,\p D)\geq\sqrt{d}\}.
\end{align*}
In view of Proposition~\ref{prop:density} and Definition~\ref{def:phivac} there exists $\varepsilon_\rho >0$ such that for all $\varepsilon \in (0,\varepsilon_\rho)$ and every $x \in \mathcal{Z}^\nu_\rho$ we find $X_{\varepsilon,\rho}^x \in \mathrm{Adm}^{(z,\textbf{0})}_{\varepsilon,\lambda}(Q^\nu_\rho(x))$   such that 
\begin{align}\label{ineq:limsup-estimate-cube}
\frac{1}{\rho^{d-1}} E_\varepsilon(X_{\varepsilon,\rho}^x, Q^\nu_\rho(x)) \leq \varphi_{\mathrm{vac}}(z,\nu) +\rho\,.
\end{align}
We then define
\begin{align}\label{def:optimizer-rho}
X_{\varepsilon,\rho}:= \begin{cases} X_{\varepsilon,\rho}^x &\text{in }   Q^\nu_\rho(x) \text{ for some } x \in \mathcal{Z}^\nu_\rho \cap D\,,\\
\mathcal{L}_\varepsilon(z) &\text{in } (H^\nu_+ \cap D) \setminus \bigcup_{x \in \mathcal{Z}_\rho^\nu} Q^\nu_\rho(x)\,,\\
\emptyset &\text{otherwise.}
\end{cases}
\end{align}
It holds $E_\varepsilon(X)<\infty$ by the boundary conditions of $X^x_{\varepsilon,\rho}$, see Definition~\ref{def:admissible-set}. By construction, setting $u_{\varepsilon,\rho} = u(X_{\varepsilon,\rho})$, we obtain that 
\begin{align}\label{ineq:L1-closeness}
\|u_{\varepsilon,\rho} - u\|_{L^1(D)} \leq C|\{ x \in D\colon \mathrm{dist}(x,H^\nu) \leq 2\rho\}| \leq C \mathrm{diam}(D)^{d-1} \rho
\end{align}
for every $\varepsilon \in (0,\varepsilon_\rho)$. As $X_{\varepsilon,\rho}^x \in \mathrm{Adm}^{(z,\textbf{0})}_{\varepsilon,\lambda}(Q^\nu_\rho(x)) $ for every $x \in \mathcal{Z}_\rho^\nu$, using \eqref{ineq:limsup-estimate-cube} and \eqref{def:optimizer-rho}, we have that 
\begin{align}\label{ineq:energy-in-U}\begin{split}
E_\varepsilon(X_{\varepsilon,\rho},D) = \sum_{x \in \mathcal{Z}_\rho^\nu} E_\varepsilon(X_{\varepsilon,\rho},Q^\nu_\rho(x)) &=  \sum_{x \in \mathcal{Z}_\rho^\nu} E_\varepsilon(X_{\varepsilon,\rho}^x,Q^\nu_\rho(x)) \\&\leq \#(\mathcal{Z}_\rho^\nu \cap D)\left(\rho^{d-1} \varphi_{\mathrm{vac}}(z,0) +\rho^d \right)
\end{split}
\end{align}
for every $\varepsilon \in (0,\varepsilon_\rho)$.  The cardinality of $\mathcal{Z}_\rho^\nu \cap D$ can be estimated by
\begin{align}\label{ineq:card-estimate-Z-rho}
\rho^{d-1}\#(\mathcal{Z}_\rho^\nu \cap D) \leq \sum_{x \in \mathcal{Z}_\rho^\nu} \mathcal{H}^{d-1}(Q^\nu_\rho(x) \cap D) \leq \mathcal{H}^{d-1}(H^\nu \cap D) = \mathcal{H}^{d-1}(S_u \cap D) \,.
\end{align}
Combining \eqref{ineq:energy-in-U}, \eqref{ineq:card-estimate-Z-rho}, and Definition~\ref{def:phivac}, we infer that
\begin{align}\label{ineq:local-limsup}
\limsup_{\rho \to 0} \limsup_{\varepsilon \to 0} E_\varepsilon(X_{\varepsilon,\rho},D) \leq \int_{S_{u} \cap D} \varphi_{\mathrm{vac}}(z,\nu)\,\mathrm{d}\mathcal{H}^{d-1} = E(u,D)\,.
\end{align}
Thanks to \eqref{ineq:L1-closeness} and \eqref{ineq:local-limsup} a diagonal argument provides us with a sequence $X_\varepsilon = X_{\varepsilon, \rho_\varepsilon}$ with $X_\varepsilon \to u$ in $L^1(D)$ and such that \eqref{eq:goal-Step3} holds. This concludes the proof.
\end{proof}

\section{Auxiliary Results}\label{subsec:auxiliary-results}
We would like to state some properties of our energy, which we will heavily use throughout. For this, we recall the definition of the neighborhood of a point below \ref{H7}, namely, given a configuration $X$ with $E_\varepsilon(X)<+\infty$, the neighborhood of $x\in X$ is defined as 
\begin{align}\label{def:neighbourhood}
\mc{N}_\varepsilon(x)=\{y\in X\setminus\{x\}\mid\vert x-y\vert\leq \varepsilon\, r_{\mathrm{int}}\}\,.
\end{align}
Recall the definition of the energy \eqref{eq:scaled-energy-local}.
\begin{lem}\label{lem:elementary-properties}
Let $E_{\mathrm{cell}}$ satisfy {\rm\ref{H1}}--{\rm\ref{H5}}. Let $\varepsilon>0$ and $X\subset\R^d$ with $E_\varepsilon(X)<+\infty$. Then, the following statements are true:
\begin{enumerate}
\item[{\rm (i)}] For all $x\in X$ we have $\dist(\{x\},X\setminus\{x\})\geq\varepsilon$.
\item[{\rm (ii)}] $E_\varepsilon(RX+\tau,RA+\tau)=E_\varepsilon(X,A)$ for all $R\in SO(d),\tau\in\R^d$ and $A\subset\R^d$ Borel.
\item[{\rm (iii)}] $E_\varepsilon(X,A)\leq E_\varepsilon(X,B)$ for all $A\subset B$.
\item[{\rm (iv)}] $E_\varepsilon(X,A\cup B)=E_\varepsilon(X,A)+E_\varepsilon(X,B)$ for all $A,B\subset\R^d$ with $A\cap B=\emptyset$.
\item[{\rm (v)}] There exists $ C>0$ such that for all $ A\subset\R^d$ Borel there holds $\#(X\cap A)\leq\frac{C}{\varepsilon^d}\mc{L}^d((A)_\varepsilon)$.
\item[{\rm (vi)}] There exists $ C>0$ such that for all  $r>0$ and $x\in X$ there holds  $\#(X\cap\overline{B}_{r\varepsilon}(x))\leq C (r+1)^d$.
\item[{\rm (vii)}] For all $\lambda, \varepsilon>0$ and $A\subset\R^d$ it holds:
$$E_{\mathrm{cell}}^{\lambda\varepsilon}(\lambda x,\lambda X)=E_{\mathrm{cell}}^{\varepsilon}(x,X)\quad\forall x\in X\,.$$
In particular
$$E_{\lambda\varepsilon}(\lambda X,\lambda A)=\lambda^{d-1}E_{\varepsilon}(X,A)\,.$$
\item[{\rm (viii)}] If $X=Y\cup Z$ and $\mathrm{dist}(Y,Z) >r_{\mathrm{int}}\varepsilon$. Then,  for all $A\subset\R^d$ we have 
\begin{align*}
E_\varepsilon(X,A)=E_\varepsilon(Y,A)+E_\varepsilon(Z,A)\,.
\end{align*}
\item[{\rm (ix)}] Let $x_0\in X$ with $\{x_0\}\cup\mc{N}_\varepsilon(x_0)\subset A$ for some $A \subset \mathbb{R}^d$. Then 
\begin{align*}
E_\varepsilon(X\setminus\{x_0\})<E_\varepsilon(X) \quad \iff \quad E_\varepsilon(X\setminus\{x_0\},A)<E_\varepsilon(X,A)\,.
\end{align*} 
\item[{\rm (x)}]  Let $X,Y \in \mathcal{X}$ and $A \subset \mathbb{R}^d$ be such that $ X \cap \overline{(A)_{r_{\mathrm{int} }\varepsilon}} = Y \cap \overline{(A)_{r_{\mathrm{int} }\varepsilon}}$. Then,
\begin{align*}
E_\varepsilon(X,A) = E_\varepsilon(Y,A)\,.
\end{align*}
\end{enumerate}
\end{lem}

\begin{proof}
All items except {\rm (v)}, {\rm (vi)} and {\rm (ix)} follow directly from \ref{H1}-\ref{H5} and the definition of the configurational energy, see \eqref{eq:cell-energy-scaled} and \eqref{eq:scaled-energy-local}. We now prove {\rm (v)}, {\rm (vi)} and {\rm (ix)}.  \\
\noindent {\it Proof of {\rm (v)}:} We note that by {\rm (i)} it holds $B_\frac{\varepsilon}{2}(x)\cap B_\frac{\varepsilon}{2}(y)=\emptyset$ for $x,y\in X$ $x\neq y$.  Therefore, as $\bigcup_{x \in X\cap A} B_\frac{\varepsilon}{2}(x)\subset(A)_\varepsilon$,  we have
\begin{align*}
\frac{\varepsilon^d}{2^d}\omega_d\#\{ X\cap A\}=\sum_{x\in X\cap A}\mc{L}^d(B_\frac{\varepsilon}{2}(x))=\mc{L}^d\left(\bigcup_{x\in X\cap A}B_\frac{\varepsilon}{2}(x)\right)\leq\mc{L}^d((A)_\varepsilon)\,,
\end{align*}
where $\omega_d$ denotes the volume of the $d$-dimensional unit ball. \\
\noindent {\it Proof of {\rm (vi)}:} This is a consequence of {\rm (v)} by choosing  $A=\overline{B}_{r\varepsilon}(x)$ and noticing $(\overline{B}_{r\varepsilon}(x))_\varepsilon=B_{(r+1)\varepsilon}(x)$ with $\mc{L}^d(B_{(r+1)\varepsilon}(x))=\omega_d\varepsilon^d(r+1)^d$. \\
\noindent {\it Proof of {\rm (ix)}:}  Let $x_0\in X$ and $A \subset \mathbb{R}^d$ be as in the statement, i.e., $\{x_0\}\cup\mc{N}_\varepsilon(x)\subset A$. Using \eqref{def:neighbourhood}, for all $x \in X \cap \mathbb{R}^d \setminus A$ we have  $|x_0-x| > \varepsilon\, r_{\mathrm{int}}$. Therefore, by \ref{H4} we have $E_{\mathrm{cell}}^\varepsilon(x,X)=E_{\mathrm{cell}}^\varepsilon(x,X\setminus\{x_0\})$ for all $x \in  X \cap \mathbb{R}^d \setminus A$. Together with the fact that $x_0 \in A$ this implies 
\begin{align*}
E_\varepsilon(X \setminus \{x_0\},\mathbb{R}^d\setminus A) &= \varepsilon^{d-1} \sum_{x \in (X \setminus \{x_0\}) \cap (\mathbb{R}^d \setminus A)}  E_{\mathrm{cell}}^\varepsilon(x, X\setminus \{x_0\})\\&=\varepsilon^{d-1} \sum_{x \in X \cap (\mathbb{R}^d \setminus A)}  E_{\mathrm{cell}}^\varepsilon(x, X) = E_\varepsilon(X,\mathbb{R}^d \setminus A)\,.
\end{align*}
Now, using {\rm (iv)}, we obtain
\begin{align*}
E_\varepsilon(X \setminus \{x_0\}) = E_\varepsilon(X \setminus \{x_0\}, A) +  E_\varepsilon(X \setminus \{x_0\},\mathbb{R}^d\setminus A) = E_\varepsilon(X \setminus \{x_0\}, A) +  E_\varepsilon(X,\mathbb{R}^d \setminus A)
\end{align*}
and 
\begin{align*}
E_\varepsilon(X) = E_\varepsilon(X, A) +  E_\varepsilon(X,\mathbb{R}^d \setminus A)\,.
\end{align*}
This implies {\rm (ix)} and concludes the proof.
\end{proof}

Next, we will show that our energies are coercive.   To prove the coercivity of our energies, we will quickly recall some properties of the Voronoi cells. In particular, for lattices, such as the ones we are considering, the Voronoi cell is a convex  $d$-dimensional polytope. By the definition of our interpolation, see~\eqref{def:interpolation}, the jump set is therefore concentrated only on some $(d-1)$-dimensional faces of such polytopes. As the Voronoi cells tessellate the space, such a $(d-1)$-dimensional face is defined by two neighboring atoms or, equivalently, it is the intersection of the closure of two neighboring Voronoi cells. This will be a key observation for proof of the next lemma. 
\begin{lem}[Coercivity]\label{lem:coercivity}
Let $E_{\mathrm{cell}}$ satisfy {\rm\ref{H1}}--{\rm \ref{H7}}. Let $X$ be a configuration such that $E_\varepsilon(X)<+\infty$ and $A\subset\R^d$ Borel. Then there is a universal constant $C>0$ such that
\begin{align*}
\mc{H}^{d-1}(J_{u}\cap A)\leq CE_\varepsilon(X,(A)_{2\varepsilon r_{\mathrm{crys}}}),
\end{align*}
where $J_{u}$ is the jumpset of the function $u$ associated to $X$ by \eqref{def:interpolation}.
\end{lem}

\begin{proof} Fix a configuration $X$ such that $E_\varepsilon(X) <+\infty$ and  
let $u$ be given by \eqref{def:interpolation}.  As $u \in PC(\mathbb{R}^d;\mathcal{Z})$ we have that $u=\sum_{j=1}^\infty z_j\chi_{G_j}$ with pairwise disjoint $G_j$ and pairwise distinct $z_j$. By the definition of $u$ each $G_j$ consists of a finite union of Voronoi cells $V_{\mc{L}_\varepsilon(z_j)}(x)$, see \eqref{def:interpolation}. Here $z_j=z(x)$ is the associated lattice for $x\in X_{\varepsilon,\mathrm{int}}$, see~\eqref{def:interpolation-points}, and the Voronoi cell corresponds to the lattice, which is defined by $z_j$ and $x\in  X_{\varepsilon,\mathrm{int}}$ is its barycenter. As mentioned above, only $(d-1)$-dimensional faces contribute to the $(d-1)$-dimensional Hausdorff measure of the jump set, as all the other faces are of lower dimension. As any Voronoi cell regardless of its orientation is convex and contained in $\overline{B}_{\varepsilon r_{\mathrm{cryst}}}(x)$, we can bound the perimeter of the set, in particular we can bound the $\mc{H}^{d-1}$-measure of each of its $(d-1)$-faces by the perimeter of the set. Namely, we have 
\begin{align}\label{ineq:perimeter-voronoi-ball}
\mathcal{H}^{d-1}(\partial V_{\mc{L}_\varepsilon(z_j)}(x))\leq \mathcal{H}^{d-1}(\partial B_{\varepsilon r_{\mathrm{cryst}}}(x))=C(d)(\varepsilon r_{\mathrm{crys}})^{d-1}=C(d,\mathcal{L})\varepsilon^{d-1}\,.
\end{align} 
If there exists $y \in \overline{B}_{\varepsilon r_{\mathrm{cryst}}}(x)$ such that $E_{\varepsilon,\mathrm{cell}}(y,X) >0$, then, due to \ref{H6} and \eqref{eq:cell-energy-scaled}, we have $E_{\varepsilon,\mathrm{cell}}(y,X) \geq c$. Now, using \eqref{eq:scaled-energy-local} and \eqref{ineq:perimeter-voronoi-ball}, we obtain
\begin{align*}
\mathcal{H}^{d-1}(J_u \cap V_{\mc{L}_\varepsilon(z_j)}(x)) \leq C(d,\mathcal{L})\varepsilon^{d-1} \leq Cc^{-1} \varepsilon^{d-1} E_{\varepsilon,\mathrm{cell}}(y,X) \leq C E_\varepsilon(X,\overline{B}_{\varepsilon r_{\mathrm{crys}}}(x))\,.
\end{align*}
Due Lemma~\ref{lem:elementary-properties}(vi) we have that $\#(X \cap \overline{B}_{\varepsilon r_{\mathrm{cryst}}}(x_0)) \leq C r_{\mathrm{cryst}}^d$ for all $x_0 \in \mathbb{R}^d$ and thus each such $y$ will be used only for a finite number of $x \in X$. This implies the statement of the Lemma for all points $x$ such that $E_\varepsilon(x,\overline{B}_{\varepsilon r_{\mathrm{cryst}}}(x))>0$. On the other hand, if $E_{\varepsilon,\mathrm{cell}}(y,X)=0$ for all $y \in \overline{B}_{\varepsilon r_{\mathrm{cryst}}}(x)$ then, as $x \in X_{\varepsilon,\mathrm{int}}$,  we have that $X\cap \overline{B}_{\varepsilon r_{\mathrm{cryst}}}(x) = \mathcal{L}_\varepsilon(z(x)) \cap \overline{B}_{\varepsilon r_{\mathrm{cryst}}}(x)$.  Now, for each $(d-1)$-dimensional face $F$ there exists $y \in \mathcal{L}_\varepsilon(z(x))$ such that $|x-z|=|y-z|$ for all $z \in F$. Hence, $|x-y| \leq |x-z| +|y-z| \leq \varepsilon 2R_V \leq \varepsilon r_{\mathrm{crys}}$ and thus $y \in X \cap \overline{B}_{\varepsilon r_{\mathrm{crys}}}$. By \ref{H3}, \ref{H7}  and Definition~\ref{def:orientation}, for $x,y \in X$ such that $|x-y|\leq \varepsilon  r_{\mathrm{crys}}$ we have that $z(x) =z(y)$. This implies that $\mathcal{H}^{d-1}(J_u \cap V_{\mc{L}_\varepsilon(z_j)}(x))=0$ and also in this case the statement of the Lemma follows.
\end{proof}

We need a relation of our interpolation functions with respect to some scaling $\lambda>0$. Then we are able to use a fundamental estimate construction that allows us to pass from $L^1$-convergence of boundary values to converging boundary values in the proof of the lower bound of our $\Gamma$-convergence result.

\begin{lem}\label{lem:scaling}
Let $E_{\mathrm{cell}}$ satisfy {\rm \ref{H1}}--{\rm \ref{H5}}. For $\varepsilon>0$ consider the configurations $X_\varepsilon$ satisfying $E_\varepsilon(X_\varepsilon)<+\infty$ and $\lambda X_\varepsilon$ for $\lambda>0$. By $u^\lambda_{\lambda\varepsilon}$ and $u_\varepsilon$ we denote the functions corresponding to $\lambda X_\varepsilon$ and $X_\varepsilon$, respectively. Then it holds
$$u_{\lambda\varepsilon}^\lambda(\lambda x)=u_\varepsilon(x)\quad\forall x\in\R^d.$$
Moreover, for each bounded set $A\subset\R^d$ we have $u_{\lambda\varepsilon}^\lambda\to u(\lambda^{-1}\cdot)$ in $L^1(\lambda A)$ as $\varepsilon\to0$ if and only if $u_\varepsilon\to u$ in $L^1(A)$ as $\varepsilon\to0$. 
\end{lem}
\begin{proof}
We notice that by Lemma~\ref{lem:elementary-properties}{\rm (vii)} we have 
$$E_{\mathrm{cell}}^\varepsilon(x,X) =0 \quad \iff \quad E_{\mathrm{cell}}^{\lambda\varepsilon}(\lambda x,\lambda X)=0\,$$
and
\begin{align*}
X \cap \overline{B}_{\varepsilon r_{\mathrm{crys}}}(x) = \varepsilon R(\mathcal{L}+\tau) \cap \overline{B}_{\varepsilon r_{\mathrm{crys}}}(x)  \iff \lambda X \cap \overline{B}_{\varepsilon \lambda r_{\mathrm{crys}}}(\lambda x) = \varepsilon \lambda R(\mathcal{L}+\tau) \cap \overline{B}_{\varepsilon \lambda r_{\mathrm{crys}}}(\lambda x)\,.
\end{align*}
This implies that we interpolate on the Voronoi cell $V_{\mathcal{L}_\varepsilon(z)}(x)$ the value $z$ if and only if we interpolate on the Voronoi cell $V_{\mathcal{L}_{\lambda\varepsilon}(z)}( \lambda x)$ the value $z$.  Thus, by \eqref{def:interpolation} it holds $u_\varepsilon(y)=u_{\lambda\varepsilon}^\lambda(\lambda y)$ for all $y\in\R^d$. Lastly, for every bounded $A\subset\R^d$, performing the change of variables $y=\lambda x$,   there holds 
\begin{align*}
\lambda^d\int_A\vert u_\varepsilon(x)-u(x)\vert\,\mathrm{d}x=\lambda^d\int_A\vert u_{\lambda\varepsilon}^\lambda(\lambda x)-u(x)\vert\,\mathrm{d}x=\int_{\lambda A}\vert u^\lambda_{\lambda\varepsilon}(y)-u(\lambda^{-1}y)\vert\,\mathrm{d} y\,.
\end{align*}
\end{proof}

\section[Cell Formula I]{Cell Formula I: Relation of $L^1$-convergence and converging boundary values}\label{sec:L1-to-converging-data}
In this section, we present a first reduction result for our cell formula. In particular, we will show that the condition of $L^1$-convergence in the definition of $\psi$, see \eqref{eq:cell-formulaL1}, can be replaced by converging boundary values. To this end, we keep $\lambda > 6r_{\mathrm{int}}$ fixed, and we omit it in the notation.  Here, the bound on $\lambda$ ensures that the upper and lower boundary regions already contain a large connected portion of the lattice. Precisely we introduce the function $\Phi:\mc{Z}\times\mc{Z}\times\mathbb{S}^{d-1}\to[0,+\infty)$ defined as
\begin{align}\label{def:Phi}
\begin{split}
\Phi(z^+,z^-,\nu):=\min\big\{\liminf_{\varepsilon\to0}\inf\big\{&E_\varepsilon(X_\varepsilon,Q^\nu(y_\varepsilon))\mid  y_\varepsilon\in\R^d\,,\\&X_\varepsilon \in \mathrm{Adm}_{\varepsilon,\lambda}^{(z_\varepsilon^+,z_\varepsilon^-)}(Q^\nu(y_\varepsilon))\big\}\mid\{z_\varepsilon^\pm\}_\varepsilon\subset\mc{Z} \text{ with }z^\pm_\varepsilon\to z^\pm\big\}\,.
\end{split}
\end{align}
Here, the admissible set is defined in \eqref{def:admissible-set}. By a standard diagonal argument, one can see that the minimum in \eqref{def:Phi} is attained. This definition shows us that near the boundary of the cube our configuration is contained in at most two different lattices $\mc{L}_\varepsilon(z_\varepsilon^\pm)$ (and it is contained in only one if $z^\pm=0$). The goal of this section is to now show the following lemma:
\begin{lem}[Relation between $\psi$ and $\Phi$]\label{lem:psi-phi}
Let $E_{\mathrm{cell}}$ satisfy {\rm \ref{H1}}--{\rm \ref{H7}}. Let $z^+,z^-\in\mc{Z}$ and $\nu\in\mathbb{S}^{d-1}$. Then
\begin{equation}\label{lemeq:psi-phi}
\psi(z^+,z^-,\nu)\geq\Phi(z^+,z^-,\nu)\,.
\end{equation}
\end{lem}
This provides the first step towards proving Proposition~\ref{prop:relationpsiphi}. In Section~\ref{sec:converging-fixed-bdryvalues}, we will show $\Phi(z^+,z^-,\nu)=\varphi(z^+,z^-,\nu)$ for all $z^+,z^-\in\mc{Z}$ and $\nu\in\mathbb{S}^{d-1}$, see Lemma~\ref{lem:Phi=varphi}. The proof of Lemma~\ref{lem:psi-phi} relies on a cut-off argument, allowing us to construct configurations attaining the appropriate boundary values. This is a customary practice in the analysis of cell formulas, in contrast to problems defined on Sobolev spaces or partitions, where this is usually done via a convex combination of functions. Here, our discrete problem is more delicate as it is very sensitive to small changes in the configuration by our cell energy, see \ref{H1}.

One of the key observations in the proof is the fact that the energy of optimal sequences in \eqref{eq:cell-formulaL1} is concentrated asymptotically arbitrarily close to the interface. In order to prove this, we need the following preliminary step, which shows us that in the definition of $\psi$ we can replace cubes by rectangles. Recall \eqref{eq:purejump} as well as \eqref{def:admissible-set}.
\begin{lem}[Density $\psi$ on rectangles]\label{lem:density-rec}
Let $E_{\mathrm{cell}}$ satisfy {\rm \ref{H1}}--{\rm \ref{H5}}. For all $z^+,z^-\in\mc{Z}, \nu\in\mathbb{S}^{d-1}$ and all $l,h>0$ there holds
\begin{align}\label{lemeq:density-rec}
\begin{split}
\psi(z^+,z^-,\nu)=\inf\big\{&\liminf_{\varepsilon\to0}\frac{1}{l^{d-1}}E_\varepsilon(X_\varepsilon,R_{l,h}^\nu(y_\varepsilon))\mid y_\varepsilon\in\R^d,\\
&\lim_{\varepsilon\to0}\int_{R_{l,h}^\nu}\vert u_\varepsilon(x+y_\varepsilon)-u_{z^+,z^-}^\nu(x)\vert\,\mathrm{d}x=0\big\}\,. 
\end{split}
\end{align}
\end{lem}

\begin{proof}
Using the respective properties of the energy functionals $E_\varepsilon$, this proof is entirely analogous to \cite[Lemma 4.2]{FriedrichKreutzSchmidt}. 
\end{proof}

We will now prove Lemma~\ref{lem:psi-phi}.

\begin{proof}[Proof of Lemma~\ref{lem:psi-phi}] Considering \eqref{eq:cell-formulaL1}, we can choose a subsequence in $\varepsilon$ (not relabeled) and configurations $X_\varepsilon\subset\R^d$ and $y_\varepsilon\in\R^d$, such that 
\begin{align*}
\lim_{\varepsilon\to0}\int_{Q^\nu}\vert u_\varepsilon(x+y_\varepsilon)-u^\nu_{z^+,z^-}(x)\vert\,dx=0
\end{align*}
and
\begin{equation}\label{eq:optimal-sequence-psi}
\psi(z^+,z^-,\nu)=\lim_{\varepsilon\to0}E_\varepsilon(X_\varepsilon, Q^\nu(y_\varepsilon))\,.
\end{equation}
The following proof will follow the proof of \cite[Lemma 4.1]{FriedrichKreutzSchmidt} with the adaptations necessary in our case. This will be done in several steps by a refined cut-off construction while taking into consideration our rigid potentials. In Step 1, we show that the energy $X_\varepsilon$ is concentrated around a strip close to the limiting interface. Step 2 will allow us to select a dominant component in the upper and lower half cube, where {\it component} refers to the points subset of a specific lattice. The following steps will modify our configuration $X_\varepsilon$ such that it coincides with the dominant component near the boundary of the respective half cube. In Step~3, we show that our configuration essentially agrees, up to   $o(\varepsilon^{-d})$ atoms, with a lattice $\mathcal{L}_\varepsilon(z_\varepsilon^+)$ on the upper half cube. In Step~4, we select a good layer by an averaging argument in which we modify our configuration. The term {\it good} refers to the configuration coinciding with the lattice $\mathcal{L}_\varepsilon(z_\varepsilon^+)$ in that layer up to $o(\varepsilon^{-(d-1)})$ atoms. This allows us to perform a geometric cut-off by interpolating between the lattice and the configuration. Step~5 shows that the previously constructed configuration is an asymptotic lower bound of the energy for the original configuration. Step~6 concludes that the constructed configuration is a competitor in the definition of $\Phi$. Steps~3-6 are done under the assumption that the majority phases in the upper and lower half-cube are determined by lattices. For the case of vacuum, we describe the necessary adaptations in Step~7. \\
\noindent {\bf Step 1:} {\it The energy concentrates near the line $H^\nu(y_\varepsilon)=\{\langle\nu,(x-y_\varepsilon)\rangle=0\}$}.
We want to show that for all $\delta\in(0,1)$ it holds
\begin{equation}\label{eq:Step1-concentration}
\lim_{\varepsilon\to0}E_\varepsilon(X_\varepsilon,Q^\nu(y_\varepsilon)\setminus R^\nu_{1,\delta}(y_\varepsilon))=0\,.
\end{equation}
Using Lemma~\ref{lem:elementary-properties}(iii), Lemma \ref{lem:density-rec}, and \eqref{eq:optimal-sequence-psi} as well as the fact that $\{X_\varepsilon\}_\varepsilon$ is admissible in the definition of $\psi$ on $R^\nu_{1,\delta}$, see \eqref{lemeq:density-rec}, we obtain
\begin{align*}
\psi(z^+,z^-,\nu)&\leq\liminf_{\varepsilon\to0}E_\varepsilon(X_\varepsilon,R^\nu_{1,\delta}(y_\varepsilon))\\
&\leq\lim_{\varepsilon\to0}E_\varepsilon(X_\varepsilon,Q^\nu(y_\varepsilon))=\psi(z^+,z^-,\nu)\,.
\end{align*}
Using Lemma~\ref{lem:elementary-properties}(iv) implies
\begin{align*}
0&\leq\limsup_{\varepsilon\to0}E_\varepsilon(X_\varepsilon,Q^\nu(y_\varepsilon)\setminus R^\nu_{1,\delta}(y_\varepsilon))\\
&=\limsup_{\varepsilon\to0}\left(E_\varepsilon(X_\varepsilon,Q^\nu(y_\varepsilon))-E_\varepsilon(X_\varepsilon,R^\nu_{1,\delta}(y_\varepsilon))\right)\\
&\leq\lim_{\varepsilon\to0}E_\varepsilon(X_\varepsilon,Q^\nu(y_\varepsilon))-\liminf_{\varepsilon\to0}E_\varepsilon(X_\varepsilon,R^\nu_{1,\delta}(y_\varepsilon))=0\,.
\end{align*}
This shows \eqref{eq:Step1-concentration} and thus concludes Step 1.
In order to shorten the notation for the rest of the proof, we omit the dependence on the center $y_\varepsilon$ and write $Q^\nu_\rho$ instead of $Q^\nu_\rho(y_\varepsilon)$ for $\rho>0$ and $R^\nu_{1,\delta}$ instead of $R^\nu_{1,\delta}(y_\varepsilon)$. In the following, we fix $\delta\in(0,1)$ small enough, and we suppose without loss of generality that $\varepsilon\ll\delta$ as we consider the limit as $\varepsilon\to0$. Also, omitting the center, we define the rectangles $\mathrm{R}^\pm_{\delta,\varepsilon}=Q^{\nu,\pm}_{1-6\varepsilon r_{\mathrm{int}}}\setminus R^\nu_{1-6\varepsilon r_{\mathrm{int}},\delta}$, where $Q^{\nu,\pm}_{r}$ is defined in Subsection~\ref{subsec:cubes}. We only describe the construction in $Q^{\nu,+}$ as the construction in $Q^{\nu,-}$ is analogous. \\
\noindent {\bf Step 2:} {\it Single dominant component in the upper and lower half.} We prove the existence of a sequence $\{z_\varepsilon^\pm\}_\varepsilon\subset\mc{Z}$ such that $z_\varepsilon^\pm\to z^\pm$ as $\varepsilon \to 0$ and
\begin{equation}\label{ineq:measure-energy-Step2}
\mc{L}^d\left(\{u_\varepsilon\neq z_\varepsilon^\pm\}\cap \mathrm{R}^\pm_{\delta,\varepsilon}\right)\leq CE_\varepsilon\left(X_\varepsilon,Q^\nu\setminus R^\nu_{1,\frac{\delta}{2}}\right)\,,
\end{equation}
where $C>0$ is a universal constant independent of $\varepsilon$.
By \eqref{eq:representationu} and \eqref{eq:finitesurfaceu}  we can write $u_\varepsilon=\sum^\infty_{j=1}\chi_{G_j^\varepsilon}z^\varepsilon_j$ for pairwise distinct $\{z_j^\varepsilon\}_j\subset\mc{Z}\setminus\{{\bf 0}\}$ and pairwise disjoint $\{G_j^\varepsilon\}_j\subset\R^d$. By Lemma~\ref{lem:coercivity} and Lemma~\ref{lem:elementary-properties}(iii) we get
\begin{align}\label{ineq:perimeter-energy}
\sum^\infty_{j=1}\mc{H}^{d-1}(\p^*G_j^\varepsilon\cap \mathrm{R}_{\delta,\varepsilon}^+)\leq CE_\varepsilon\left(X_\varepsilon,(\mathrm{R}^+_{\delta,\varepsilon})_{2\varepsilon r_{\mathrm{crys}}}\right)\leq C E_\varepsilon\left(X_\varepsilon,Q^\nu\setminus R^\nu_{1,\frac{\delta}{2}}\right)\,,
\end{align}
where in the last inequality we used $(\mathrm{R}_{\delta,\varepsilon}^+)_{2\varepsilon r_{\mathrm{crys}}}\subset Q^\nu\setminus R^\nu_{1,\frac{\delta}{2}}$ for $\varepsilon\ll\delta$ small enough. We also define the vacuum inside $Q^\nu$ by $G_0^\varepsilon:=Q^\nu\setminus\bigcup_{j=1}^\infty G_j^\varepsilon$. By the relative isoperimetric inequality (see \cite[Theorem 2, Section 5.6.2]{EvansGariepy92}\footnote{The theorem there is only stated and proved in a ball, but the argument uses Poincaré inequalities and therefore extends to the rectangles $\mathrm{R}^+_{\delta,\varepsilon}$. As the ratio of all sides of $\mathrm{R}^+_{\delta,\varepsilon}$ can be controlled uniformly independently of $\varepsilon$ and $\delta$, the constant is independent of $\delta$ and $\varepsilon$.}) there exists a $c>0$ such that for all $j\in\N_0$ it holds
\begin{align}\label{ineq:isoperimetric-ineq}
\begin{split}
\min\big\{\mc{L}^d(G_j^\varepsilon\cap \mathrm{R}^+_{\delta,\varepsilon}),\mc{L}^d(\mathrm{R}^+_{\delta,\varepsilon}\setminus G_j^\varepsilon)\big\}
&\leq\min\big\{\mc{L}^d(G_j^\varepsilon\cap \mathrm{R}^+_{\delta,\varepsilon}),\mc{L}^d(\mathrm{R}^+_{\delta,\varepsilon}\setminus G_j^\varepsilon)\big\}^\frac{d-1}{d}\mc{L}^d(\mathrm{R}^+_{\delta,\varepsilon})^\frac1d\\
&\leq c\mc{H}^{d-1}(\p^*\!G_j^\varepsilon\cap \mathrm{R}^+_{\delta,\varepsilon})\,,
\end{split}
\end{align}
where we used $\mc{L}^d(\mathrm{R}^+_{\delta,\varepsilon})\leq1$. Thus, from \eqref{ineq:perimeter-energy}, \eqref{ineq:isoperimetric-ineq}, and $\p^*\!G_0^\varepsilon\cap \mathrm{R}^+_{\delta,\varepsilon}\subset\bigcup_{j=1}^\infty(\p^*\!G_j^\varepsilon\cap \mathrm{R}^+_{\delta,\varepsilon})$ it follows that
\begin{equation}\label{ineq:volume-energy}
\sum^\infty_{j=0}\min\left\{\mc{L}^d(G_j^\varepsilon\cap \mathrm{R}^+_{\delta,\varepsilon}),\mc{L}^d(\mathrm{R}^+_{\delta,\varepsilon}\setminus G_j^\varepsilon)\right\}\leq C E_\varepsilon\left(X_\varepsilon,Q^\nu\setminus R^\nu_{1,\frac{\delta}{2}}\right)\,.
\end{equation}
This implies the existence of a unique dominant component, i.e., the existence of $j_\varepsilon\in\N_0$ such that
\begin{equation}\label{ineq:dominant-comp-measure}
\mc{L}^d(G_{j_\varepsilon}^\varepsilon\cap \mathrm{R}^+_{\delta,\varepsilon})>\frac12\mc{L}^d(\mathrm{R}^+_{\delta,\varepsilon})\geq \frac{1}{4}\,,
\end{equation}
which can be seen as follows. Assume that this is not the case. Then for all $j\in\N_0$ we have
\begin{align*}
\min\left\{\mc{L}^d(G_j^\varepsilon\cap \mathrm{R}^+_{\delta,\varepsilon}),\mc{L}^d(\mathrm{R}^+_{\delta,\varepsilon}\setminus G_j^\varepsilon)\right\}=
\mc{L}^d(G_j^\varepsilon\cap \mathrm{R}^+_{\delta,\varepsilon})\,.
\end{align*}
Using \eqref{ineq:volume-energy}, for $\delta$ and $\varepsilon$ small enough, this implies  
\begin{align*}
\frac{1}{4}\leq \mc{L}^d(\mathrm{R}^+_{\delta,\varepsilon})=\sum^\infty_{j=0}\mc{L}^d(G_j^\varepsilon\cap \mathrm{R}^+_{\delta,\varepsilon})\leq C E_\varepsilon\left(X_\varepsilon,Q^\nu\setminus R^\nu_{1,\frac{\delta}{2}}\right)\,.
\end{align*}
This is a contradiction to Step 1 (applied with $\delta/2$), which shows that there exists $j_\varepsilon \in \mathbb{N}_0$ such that \eqref{ineq:dominant-comp-measure} is true. Now, using \eqref{ineq:volume-energy} and \eqref{ineq:dominant-comp-measure}, this implies \eqref{ineq:measure-energy-Step2} with the choice $z^+_\varepsilon=z^\varepsilon_{j_\varepsilon}$.  To conclude this step, we note that, by using 
\begin{align*}
\lim_{\varepsilon\to0}\int_{Q^\nu}\vert u_\varepsilon(x+y_\varepsilon)-u^\nu_{z^+,z^-}(x)\vert\,\mathrm{d} x=0
\end{align*}
and \eqref{ineq:dominant-comp-measure}, necessarily $z_\varepsilon^+\to z^+$ as $\varepsilon \to 0$. The rest of the proof will be divided into the two cases: (a) $z^+_\varepsilon\neq\textbf{0}$ and (b) $z^+_\varepsilon=\textbf{0}$, that is, $X_\varepsilon$ converges to the lattice with orientation $z^+$ in the upper half of the cube or it converges to vacuum in the upper half of the cube. We will deal with case (a) now and at the end of the proof indicate the necessary changes to treat (b). \\
\noindent {\bf Step 3:} {\it Cardinality estimate.} We want to prove that there exists $C>0$ such that
\begin{equation}\label{ineq:card-estimate-Step3}
\varepsilon^d\#((\mc{L}_\varepsilon(z^\pm_\varepsilon)\triangle X_\varepsilon)\cap \mathrm{R}^+_{\delta,\varepsilon})\leq C E_\varepsilon\left(X_\varepsilon,Q^\nu\setminus R^\nu_{1,\frac{\delta}{2}}\right)\,.
\end{equation}
First, let $x\in(\mc{L}_\varepsilon(z^+_\varepsilon)\setminus X_\varepsilon)\cap \mathrm{R}^+_{\delta,\varepsilon}$. Then by definition of $u_\varepsilon$ in \eqref{def:interpolation} we have
\begin{align*}
u_\varepsilon(y)\neq z_\varepsilon^+\quad \text{for all } y\in B_{r_V \varepsilon}(x)\,,
\end{align*}
otherwise there exist $y\in B_{r_V\varepsilon}(x)$ and $x'\in X_\varepsilon$ such that $y \in V_{\mathcal{L}_\varepsilon(z_\varepsilon^+)}(x')$, in particular $|y-x'| \leq R_V \varepsilon$, and $\overline{B}_{\varepsilon r_{\mathrm{crys}} }(x')\cap X_\varepsilon=\overline{B}_{\varepsilon r_{\mathrm{crys}} }(x')\cap\mc{L}_\varepsilon(z_\varepsilon^+)$. Thus, $x,x'\in\mc{L}_\varepsilon(z^+_\varepsilon)$, which together with $\vert x-x'\vert \leq \vert x-y\vert+\vert y- x'\vert\leq 2R_V \varepsilon \leq   r_{\mathrm{crys}} \varepsilon$ (for the last inequality recall \ref{H3}) implies $x\in X_\varepsilon$, which is a contradiction.
On the other hand if $x\in(X_\varepsilon\setminus\mc{L}_\varepsilon(z_\varepsilon^+))\cap  \mathrm{R}^+_{\delta,\varepsilon}$, then there exists a $x_0\in\mc{L}_\varepsilon(z_\varepsilon^+)\cap \mathrm{R}^+_{\delta,\varepsilon}$ with $\vert x-x_0\vert \leq R_V\varepsilon$. This in particular implies $X\cap \overline{B}_{r_{\mathrm{crys}}\varepsilon}(x_0)\neq\mc{L}_\varepsilon(z^+_\varepsilon)\cap \overline{B}_{r_{\mathrm{crys}}\varepsilon}(x_0)$ and thus by \eqref{def:interpolation} together with \ref{H3} we have
\begin{align*}
u_\varepsilon(y)\neq z_\varepsilon^+\quad \text{for all } y\in B_{r_V \varepsilon}(x_0)\,.
\end{align*}
We note that such a $x_0$ can only be chosen for at most finitely many $x\in X_\varepsilon$ independent of $\varepsilon$ as $\#(X_\varepsilon\cap \overline{B}_{R_V\varepsilon}(x_0))\leq C$ by finite energy and Lemma~\ref{lem:elementary-properties}(vi).
As $\mc{L}^d(B_{r_V \varepsilon}(x)\cap  \mathrm{R}^+_{\delta,\varepsilon})\geq c\varepsilon^d$ for all $x\in\mc{L}_\varepsilon(z^+_\varepsilon)\cap  \mathrm{R}^+_{\delta,\varepsilon}$ we conclude by the last two identities for $u_\varepsilon$ that together with Step 2 (applied with the set $\left( \mathrm{R}^+_{\delta,\varepsilon}\right)_{r_{\mathrm{crys}}\varepsilon}$)   we have
\begin{align*}
\varepsilon^d\#((\mc{L}_\varepsilon(z^+_\varepsilon)\triangle X_\varepsilon)\cap \mathrm{R}^+_{\delta,\varepsilon})\leq C\mc{L}^d(\{u_\varepsilon\neq z_\varepsilon^+\}\cap  \left( \mathrm{R}^+_{\delta,\varepsilon}\right)_{r_{\mathrm{crys}}\varepsilon})\leq C E_\varepsilon(X_\varepsilon, Q^\nu\setminus R^\nu_{1,\frac{\delta}{2}})\,.
\end{align*}
\noindent {\bf Step 4:} {\it Cut-off construction.} 
In this step, given $\lambda > 6r_{\mathrm{int}}$, for every $\delta>0$, we construct a new configuration $Y_\varepsilon\subset\R^2$ such that $Y_\varepsilon=\mc{L}_\varepsilon(z^+_\varepsilon)$ on $\p_{\lambda\varepsilon}^+ Q^\nu$, see \eqref{def:discrete-boundary}, and
\begin{align*}
\limsup_{\varepsilon \to 0 } E_\varepsilon(Y_\varepsilon,Q^\nu(y_\varepsilon)) \leq \lim_{\varepsilon \to 0 } E_\varepsilon(X_\varepsilon,Q^\nu(y_\varepsilon))+C\delta\,.
\end{align*} 
The proof of the energy estimate of $Y_\varepsilon$ is postponed to Step~5, while in Step~4 we only describe the construction of $Y_\varepsilon$. In Step~6, we briefly describe the conclusions once, then we perform the analogous construction on the lower half cube.
Let $N_\varepsilon=\lfloor\frac{\delta}{9 r_{\mathrm{int}} \varepsilon}\rfloor$ (we omit the dependence on $\delta$ for simplicity of notation).
For $k\in\{0,\dots,N_\varepsilon+1\}$ let $r_k:=1-\delta+4k r_{\mathrm{int}} \varepsilon$. We define
\begin{equation}\label{def:Sk-Step4}
S_k^\varepsilon:=(Q^{\nu,+}_{r_k}\setminus Q^{\nu,+}_{r_{k-1}})\setminus R^\nu_{1,\delta}\,.
\end{equation}
For $k\in\{1,\dots,N_\varepsilon\}$ we also define the thickened layers $L_k^\varepsilon:=S^\varepsilon_{k-1}\cup S^\varepsilon_{k}\cup S^\varepsilon_{k+1}$.
Notice that for $\varepsilon>0$ small enough we have   $r_{N_\varepsilon+1}\leq 1-\delta+4(\frac{\delta}{9 r_{\mathrm{int}}\varepsilon}+1)r_{\mathrm{int}} \varepsilon =1-\frac{5}{9}\delta+4\varepsilon r_{\mathrm{int}}\leq 1-\frac{\delta}{2}$ and therefore $S_{k}^\varepsilon\subseteq Q^{\nu,+}_{r_{N_\varepsilon+1}}\setminus R_{1,\delta}^\nu$. This implies $L_k^\varepsilon\subset \mathrm{R}^+_{\delta,\varepsilon}$ for all $k \in \{1,\ldots,N_\varepsilon\}$. Using \eqref{ineq:card-estimate-Step3}, there exists $k_\varepsilon\in\{1,\dots,N_\varepsilon\}$ such that
\begin{align}\label{ineq:card-estimate-Strip-S4}
\begin{split}
\#((\mc{L}_\varepsilon(z_\varepsilon^+)\triangle X_\varepsilon)\cap L^\varepsilon_{k_\varepsilon})&\leq\frac{1}{N_\varepsilon}\sum^{N_\varepsilon}_{k=1}\#((\mc{L}_\varepsilon(z_\varepsilon^+)\triangle X_\varepsilon)\cap L_k^\varepsilon)
\\&\leq\frac{3}{N_\varepsilon}\#((\mc{L}_\varepsilon(z^+_\varepsilon)\triangle X_\varepsilon)\cap \mathrm{R}_{\delta,\varepsilon}^+)
\leq\frac{C}{\varepsilon^{d-1}\delta}E_\varepsilon\left(X_\varepsilon,Q^\nu\setminus R^\nu_{1,\frac{\delta}{2}}\right)\,,
\end{split}
\end{align}
where we used that $\varepsilon^{d-1}\delta\leq C N_\varepsilon\varepsilon^d$. The factor $3$ in the second inequality comes from the fact that each $S_k^\varepsilon$ is counted at most three times. We now set $D^\varepsilon:=Q^\nu_{r_{{k_\varepsilon}-1}}\cup(Q^{\nu,-}\setminus R_{1,\delta}^\nu)$. We define the configuration $Y_\varepsilon^+$ as follows:
\begin{equation}\label{def:Yeps-S4}
Y_\varepsilon:=\begin{cases}
\mc{L}_\varepsilon(z_\varepsilon^+) &\text{in  } (\mathrm{R}^+_{\delta,\varepsilon}\setminus Q^\nu_{r_{k_\varepsilon}})\cup \p_{\lambda\varepsilon}^+ Q^\nu\,, \\
\emptyset &\text{in  } (R_{1,\delta}^\nu\setminus Q^\nu_{r_{k_\varepsilon-1}}) \setminus\p_{\lambda\varepsilon}^+ Q^\nu\,, \\
X_\varepsilon\cap\mc{L}_\varepsilon(z_\varepsilon^+) &\text{in  } S^\varepsilon_{k_\varepsilon},\\
X_\varepsilon &\text{in  } D^\varepsilon\,.
\end{cases}
\end{equation}
We refer to Figure~\ref{fig:2} for an illustration of the construction. In $D^\varepsilon$ the configuration remains unchanged, and near the boundary of the upper half cube it coincides with the lattice $\mc{L}_\varepsilon(z_\varepsilon^+)$. In $S_{k_\varepsilon}^\varepsilon$, we geometrically interpolate  between $X_\varepsilon$ and $\mathcal{L}_\varepsilon(z^+_\varepsilon)$, i.e., $S_{k_\varepsilon}^\varepsilon$ can be understood as a transition layer. Finally, in the region close to the hyperplane and close to the boundary, we do not put any atoms. We note that our construction ensures $\vert y_1-y_2\vert\geq\varepsilon$ for all $y_1,y_2\in Y_\varepsilon^+, y_1\neq y_2$ and thus by \ref{H1}
\begin{equation}\label{ineq:Yeps-finite-energy-S4}
E_\varepsilon(Y_\varepsilon)<+\infty\,.
\end{equation}
Finally, we point out that $Y_\varepsilon\not\subset Q^\nu$ due to the definition of $\p_{\lambda\varepsilon} Q^\nu$, see \eqref{def:discrete-boundary} as well as Figure \ref{fig:2}. 
\begin{figure}[h]
\centering
\begin{tikzpicture}

\draw(0,0) rectangle(10,10);
\draw[fill=black!5](0.5,0.5) rectangle(9.5,9.5);
\draw[fill=black!5](1,1) rectangle(9,9);
\draw[fill=black!20](0,0) rectangle(10,5);
\draw[fill=black!10](0,5.25)--(0.25,5.25)--(0.25,5.5)--(9.75,5.5)--(9.75,5.25)--(10,5.25)--(10,4.75)--(9.75,4.75)--(9.75,4.5)--(0.25,4.5)--(0.25,4.75)--(0,4.75)--cycle;
\draw[fill=black!20](1.5,4.5) rectangle(8.5,8.5);
\draw[black!20](1.5,4.5)--(8.5,4.5);
\draw(5,8.4) node[above]{$S_{k_\varepsilon}^\varepsilon$};
\draw(5,9.3) node[above]{$(\mathrm{R}^+_{\delta,\varepsilon}\setminus Q^\nu_{r_{k_\varepsilon}})\cup\partial^+_{\lambda\varepsilon} Q^\nu$};
\draw[<->](-0.1,4.5)-- node[left]{$\delta$}(-0.1,5.5);
\draw[<->](0,5.6)node[above right]{$\sim\delta$}-- (1.5,5.6);
\draw[thick](0,5)--(10,5);
\draw[dashed](0,5.5)--(10,5.5);
\draw[dashed](0,4.5)--(10,4.5);
\draw(10,5) node[right]{$R^\nu_{1,\delta}$};
\draw[->](7,5)--(7,6) node[right]{$\nu$};
\draw(5,2) node[above]{$D^\varepsilon\cup\partial^-_{\lambda\varepsilon} Q^\nu$};
\end{tikzpicture}
\caption{Illustration of the different regions in the definition of $Y_\varepsilon$. The dark grey region is $D^\varepsilon\cup\partial^-_{\lambda\varepsilon} Q^\nu$, the grey region is $(R^\nu_{1,\delta}\setminus Q^\nu_{r_{k_\varepsilon-1}})\setminus(\partial^+_{\lambda\varepsilon} Q^\nu\cup\partial^-_{\lambda\varepsilon} Q^\nu)$, the light grey region is $S^\varepsilon_{k_\varepsilon}$ and the white region is $(\mathrm{R}^+_{\delta,\varepsilon}\setminus Q^\nu_{r_{k_\varepsilon}})\cup\partial^+_{\lambda\varepsilon} Q^\nu$. The dashed lines enclose $R^\nu_{1,\delta}$.}\label{fig:2}
\end{figure}
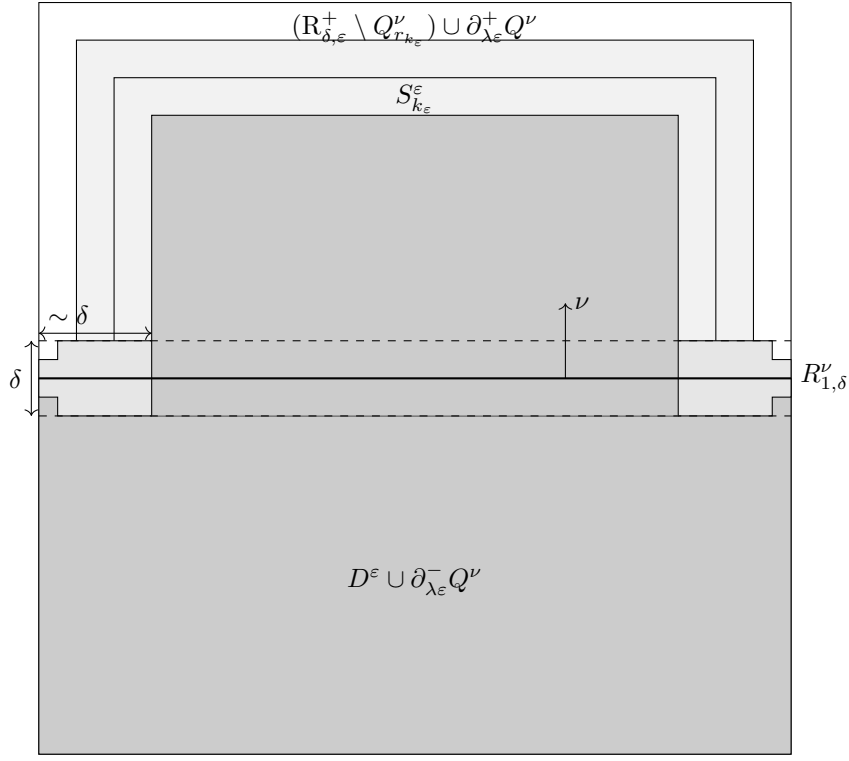
\noindent {\bf Step 5:} {\it Energy estimate.} In this step we show
\begin{equation}\label{ineq:energy-estimate-Step4}
\liminf_{\varepsilon\to0}E_\varepsilon(Y_\varepsilon,Q^\nu)\leq\liminf_{\varepsilon\to0}E_\varepsilon(X_\varepsilon,Q^\nu)+C\delta
\end{equation}
for some universal $C>0$. To achieve this we look at three distinct regions
\begin{equation}\label{def:different-regions}
A_1^\varepsilon:=\overline{(R_{1,\delta}^\nu\setminus Q^\nu_{r_{k_\varepsilon}-1})_{r_{\mathrm{int}}\varepsilon}}\,,\quad A_2^\varepsilon:=\overline{(S_{k_\varepsilon}^\varepsilon)_{r_{\mathrm{int}}\varepsilon}}\setminus A_1^\varepsilon\,,\quad A_3^\varepsilon:=Q^\nu\setminus(A_1^\varepsilon\cup A_2^\varepsilon)\,.
\end{equation}
{\it Energy estimate on }$A_1^\varepsilon$: We claim that there exists a $C>0$ such that
\begin{align}\label{ineq:A1eps-estimate-S5}
E_\varepsilon(Y_\varepsilon,A_1^\varepsilon)\leq C\delta\,.
\end{align}
By \eqref{def:Yeps-S4} we have $ Y_\varepsilon\cap(R_{1,\delta}^\nu\setminus Q^\nu_{r_{k_\varepsilon}-1})\subset (\mc{L}_\varepsilon(z^+_\varepsilon)\cap R_{1,\delta}^\nu\cap \p_{\lambda\varepsilon} Q^\nu)$. As $\mc{L}^d((R_{1,\delta}^\nu\cap\p_\varepsilon Q^\nu)_\varepsilon)\leq C\delta\varepsilon$, we use Lemma~\ref{lem:elementary-properties}(v) (applicable as $E_\varepsilon(Y_\varepsilon^+)<+\infty$) to get
\begin{equation}\label{ineq:card-estimate-A1}
\#\left(Y_\varepsilon\cap(R_{1,\delta}^\nu\setminus Q^\nu_{r_{k_\varepsilon}-1})\right)\leq C\delta \varepsilon^{1-d}\,.
\end{equation}
Note that $\mc{H}^{d-1}(\p(R^\nu_{1,\delta}\setminus Q^\nu_{r_{k_\varepsilon}-1}))\leq C\delta$. Hence, by Lemma \ref{lem:elementary-properties}(v), we obtain, using the definition of the Minkowski content,
\begin{align*}
\#((A_1^\varepsilon\cap Y_\varepsilon)\setminus(R^\nu_{1,\delta}\setminus Q^\nu_{r_{k_\varepsilon}-1}))&\leq C\varepsilon^{-d}\mc{L}^d\left(\left(\overline{(R^\nu_{1,\delta}\setminus Q^\nu_{r_{k_\varepsilon}-1})_{r_{\mathrm{int}}\varepsilon}}\setminus(R^\nu_{1,\delta}\setminus Q^\nu_{r_{k_\varepsilon}-1})\right)_\varepsilon\right)\\
&\leq C\varepsilon^{1-d}\mc{H}^{d-1}\left(\p(R_{1,\delta}^\nu\setminus Q^\nu_{r_{k_\varepsilon}-1})\right)\leq C\delta \varepsilon^{1-d}\,.
\end{align*}
This, along with \eqref{ineq:card-estimate-A1}, yields $\#(A_1^\varepsilon\cap Y_\varepsilon)\leq C\delta \varepsilon^{1-d}$ and thus \eqref{ineq:A1eps-estimate-S5} follows from \ref{H2}.\\
{\it Energy estimate on }$A_2^\varepsilon$: We prove that there exists a universal $C>0$ such that 
\begin{equation}\label{ineq:A2-estimate-S5}
E_\varepsilon(Y_\varepsilon,A_2^\varepsilon)\leq C \delta^{-1} E_\varepsilon\left(X_\varepsilon, Q^\nu\setminus R^\nu_{1,\frac\delta2}\right).
\end{equation}
By the definition of $L_{k_\varepsilon}^\varepsilon$ below \eqref{def:Sk-Step4} we have $\overline{(A_2^\varepsilon)_{r_{\mathrm{int}}\varepsilon}}\subset L_{k_\varepsilon}^\varepsilon$. We claim that there exists a constant $C>0$ (depending only on \ref{H2} and \ref{H6}) such that for all $x\in X_\varepsilon\cap Y_\varepsilon\cap A_2^\varepsilon$ we have
\begin{align}\label{ineq:YepsXeps-A2}
E_{\mathrm{cell}}^\varepsilon(x,Y_\varepsilon)\leq C\left( E_{\mathrm{cell}}^\varepsilon(x,X_\varepsilon)
+\#(\overline{B}_{r_{\mathrm{crys}}\varepsilon}(x)\cap (X_\varepsilon\setminus\mc{L}_\varepsilon(z^+_\varepsilon)))\right)\,.
\end{align}
In fact, if $\overline{B}_{r_{\mathrm{crys}}\varepsilon}(x)\cap(X_\varepsilon\setminus\mc{L}_\varepsilon(z^+_\varepsilon))\neq\emptyset$, then the right-hand side is larger or equal to $C$. As the left-hand side is less than or equal to $C$ by \ref{H2}, \eqref{ineq:YepsXeps-A2} follows in this case. Now, if $\overline{B}_{r_{\mathrm{crys}}\varepsilon}(x)\cap(X_\varepsilon\setminus\mc{L}_\varepsilon(z_\varepsilon^+))=\emptyset$ it holds $X_\varepsilon\cap\overline{B}_{r_{\mathrm{crys}}\varepsilon}(x)\subseteq\mc{L}_\varepsilon(z_\varepsilon^+)\cap\overline{B}_{r_{\mathrm{crys}}\varepsilon}(x)$. Now, using \eqref{def:Yeps-S4},  either 
\begin{align*}
E_{\mathrm{cell}}^\varepsilon(x,X_\varepsilon)=0 &\iff X_\varepsilon\cap\overline{B}_{r_{\mathrm{crys}}\varepsilon}(x)=\mc{L}_\varepsilon(z_\varepsilon^+)\cap\overline{B}_{r_{\mathrm{crys}}\varepsilon}(x) =  Y_\varepsilon\cap\overline{B}_{r_{\mathrm{crys}}\varepsilon}(x)\\ &\iff E_{\mathrm{cell}}^\varepsilon(x,Y_\varepsilon)=0
\end{align*}
 or $E_{\mathrm{cell}}^\varepsilon(x,X_\varepsilon)\geq c$ and, due to \ref{H2}, $E_{\mathrm{cell}}^\varepsilon(x,Y_\varepsilon)\leq C$. In both cases, we have \eqref{ineq:YepsXeps-A2}. Now, splitting the sum into $X_\varepsilon\cap Y_\varepsilon$ and $Y_\varepsilon\setminus X_\varepsilon$, using \ref{H2}, we obtain
\begin{align}\label{ineq:YepsA2}
E_\varepsilon(Y_\varepsilon, A_2^\varepsilon)&\leq C\varepsilon^{d-1}\#\{x\in A_2^\varepsilon\cap (Y_\varepsilon\setminus X_\varepsilon)\}+\varepsilon^{d-1}\sum_{x\in X_\varepsilon\cap Y_\varepsilon\cap A_2^\varepsilon}E_{\mathrm{cell}}^\varepsilon(x,Y_\varepsilon)\,.
\end{align}
Note that by \eqref{def:Yeps-S4}, $Y_\varepsilon\subseteq\mc{L}_\varepsilon(z_\varepsilon^+)\cap X_\varepsilon$ in $L^\varepsilon_{k_\varepsilon}$. Thus, by \eqref{ineq:card-estimate-Strip-S4}, we have
\begin{align}\label{ineq:card-estimateYepswithoutXeps}
\#\left\{x\in(Y_\varepsilon\cap L_{k_\varepsilon}^\varepsilon)\setminus X_\varepsilon\right\}\leq\#\left\{x\in(\mc{L}_\varepsilon(z_\varepsilon^+)\triangle X_\varepsilon)\cap L_{k_\varepsilon}^\varepsilon\right\}\leq C\delta^{-1} \varepsilon^{1-d} E_\varepsilon\left(X_\varepsilon,Q^\nu\setminus R_{1,\frac{\delta}{2}}^\nu\right)\,.
\end{align}
Additionally, using \eqref{ineq:card-estimate-Strip-S4} once more together with \eqref{ineq:YepsXeps-A2}, for $\delta>0$ small enough we obtain
\begin{align}\label{ineq:cellYepsA2}
\notag
\sum_{x\in X_\varepsilon\cap Y_\varepsilon\cap A_2^\varepsilon}E_{\mathrm{cell}}^\varepsilon(x,Y_\varepsilon) &\leq  C\varepsilon^{1-d} E_\varepsilon(X_\varepsilon,L_{k_\varepsilon}^\varepsilon) + C\sum_{x\in Y_\varepsilon\cap X_\varepsilon\cap A_2^\varepsilon}\left(\#\left(\overline{B}_{r_{\mathrm{crys}}\varepsilon}(x)\cap(X_\varepsilon\setminus\mc{L}_\varepsilon(z_\varepsilon^+))\right)\right)\\& \leq C\varepsilon^{1-d} E_\varepsilon(X_\varepsilon,L_{k_\varepsilon}^\varepsilon) + C \#\{x\in(\mc{L}_\varepsilon(z_\varepsilon^+)\triangle X_\varepsilon)\cap L_{k_\varepsilon}^\varepsilon\} \\&\leq C\delta^{-1} \varepsilon^{1-d} E_\varepsilon\left(X_\varepsilon, Q^\nu\setminus R^\nu_{1,\frac\delta2}\right)\,,
\notag
\end{align}
where the second inequality holds since $\vert x_1-x_2\vert\geq\varepsilon\ \forall x_1,x_2\in X_\varepsilon, x_1\neq x_2$ and $\overline{B}_{r_{\mathrm{crys}}\varepsilon}(x)\subset L_{k_\varepsilon}^\varepsilon$ for all $x\in A_2^\varepsilon$. Hence, each point in $(X_\varepsilon\setminus\mc{L}_\varepsilon(z_\varepsilon^+))\cap L_{k_\varepsilon}^\varepsilon$ is only accounted for at most a bounded number of times independent of $\varepsilon$ (see Lemma~\ref{lem:elementary-properties}(vi)). Now, using \eqref{ineq:YepsA2}-\eqref{ineq:cellYepsA2}, $L_{k_\varepsilon}^\varepsilon\subset Q^\nu\setminus R^\nu_{1,\delta}$ and Lemma~\ref{lem:elementary-properties}(iii) we obtain \eqref{ineq:A2-estimate-S5}.\\
{\it Energy estimate on }$A_3^\varepsilon$: We claim that
\begin{align}\label{ineq:Yeps-A3}
E_\varepsilon(Y_\varepsilon,A_3^\varepsilon)\leq E_\varepsilon(X_\varepsilon,Q^\nu)\,.
\end{align}
Recalling \eqref{def:different-regions}, each $x\in A_3^\varepsilon\cap Y_\varepsilon$ either lies in $T^\varepsilon:=(\mathrm{R}^+_{\delta,\varepsilon}\setminus Q^\nu_{r_{k_\varepsilon}})\cup(\p^+_{\lambda\varepsilon} Q^\nu\setminus R^\nu_{1,\delta})$ or in $D^\varepsilon$.\\
If $x\in A_3^\varepsilon\cap Y_\varepsilon\cap T^\varepsilon$, then also $\overline{B}_{r_{\mathrm{int}}\varepsilon}(x)\subset T^\varepsilon$ by the definition of $A_1^\varepsilon,A_2^\varepsilon$ and the boundary regions. Thus, \eqref{def:Yeps-S4} implies $E_{\mathrm{cell}}^\varepsilon(x,Y_\varepsilon)=0$, see \ref{H3}. If $x\in A_3^\varepsilon\cap Y_\varepsilon\cap D^\varepsilon$, then $X_\varepsilon\cap\overline{B}_{r_{\mathrm{int}}\varepsilon}(x)=Y_\varepsilon\cap\overline{B}_{r_{\mathrm{int}}\varepsilon}(x)$, which together with \ref{H4} implies $E_{\mathrm{cell}}^\varepsilon(x,X_\varepsilon)=E_{\mathrm{cell}}^\varepsilon(x,Y_\varepsilon)$. Thus, by \eqref{eq:scaled-energy-local}, Lemma~\ref{lem:elementary-properties}(iii),(iv), we obtain \eqref{ineq:Yeps-A3}. Indeed,
\begin{align*}
E_\varepsilon(Y_\varepsilon, A_3^\varepsilon)&=E_\varepsilon(Y_\varepsilon, A_3^\varepsilon\cap T^\varepsilon)+ E_\varepsilon(Y_\varepsilon, A_3^\varepsilon\cap D^\varepsilon)=E_\varepsilon(Y_\varepsilon, A_3^\varepsilon\cap D^\varepsilon)\leq E_\varepsilon(X_\varepsilon,Q^\nu)\,.
\end{align*}
Now, by Lemma~\ref{lem:elementary-properties}(iv),
$$E_\varepsilon(Y_\varepsilon,Q^\nu)=E_\varepsilon(Y_\varepsilon,A_1^\varepsilon)+E_\varepsilon(Y_\varepsilon,A_2^\varepsilon)+E_\varepsilon(Y_\varepsilon,A_3^\varepsilon)\,.$$
Thus, we obtain \eqref{ineq:Yeps-finite-energy-S4} by \eqref{eq:Step1-concentration}, \eqref{ineq:A1eps-estimate-S5}, \eqref{ineq:A2-estimate-S5}, and \eqref{ineq:Yeps-A3}. \\
\noindent {\bf Step 6:} {\it Conclusion.}  Repeating the cut-off construction in Step~4 on $Q^{\nu,-}$ for $z_\varepsilon^-$, we obtain a configuration $Y_\varepsilon$ such that $Y_\varepsilon\in \mathrm{Adm}_{\varepsilon,\lambda}^{(z_\varepsilon^+,z_\varepsilon^-)}(Q^\nu(y_\varepsilon))$  and
\begin{align}\label{ineq:Step6}
\liminf_{\varepsilon\to0}E_\varepsilon(Y_\varepsilon,Q^\nu(y_\varepsilon))\leq\liminf_{\varepsilon\to0}E_\varepsilon(X_\varepsilon,Q^\nu(y_\varepsilon))+C\delta
\end{align}
by \eqref{ineq:Yeps-finite-energy-S4}, where we included the center $y_\varepsilon$ in the notation for clarification. Since $z_\varepsilon^\pm\to z^\pm$ as $\varepsilon \to 0$ by Step $2$, we observe by \eqref{def:Phi} that
$$\liminf_{\varepsilon\to0}E_\varepsilon(Y_\varepsilon, Q^\nu(y_\varepsilon))\geq\Phi(z^+,z^-,\nu)\,.$$
By using \eqref{eq:optimal-sequence-psi}, \eqref{ineq:Yeps-finite-energy-S4}, and by passing to $\delta\to0$, we obtain the statement of the lemma in the case that $z_\varepsilon^\pm \neq \bf 0$.\\
\noindent {\bf Step 7:} {\it Adaptations for $z_\varepsilon^+ = \bf 0$.} Here, we describe the necessary adaptations to be done in case {\rm (b)} at the end of Step 2, i.e., when $z_\varepsilon^+=\textbf{0}$. \\
\noindent {\bf Step 3 for case {\rm (b)}:} {\it Cardinality estimate .} We prove that there exists a universal constant $C>0$ such that
\begin{align}\label{ineq:Card-estimate-Step3b}
\varepsilon^d\#(X_\varepsilon\cap \mathrm{R}^+_{\delta,\varepsilon})\leq CE_\varepsilon\left(X_\varepsilon,Q^\nu\setminus R^\nu_{1,\frac\delta2}\right)\,.
\end{align} 
Using \ref{H6}, \eqref{eq:cell-energy-scaled}, and \eqref{eq:scaled-energy-local}, \eqref{def:interpolation}, we obtain
\begin{align*}
\varepsilon^d\#(X_\varepsilon\cap \mathrm{R}^+_{\delta,\varepsilon}) &= \varepsilon^d \#\{x \in X_\varepsilon \cap  \mathrm{R}^+_{\delta,\varepsilon} \mid E_{\varepsilon}^\mathrm{cell}(x,X_\varepsilon) =0\} + \varepsilon^d \#\{x \in X_\varepsilon \cap  \mathrm{R}^+_{\delta,\varepsilon} \mid E_{\varepsilon}^\mathrm{cell}(x,X_\varepsilon) >0\} \\&\leq c\left( |\{u_\varepsilon \neq {\bf 0}\} \cap R^+_{\delta,\varepsilon}|  + \varepsilon E_\varepsilon(X_\varepsilon, \mathrm{R}^+_{\delta,\varepsilon})\right) \leq C E_\varepsilon(X_\varepsilon, Q^\nu \setminus R^{\nu}_{1,\frac{\delta}{2}})  \,,
\end{align*}
where we used that for $x\in X_\varepsilon$ with $E_{\varepsilon}^{\mathrm{cell}}(x,X_\varepsilon)=0$ we have $u_\varepsilon(y)\neq\textbf{0}$ for all  $y \in B_{\frac{\varepsilon}{2}}(x)$ and for all $x \in \mathrm{R}^+_{\delta,\varepsilon}$ we have $|B_{\frac{\varepsilon}{2}}(x)\cap \mathrm{R}^+_{\delta,\varepsilon}| \geq c \varepsilon^d$ for some constant depending only on $d$. This concludes Step 3 in case (b). \\
\noindent {\bf Step 4 for case {\rm (b)}:} {\it Cut-off construction.}  We construct $Y_\varepsilon$ such that $Y_\varepsilon=\emptyset$ on $\p_{\lambda\varepsilon}^+Q^\nu$. Again, set $N_\varepsilon=\lfloor\frac{\delta}{9r_{\mathrm{int}}\varepsilon}\rfloor$ and define $S_k^\varepsilon$ as in \eqref{def:Sk-Step4} as well as $L_k^\varepsilon=S_{k-1}^\varepsilon\cup S_k^\varepsilon\cup S_{k+1}^\varepsilon$. By averaging over $k$ and using \eqref{ineq:Card-estimate-Step3b}, there exists $k_\varepsilon\in\{1,\dots,N_\varepsilon\}$ such that
\begin{align}\label{ineq:averaging-Step4b}
\#(X_\varepsilon\cap L_{k_\varepsilon}^\varepsilon)&\leq\frac{1}{N_\varepsilon}\sum_{k=1}^{N_\varepsilon}\#(X_\varepsilon\cap L_{k}^\varepsilon)\leq\frac{3}{N_\varepsilon}\#(X_\varepsilon\cap \mathrm{R}_{\delta,\varepsilon}^+)\leq\frac{C}{\varepsilon^{d-1}\delta}E_\varepsilon\left(X_\varepsilon,Q^\nu\setminus R_{1,\frac\delta2}^\nu\right)\,,
\end{align}
where we again use that each strip $S_k^\varepsilon$ is counted at most three times. We define
\begin{align}\label{eq:Yeps-Step3b}
Y_\varepsilon:=\begin{cases}\emptyset &\text{ in }(\mathrm{R}_{\delta,\varepsilon}^+\cup R^\nu_{1,\delta})\setminus Q^\nu_{r_{k_\varepsilon}-1} \cup \partial^+_{\lambda\varepsilon}Q^\nu\,,\\ X_\varepsilon & \text{ otherwise.}\end{cases}
\end{align}
Note that, due to $E_\varepsilon(X_\varepsilon)<+\infty$, we have $E_\varepsilon(Y_\varepsilon)<+\infty$. \\
\noindent {\bf Step 5 for case (b):} {\it Energy estimate.} 
We split the estimate into the three sets $A_1^\varepsilon, A_2^\varepsilon, A_3^\varepsilon$ defined in \eqref{def:different-regions}.\\
{\it Energy estimate on }$A_1^\varepsilon$:  We claim that there exists a $C>0$ such that
\begin{align}\label{ineq:A1-Step3b}
E_\varepsilon(Y_\varepsilon,A_1^\varepsilon)\leq C\delta\,.
\end{align}
By \eqref{eq:Yeps-Step3b} we have $Y_\varepsilon\cap(R_{1,\delta}^\nu\setminus Q^\nu_{r_{k_\varepsilon}})=\emptyset$. Furthermore, as $R^\nu_{1,\delta}\setminus Q^\nu_{r_{k_\varepsilon}-1}$ satisfies $\mc{H}^{d-1}(\p(R^\nu_{1,\delta}\setminus Q^\nu_{r_{k_\varepsilon}-1}))\leq C\delta$ and $Y_\varepsilon$ has finite energy, we obtain by Lemma~\ref{lem:elementary-properties}(v)
\begin{align*}
\#(A_1^\varepsilon\cap Y_\varepsilon)=\#\left((A_1^\varepsilon\setminus(R^\nu_{1,\delta}\setminus Q^\nu_{r_{k_\varepsilon}}))\cap Y_\varepsilon\right)&\leq C\varepsilon^{-d}\mc{L}^d\left((A_1^\varepsilon\setminus(R^\nu_{1,\delta}\setminus Q^\nu_{r_{k_\varepsilon}}))_{\varepsilon}\right)\\
&\leq C\varepsilon^{-(d-1)}\mc{H}^{d-1}\left(\p(R^\nu_{1,\delta}\setminus Q^\nu_{r_{k_\varepsilon}-1})\right)\leq C\delta \varepsilon^{1-d}\,.
\end{align*}
Then, \eqref{ineq:A1-Step3b} follows from \ref{H2}.\\
{\it Energy estimate on }$A_2^\varepsilon$:  We claim that there exists a $C>0$ such that
\begin{align}\label{ineq:A2-Step3b}
E_\varepsilon(Y_\varepsilon,A_2^\varepsilon)\leq\frac{C}{\delta}E_\varepsilon\left(X_\varepsilon,Q^\nu\setminus R^\nu_{1,\frac\delta2}\right)\,.
\end{align}
As $x\in Y_\varepsilon\cap A_2^\varepsilon$ implies $x\in X_\varepsilon\cap L_{k_\varepsilon}^\varepsilon$, \eqref{ineq:A2-Step3b} follows from \eqref{ineq:averaging-Step4b} and \ref{H2}.\\
{\it Energy estimate on }$A_3^\varepsilon$: We observe that
\begin{align}\label{ineq:A3-Step3b}
E_\varepsilon(Y_\varepsilon,A_3^\varepsilon)\leq E_\varepsilon(X_\varepsilon,Q^\nu)\,.
\end{align}
Indeed, if $x\in Y_\varepsilon\cap(Q^\nu\setminus(A_1^\varepsilon\cup A_2^\varepsilon))$,  then $X_\varepsilon\cap\overline{B}_{r_{\mathrm{int}}\varepsilon}(x)=Y_\varepsilon\cap \overline{B}_{r_{\mathrm{int}}\varepsilon}(x)$, which implies that $E^\varepsilon_{\mathrm{cell}}(x,X_\varepsilon)=E^\varepsilon_{\mathrm{cell}}(x,Y_\varepsilon)$. Thus, \eqref{ineq:A3-Step3b} follows by the definition of the energy and Lemma~\ref{lem:elementary-properties}(iii).\\
Using \eqref{ineq:A1-Step3b}-\eqref{ineq:A3-Step3b} together with \eqref{eq:Step1-concentration}, we obtain
\begin{align*}
\liminf_{\varepsilon\to0}E_\varepsilon(Y_\varepsilon,Q^\nu)\leq\liminf_{\varepsilon\to0}E_\varepsilon(X_\varepsilon,Q^\nu)+C\delta\,,
\end{align*}
which is the analogue to \eqref{ineq:energy-estimate-Step4}. As Step 6 remains unchanged, this concludes the proof.
\end{proof}
\section{Reduction of the Problem to Subsets of Two Lattices}\label{sec:reduction-two-lattices}
From now on, it will be convenient to rewrite the problem with lattice spacing $1$ and cubes of size $T$ with $T\to +\infty$. This is a rescaling argument using Lemma~\ref{lem:elementary-properties}(vii). See \eqref{def:Phi-equivalent} for the rescaled problem. The goal of this section is to restrict the set of competitors for the cell problem with prescribed boundary values in a $\lambda$-neighborhood of the boundary, see~\eqref{def:discrete-boundary}. This will be the consequence of Lemma~\ref{lem:separation-2-lattices} and Lemma~\ref{lem:independence-of-bdry-layer}. In the proof of Lemma~\ref{lem:separation-2-lattices} we use the concept of connectedness in a graph theoretical sense. To this end, we recall the definition of neighboring points, see \eqref{def:neighbourhood} for $\varepsilon=1$, i.e.,
\begin{align*}
\mathcal{N}(x) =\{y \in X\setminus\{x\}\mid |x-y|\leq r_{\mathrm{int}}\}\,.
\end{align*}
 Now, for a configuration $X$ we define a path $p=(x_0,\dots,x_N)$ as a collection of atoms with $x_i\in X$ for all $i=0,\dots,N$ and $x_i\in\mc{N}(x_{i-1})$ for $i=1,\dots, N$. Two atoms $x,y\in X$ are called connected if there exists a path $p=(x_0,\dots,x_N)$ in $X$ with $x_0=x,x_N=y$. A connected component of $\tilde{X}\subset X$ is a subset, maximal with respect to set inclusion, and such that for all $x,y\in\tilde{X}$ there is a path $p=(x_0,\dots,x_N)$ with $x_0=x,x_N=y$ and $x_i\in\tilde{X}$ for all $i=1,\dots,N$.
\begin{lem}\label{lem:separation-2-lattices}
Let $E_{\mathrm{cell}}$ satisfy {\rm \ref{H1}}--{\rm \ref{H10}}. Let $z^+,z^-\in\mc{Z},\nu\in\mathbb{S}^{d-1},x_0\in\R^d$, $\lambda> 8 r_{\mathrm{int}}$, and $T>0$.  Let $X\subset\R^d$ be a competitor of
\begin{align}\label{lem-eq:minimizer}
\inf\left\{E_1(X,Q_T^\nu(x_0))\mid X \in \mathrm{Adm}_{1,\lambda}^{z^+,z^-}(Q^\nu_T(x_0))\right\}\EEE
\end{align} 
and set $\hat \lambda = \lambda-2r_{\mathrm{int}}$. Then, there exist configurations $X^\pm\subset X$ with the following properties: 
\begin{enumerate}
\item[{\rm (i)}](Subset of two Lattices) It holds $X^\pm\subset\mc{L}(z^\pm)$.
\item[{\rm (ii)}](Separation) It holds $\dist(X^+,X^-)>r_{\mathrm{int}}$.
\item[{\rm (iii)}](Admissibility) It holds $X^+ \in \mathrm{Adm}^{(z^+,\textbf{0})}_{1,\hat\lambda}(Q^\nu_T(x_0))$ and $X^- \in \mathrm{Adm}^{(\textbf{0},z^-)}_{1,\hat\lambda}(Q^\nu_T(x_0))$. In particular,
\begin{align*}
\inf\EEE\big\{E_1(X,Q_T^\nu(x_0))\mid X \in \mathrm{Adm}_{1,\hat \lambda}^{(z^+,\textbf{0})}(Q^\nu_T(x_0))\big\} \leq E_1(X^+,Q^\nu_T) 
\end{align*}
and 
\begin{align*}
\inf\EEE\big\{E_1(X,Q_T^\nu(x_0))\mid X\in \mathrm{Adm}_{1,\hat \lambda}^{(\textbf{0},z^-)}(Q^\nu_T(x_0))\big\}\leq E_1(X^-,Q^\nu_T)\,.
\end{align*}
\item[{\rm (iv)}](Energy bound) It holds
\begin{align*}
E_1(X^+,Q^\nu_T(x_0))+E_1(X^-,Q^\nu_T(x_0))\leq E_1(X,Q^\nu_T(x_0))\,.
\end{align*}
\end{enumerate}
\end{lem}
\begin{proof}
Let $X$ be a  competitor  of \eqref{lem-eq:minimizer}. In order to simplify notation, we omit the center and just write $Q^\nu_T$ instead of $Q^\nu_T(x_0)$. Without loss of generality, we assume that $X\subset(Q_T^\nu)_{2r_{\mathrm{int}}}$. The idea of the proof is to successively remove atoms of $X$ which lower the energy and eventually lead to configurations that fulfill the conditions stated in Lemma~\ref{lem:separation-2-lattices}. Notice that we cannot ensure that we modify our configuration $X$ at any step in the construction described below. Thus, we cannot guarantee strict inequality in (iv). \\ 
\noindent {\bf Step 1:} {\it Removing high energy inducing atoms.} We claim that, up to removing some atoms which only decreases the energy, we can assume that $X \in \mathrm{Adm}_{1,\hat \lambda}^{(z^+,z^-)}(Q^\nu_T)$ and for all $x \in X \cap Q^\nu_T \setminus\p_{\hat{\lambda}}Q^\nu_T$ the following hold true:
\begin{itemize}
\item[(1)] $\#\mc{N}(x)\geq d_{\mathrm{unique}}$,
\item[(2)] There exists a unique $z\in\mc{Z}$, denoted by $z(x)$, such that $\{x\}\cup\mc{N}(x)\subseteq\mc{L}(z)$,
\item[(3)] For all $x,y \in  X \cap Q^\nu_T \setminus\p_{\hat{\lambda}}Q^\nu_T$ such that $y \in \mathcal{N}(x)$ we have $z(x) =z(y)$.
\end{itemize}
First, assume that there exists $x \in X \cap Q^\nu_T \setminus\p_{\hat{\lambda}}Q^\nu_T $ such that $\#\mc{N}(x)< d_{\mathrm{unique}}$. Then, by \ref{H8}, noting that $x \cup \mathcal{N}(x) \subset \overline{B}_{r_{\mathrm{int}}}(x) \subset Q_T^\nu$, and using Lemma~\ref{lem:elementary-properties}(ix), we have that
\begin{align*}
E_1(X \setminus \{x\}, Q^\nu_T) \leq E_1(X, Q^\nu_T) \,.
\end{align*} 
Thus, we can assume that (1) holds true for all $x \in X \cap Q^\nu_T \setminus\p_{\hat{\lambda}}Q^\nu_T $. Now, if there is $x \in X \cap Q^\nu_T \setminus\p_{\hat{\lambda}}Q^\nu_T$ such that for all $z \in \mathcal{Z}$ we have $x \cup \mathcal{N}(x) \nsubseteq \mathcal{L}(z)$, then,  by using \ref{H9}, noting that $x \cup \mathcal{N}(x) \subset \overline{B}_{r_{\mathrm{int}}}(x) \subset Q_T^\nu$, and using Lemma~\ref{lem:elementary-properties}(ix), we obtain 
\begin{align*}
E_1(X \setminus \{x\}, Q^\nu_T) \leq E_1(X, Q^\nu_T) \,.
\end{align*} 
This shows that we can assume {\rm (2)}. Finally, now assume that there exist  $x,y \in  X \cap Q^\nu_T \setminus\p_{\hat{\lambda}}Q^\nu_T$ such that $y \in \mathcal{N}(x)$ and $z(x) \neq z(y)$. Then, by \ref{H10} and Lemma~\ref{lem:elementary-properties}(ix), it holds
 \begin{align*}
 E_1(X\setminus\{x\},Q^\nu_T)\leq E_1(X,Q^\nu_T)\,.
 \end{align*}
This shows that we can assume {\rm (3)}.
Note that, as $X \in \mathrm{Adm}_{1,\lambda}^{(z^+,z^-)}(Q^\nu_T)$ and we removed only atoms in $Q_T^\nu\setminus\p_{\hat \lambda}Q^\nu_T$ we have that the newly obtained configuration, again denoted by $X$, satisfies $X \in \mathrm{Adm}_{1,\hat\lambda}^{(z^+,z^-)}(Q^\nu_T)$. \\
\noindent {\bf Step 2:} {\it Separating $\p^+_{\hat{\lambda}}Q^\nu_T$ from $\p^-_{\hat{\lambda}}Q^\nu_T$.} Next, we show that we can assume that there does not exist a path connecting $\p^+_{\hat{\lambda}}Q^\nu_T$ with $\p^-_{\hat{\lambda}}Q^\nu_T$. Assume there is one such path $p$ of minimal length, i.e., $p=(x_1,\dots,x_N)$ in $X$ with
\begin{align}\label{eq:path-properties-S2}
\begin{split}
x_1 \in X\cap \p^+_{\hat{\lambda}}Q^\nu_T\,,\, & x_N \in X\cap \p^+_{\hat{\lambda}}Q^\nu_T \\&\quad \text{and } x_i \notin X \cap (\p^+_{\hat{\lambda}}Q^\nu_T\cup\p^-_{\hat{\lambda}}Q^\nu_T) \text{ for all } i\in\{2,\dots,N-1\}\,.
\end{split}
\end{align}
Since we started with $X\in\mathrm{Adm}_{1,\lambda}^{z^+,z^-}(Q^\nu_T)$ it holds $\overline{B}_{r_{\mathrm{int}}}(x_2)\cap X\subset\p^+_\lambda Q^\nu_T\cap X\subset\mc{L}(z^+)$ and $\overline{B}_{r_{\mathrm{int}}}(x_{N-1})\cap X\subset\p^-_\lambda Q^\nu_T\cap X\subset\mc{L}(z^-)$. By Step~1(2) this implies $z(x_2)=z^+, z(x_{N-1})=z^-$. Now succesively using Step~1(3) along the path this implies $z^+=z(x_2)=\dots=z(x_{N-1})=z^-$, which contradicts $z^+\neq z^-$ and shows that such a path cannot exist. In particular, this shows that any atom in $X$ is connected to at most one of the boundary regions.\\
\noindent {\bf Step 3:} {\it Reduction to a subset of $\mc{L}(z^+)\cup\mc{L}(z^-)$.} We claim that we can assume that $X\subset\mc{L}(z^+)\cup\mc{L}(z^-)$. To this end, assume there exists $x\in X\setminus(\mc{L}(z^+)\cup\mc{L}(z^-))$ and notice that, as $X \in\mathrm{Adm}_{1,\hat \lambda}^{(z^+,z^-)}(Q^\nu_T)$,  necessarily this implies $\overline{B}_{r_{\mathrm{int}}}(x)\subset Q^\nu_T$. By Step~1(1),(2) there exists $z(x)\in\mc{Z} \setminus \{z^+,z^-\}$ such that $\{x\}\cup\mc{N}(x)\subset\mc{L}(z(x))$. We claim that such an $x$ is not connected to $X\cap\p^\pm_{\hat{\lambda}}Q^\nu_T$. Indeed, assume the contrary, i.e., there exists a path $p$ of minimal length $p=(x,x_1,\dots,x_N) \subset X$ such that $x_N\in X \cap (\p^+_{\hat \lambda}Q^\nu_T\cup \p^-_{\hat \lambda}Q^\nu_T)$ and $x_i\notin\p^+_{\hat \lambda}Q^\nu_T\cup  \p^-_{\hat \lambda}Q^\nu_T$ for all $i\in \{1,\ldots,N-1\}$. By Step~1(1) we can assume that $\#\mc{N}(x_i)\geq d_{\mathrm{unique}}$ for all $i\in\{1,\dots,N-1\}$. But then, arguing as in Step~2, this yields a contradiction. Using \ref{H2} and Lemma~\ref{lem:elementary-properties}(viii),(ix), we can remove any connected component not  connected to $X \cap (\p^+_{\hat \lambda}Q^\nu_T\cup \p^-_{\hat \lambda}Q^\nu_T)$  without increasing the energy. Clearly, after removing these components, we still have $X\in\mathrm{Adm}_{1,\hat\lambda}^{(z^+,z^-)}(Q^\nu_T)$. \\
\noindent {\bf Step 4:} {\it Conclusion.} We define $X^\pm$  as the union of all maximal connected components containing atoms in $X \cap \p^\pm_{\hat \lambda}Q^\nu_T$. It remains to verify that $X^\pm$ satisfy {\rm (i)}-{\rm (iv)}. Statement {\rm (i)} is a consequence of Step~2 and Step~3. Indeed, due to Step~3, we have $X^\pm \subset \mathcal{L}(z^\pm)$. Clearly $X^\pm$ contains points in $\mathcal{L}(z^\pm)$ (namely all the points in $X \cap \partial^\pm_{\hat \lambda}Q^\nu_T$) and  due to Step~1,2, and the choice of $X^\pm$ {\rm (ii)} is satisfied. Statement {\rm (iii)} is satisfied due to the fact that $X \in \mathrm{Adm}_{1,\lambda}^{(z^+,z^-)}(Q^\nu_T) $ and $X^\pm$ where chosen as maximal connected components of $X$ containing  $X \cap \p^\pm_{\hat \lambda}Q^\nu_T$. Indeed,  we have $X^\pm \cap \partial^c_{\lambda}Q^\nu_T \subset X \cap \partial^c_{\lambda}Q^\nu_T =\emptyset$ and $\mathcal{L}(z^\pm)\cap \p^\pm_{\hat \lambda}Q^\nu_T =X \cap \p^\pm_{\hat \lambda}Q^\nu_T = X^\pm \cap \p^\pm_{\hat \lambda}Q^\nu_T$. Finally, {\rm (iv)} follows from Step~1-3, {\rm (ii)}, Lemma~\ref{lem:elementary-properties}(viii),(ix), and from the choice of $X^\pm$. 
\end{proof}
We will now show that the asymptotic cell problem is independent of the choice of the thickness $\lambda$ of the boundary layer, provided the thickness is big enough to ensure that the boundary values are attained and small enough in order to give an energy contribution that vanishes as the side-lengths of the cube tend to infinity.

\begin{lem}\label{lem:independence-of-bdry-layer}
Let $E_{\mathrm{cell}}$ satisfy {\rm \ref{H1}}--{\rm\ref{H5}}. Let $z^\pm\in\mc{Z}$, $\nu\in\mathbb{S}^{d-1}$, and $ 6r_{\mathrm{int}} <\lambda_1 <\lambda_2 $ be given. Then, it holds
\begin{align*}
\min\big\{\liminf_{T\to\infty}\frac{1}{T^{d-1}}\inf\big\{E_1(X,Q^\nu_T(y_T))&\mid y_T\in\R^d, X \in \mathrm{Adm}_{1,\lambda_1}^{(z^+_T,z^-_T)}(Q_T^\nu(y_T))\big\}\mid\\&\qquad\qquad \{z^\pm_T\}_T\subset\mc{Z}\text{ with }z^\pm_T\to z^\pm\big\}\\=\min\big\{\liminf_{T\to\infty}\frac{1}{T^{d-1}}\inf\big\{E_1(X,Q^\nu_T(y_T))&\mid y_T\in\R^d, X \in \mathrm{Adm}_{1,\lambda_2}^{(z^+_T,z^-_T)}(Q_T^\nu(y_T))\big\}\mid\\&\qquad\qquad  \{z^\pm_T\}_T\subset\mc{Z}\text{ with }z^\pm_T\to z^\pm\big\}\,.
\end{align*}
\end{lem}
\begin{proof} Note that, as $\mathrm{Adm}_{1,\lambda_2}^{(z^+,z^-)}(Q_T^\nu(y)) \subset  \mathrm{Adm}_{1,\lambda_1}^{(z^+,z^-)}(Q_T^\nu(y))$ for all $ y\in\R^d, z^\pm \in \mathcal{Z},T >0$, and $\lambda_1 <\lambda_2$, we only need to show 
\begin{align}\label{ineq-proof:bdry-layer}
\begin{split}
\min\big\{\liminf_{T\to\infty}\frac{1}{T^{d-1}}\inf\big\{E_1(X,Q^\nu_T(y_T))&\mid y_T\in\R^d, X \in \mathrm{Adm}_{1,\lambda_1}^{(z^+_T,z^-_T)}(Q_T^\nu(y_T))\big\}\mid\\&\qquad\qquad \{z^\pm_T\}_T\subset\mc{Z}\text{ with }z^\pm_T\to z^\pm\big\}\\\leq\min\big\{\liminf_{T\to\infty}\frac{1}{T^{d-1}}\inf\big\{E_1(X,Q^\nu_T(y_T))&\mid y_T\in\R^d, X \in \mathrm{Adm}_{1,\lambda_2}^{(z^+_T,z^-_T)}(Q_T^\nu(y_T))\big\}\mid\\&\qquad\qquad  \{z^\pm_T\}_T\subset\mc{Z}\text{ with }z^\pm_T\to z^\pm\big\}\,.
\end{split}
\end{align}
In order to prove this fix $z^\pm \in \mathcal{Z}$, $y \in \mathbb{R}^d$, and $T >0$ and let $X \in \mathrm{Adm}_{1,\lambda_1}^{(z^+,z^-)}(Q_T^\nu(y))$ be such that $E_1(X) <+\infty$. We define $T_{\lambda_2} = T+\lambda_2$ and define $X_{\lambda_2}$ by
\begin{align*}
X_{\lambda_2} = \begin{cases} X &\text{in } Q_T(y)\,,\\
\mathcal{L}(z^\pm) &\text{in } \{x \in \mathbb{R}^d \setminus Q_T(y) \mid \pm \langle x-y,\nu\rangle \geq r_{\mathrm{int}} \}\,, \\
\emptyset &\text{otherwise.}
\end{cases}
\end{align*}
Clearly, $X_{\lambda_2} \in \mathrm{Adm}_{1,\lambda_2}^{(z^+,z^-)}(Q_{T_{\lambda_2}}^\nu(y))$ and therefore
\begin{align}\label{ineq-proof:bdry-layer-2}
\inf\big\{E_1(X,Q^\nu_{T_{\lambda_2}}(y))\mid y\in\R^d, X \in \mathrm{Adm}_{1,\lambda_2}^{(z^+,z^-)}(Q_{T_{\lambda_2}}^\nu(y))\big\} \leq E_1(X_{\lambda_2},Q^\nu_{T_{\lambda_2}}(y))\,.
\end{align}
Using Lemma~\ref{lem:elementary-properties}(iv),(x), and $X \in \mathrm{Adm}_{1,\lambda_1}^{(z^+,z^-)}(Q_T^\nu(y))$ we have
\begin{align}\label{eq:additive-decomp-lem62}
\begin{split}
E_1(X_{{\lambda_2}},Q^\nu_{T_{\lambda_2}}(y))&=E_1(X_{{\lambda_2}},Q^\nu_{T}(y))+E_1(X_{{\lambda_2}},Q^\nu_{T_{\lambda_2}}(y)\setminus Q^\nu_T(y))\\& =E_1(X,Q^\nu_{T}(y))+E_1(X_{{\lambda_2}},Q^\nu_{T_{\lambda_2}}(y)\setminus Q^\nu_T(y))\,.
\end{split}
\end{align}
To estimate the second summand, define  $ A_T:= Q^\nu_{T_{\lambda_2}}(y)\setminus Q^\nu_T(y)\cap\{x\mid\vert \langle x-y,\nu\rangle\vert \leq 2 \mathrm{r}_{\mathrm{int}}\}$. By definition of $X_{\lambda_2}$ and \ref{H3} we have $E_1^{\mathrm{cell}}(x,X_{\lambda_2})=0$  for all $x \in (X_{\lambda_2} \cap  Q^\nu_{T_{\lambda_2}}(y)\setminus Q^\nu_T(y)) \setminus A_T$ and $\mc{L}^d((A_T)_1)\leq Cr_{\mathrm{int}}\lambda_2 T^{d-2}$ for some dimensional constant $C>0$. Thus, by Lemma~\ref{lem:elementary-properties}(v) and \ref{H2} there holds
\begin{align}\label{ineq:estimate-X-lamda2outside}
E_1(X_{{\lambda_2}},Q^\nu_{T_{\lambda_2}}(y)\setminus Q^\nu_T(y))\leq Cr_{\mathrm{int}}\lambda_2 T^{d-2}\,.
\end{align}
Using \eqref{ineq-proof:bdry-layer-2}, \eqref{eq:additive-decomp-lem62}, \eqref{ineq:estimate-X-lamda2outside}, dividing  by $T^{d-1}$ and taking $\liminf$ as $T$ tends to infinity yields \eqref{ineq-proof:bdry-layer}  by the arbitrariness of $z^\pm\in \mathcal{Z},y\in\R^d,T>0$ and $X \in \mathrm{Adm}_{1,\lambda_1}^{(z^+,z^-)}(Q_T^\nu(y))$.
\end{proof} 
\section[Cell Formula II]{Cell Formula II: Vacuum energy}\label{sec:vacuum-energy}
In this section, we show that the surface energy density decomposes into twice the vacuum energy. To this end, we prove some preparatory lemmas that will allow us to conclude the proof. In this section, it is convenient to rescale the problem in order to have a constant lattice spacing (normalized to be $1$), whereas the size of the cubes grows. Using our rescaling property in Lemma~\ref{lem:elementary-properties}(vii) for $T>0$ we can rewrite \eqref{def:Phi} as
\begin{align} \label{def:Phi-equivalent}
\begin{split}
\Phi(z^+,z^-,\nu):=\min\big\{&\liminf_{T\to\infty}\frac{1}{T^{d-1}}\inf\big\{E_1(X,Q^\nu_T(y_T))\mid \\& \quad y_T\in\R^d,X \in \mathrm{Adm}_{1,\lambda}^{(z^+_T,z^-_T)}(Q^\nu_T(y_T))\big\}\mid
\{z^\pm_T\}_T\subset\mc{Z}\text{ with }z^\pm_T\to z^\pm\big\}\,
\end{split}
\end{align} 
and we set
\begin{align}\label{def:Phi-vac}
\Phi_{\mathrm{vac}}(z,\nu):=\Phi(z,\textbf{0},\nu)\,.
\end{align}
Recall that for the vacuum energy, using that $z^-_T\to\textbf{0}\lra z^-_T=\textbf{0}$ for $T$ large enough, it suffices to consider converging sequences in the upper half cube, while on the lower half we can assume vacuum boundary conditions for $T$ large enough. 

 In the following, we derive several properties of the asymptotic cell formula determining the energy density for the solid-vacuum transition. More precisely, we show that the centers of the cubes defining the cell-problem can be chosen arbitrarily (Lemma~\ref{lem:center-independence}) and the asymptotic-cell problem is independent of the orientation of the cube provided that one side is normal to $\nu$ (Lemma~\ref{lem:independence-of-orientation}). This allows us to conclude that we can replace converging boundary values by fixed ones (Lemma~\ref{lem:fixed-bdryvalues}).

\subsection{Independence of cube centers}\label{subsec:centers}
The goal of this subsection is to show the following:
\begin{lem}\label{lem:center-independence}
Let $E_{\mathrm{cell}}$ satisfy {\rm \ref{H1}}--{\rm \ref{H10}}. Let $\nu\in\mathbb{S}^{d-1},z\in\mc{Z}, \lambda > 8r_{\mathrm{int}}$, and $\{x_T\}_T \subset \mathbb{R}^d$ be any sequence of centers. Then, there holds
\begin{align}\label{eq-lem:center-independence}
\begin{split}
\Phi_{\mathrm{vac}}(z,\nu)=\min\big\{\liminf_{T\to\infty}\frac{1}{T^{d-1}}\inf\big\{E_1(X,Q^\nu_T(x_T))\mid &X \in \mathrm{Adm}_{1,\lambda}^{(z_T,\textbf{0})}(Q_T^\nu(x_T))\big\}\mid\\&\{z_T\}_T\subset\mc{Z}\text{ with }z_T\to z\big\}\,.
\end{split}
\end{align}
Furthermore, the choice $z_T \to z$ as $T\to +\infty$ can be chosen independently of $\{x_{T}\}_T$ and the convergence is uniform with respect to $x_T$.
\end{lem}
\begin{proof}
Let $\{x_T\}_T \subset \mathbb{R}^d$ be any sequence of centers. Notice that we always have
\begin{align*}
\Phi_{\mathrm{vac}}(z,\nu)\leq\min\big\{\liminf_{T\to\infty}\frac{1}{T^{d-1}}\inf\big\{E_1(X,Q^\nu_T(x_T))\mid &X \in \mathrm{Adm}_{1,\lambda}^{(z_T,\textbf{0})}(Q_T^\nu(x_T))\big\}\mid\\&\{z_T\}_T\subset\mc{Z}\text{ with }z_T\to z\big\}\,.
\end{align*}
so that it suffices to prove
\begin{align}\label{ineq-lemproof:center-independence}
\begin{split}
\Phi_{\mathrm{vac}}(z,\nu)\geq\min\big\{\liminf_{T\to\infty}\frac{1}{T^{d-1}}\inf\big\{E_1(X,Q^\nu_T(x_T))\mid &X \in \mathrm{Adm}_{1,\lambda}^{(z_T,\textbf{0})}(Q_T^\nu(x_T))\big\}\mid\\&\{z_T\}_T\subset\mc{Z}\text{ with }z_T\to z\big\}\,.
\end{split}
\end{align}
To this end, let $\{X_T\}_T,\{y_T\}_T,\{z_T\}_T$ be an optimal sequence for $\Phi_{\mathrm{vac}}(z,\nu)$, namely a sequence such that for $X_T \in \mathrm{Adm}_{1,\lambda}^{(z_T,\textbf{0})}(Q_T^\nu(y_T))$, up to a subsequence (not relabeled), there holds
\begin{align}\label{ineq:optimal-sequence-centerlem}
\lim_{T\to\infty}\frac{1}{T^{d-1}}E_1(X_T,Q^\nu_T(y_T))=\Phi_{\mathrm{vac}}(z,\nu)\,.
\end{align} 
Due to Lemma~\ref{lem:separation-2-lattices} and Lemma~\ref{lem:independence-of-bdry-layer} the sequence $X_T$ can be chosen such that $X_T \subset \mathcal{L}(z_T)$ for all $T$. Using \ref{L3}, there exists $k_T \in  L\mathbb{Z}^d$ and a constant $C_{\mathcal{L}}>0$ (not depending on $T$) such that for $\overline{y}_T:=y_T+k_T$ we have  $|x_T-\overline{y}_T| \leq C_{\mathcal{L}}$. Now define
\begin{align}\label{def:Xhat-lem-independence}
\hat{X}_T := \begin{cases} X_T+k_T & \text{in } Q^\nu_T(\overline{y}_T)\,,\\
\mathcal{L}(z_T) &\text{in } \{x \in \mathbb{R}^d \mid \langle x-x_T,\nu\rangle \geq r_{\mathrm{int}}\}\setminus Q^\nu_T(\overline{y}_T) \,,\\
\emptyset &\text{otherwise.}
\end{cases} 
\end{align}
We observe that for $\hat{T} = T + \sqrt{d} C_{\mathcal{L}} + \lambda$ there holds $\hat{X}_T \in \mathrm{Adm}_{1,\lambda}^{(z_T,\textbf{0})}(Q_{\hat{T}}^\nu(x_T))$ and clearly, as $X_T \subset \mathcal{L}(z_T)$, we have $\hat{X}_T \subset \mathcal{L}(z_T)$. Therefore,
\begin{align*}
\inf\big\{E_1(X,Q^\nu_{\hat{T}}(x_T))\mid &X \in \mathrm{Adm}_{1,\lambda}^{(z_T,\textbf{0})}(Q_T^\nu(x_T))\big\} \leq E_1(\hat{X}_T,Q^\nu_{\hat{T}}(x_T))\,.
\end{align*}
We set
\begin{align*}
A_T:=&\{x \in \mathbb{R}^d \mid \vert\langle x-x_T,\nu\rangle\vert\leq \sqrt{d}C_\mathcal{L} +2r_{\mathrm{int}}\}\cap(Q^\nu_{\hat{T}}(x_T)\setminus Q^\nu_{T-2r_\mathrm{int}}(\bar{y}_T)) 
\end{align*}
and use Lemma~\ref{lem:elementary-properties}(iv) to obtain 
\begin{align}\label{eq:XhatT-energydecomp}
E_1(\hat{X}_T,Q^\nu_{\hat{T}}(x_T)) = E_1(\hat{X}_T,Q^\nu_{T}(\overline{y}_T) \setminus A_T) + E_1(\hat{X}_T,(Q^\nu_{\hat{T}}(x_T) \cup A_T)\setminus(Q^\nu_{T}(\overline{y}_T))
\end{align}
As $\hat{X}_T \cap \overline{(Q^\nu_T(\overline{y}_T) \setminus A_T)_{r_{\mathrm{int}}}}=(X_T+k_T) \cap \overline{(Q^\nu_T(\overline{y}_T) \setminus A_T)_{r_{\mathrm{int}}} }$, using Lemma~\ref{lem:elementary-properties}(ii),(iii),(x)
\begin{align}\label{ineq:estimateXhat-inside}
E_1(\hat{X}_T,Q^\nu_{T}(\overline{y}_T) \setminus A_T)  = E_1(X_T+k_T,Q^\nu_{T}(\overline{y}_T) \setminus A_T)  \leq E_1(X_T,Q^\nu_{T}(y_T) )\,.
\end{align}
Furthermore, as $\hat{X}_T \cap \overline{B}_{r_{\mathrm{crys}}}(x) = \mathcal{L}(z_T) \cap \overline{B}_{r_{\mathrm{crys}}}(x) $ for all $x \in Q^\nu_{\hat{T}}(x_T)\setminus (Q^\nu_{T}(\bar{x}_T) \cup A_T)$, by \ref{H2}, \ref{H3}, and Lemma~\ref{lem:elementary-properties}(v), we have
\begin{align*}
E_1(\hat{X}_T,(Q^\nu_{\hat{T}}(x_T) \cup A_T)\setminus(Q^\nu_{T}(\overline{y}_T)) \leq E_1(\hat{X}_T, A_T) \leq C \#(\hat{X}_T \cap A_T) \leq C\mathcal{L}^d((A_T)_1) \leq C T^{d-2}\,.
\end{align*} 
This together with \eqref{eq:XhatT-energydecomp}, and \eqref{ineq:estimateXhat-inside} yields
\begin{align}\label{ineq:Xhat-XT}
E_1(\hat{X}_T,Q^\nu_{\hat{T}}(x_T)) \leq E_1(X_T,Q^\nu_{T}(y_T))  + CT^{d-2}\,.
\end{align}
Dividing by $T^{d-1}$, noting that $\hat{T} \geq T$, recalling \eqref{ineq:optimal-sequence-centerlem}, $z_T \to z$ as $T \to +\infty$, and taking the limit as $T\to +\infty$ yields \eqref{ineq-lemproof:center-independence}. Note that this in particular shows that $z_T$ can be chosen independently of $x_T$ and that the convergence is uniform with respect to the choice of the centers $x_T$. This concludes the proof.
\end{proof}
The next corollary shows that the surface energy density does not depend on the translation, which can be chosen to be constant when considering the asymptotic cell problems with converging boundary data.

\begin{cor}\label{cor:translation-independence}
Let $E_{\mathrm{cell}}$ satisfy {\rm \ref{H1}}--{\rm\ref{H10}}. Let $z_1=(R,\tau_1,1),z_2=(R,\tau_2,1)$ and $\nu\in\mathbb{S}^{d-1}$ be given. It holds
\begin{align}\label{coreq:translation-independence}
\Phi_{\mathrm{vac}}(z_1,\nu)=\Phi_{\mathrm{vac}}(z_2,\nu)\,.
\end{align}
Furthermore, for all $z=(R,\tau,1)$ and $z_T\to z$  optimal sequences for $\Phi_{\mathrm{vac}}(z,\nu)$, we can assume $z_T=(R_T,\tau,1)$ for all $T>0$.
\end{cor}
\begin{proof} Let $\{X_T\}_T,\{z_T\}_T$ be an optimal sequence for $\Phi_{\mathrm{vac}}(z_2,\nu)$, i.e., due to Lemma~\ref{lem:center-independence} applied with $x_T=0$ for all $T$, we have $X_T \in \mathrm{Adm}_{1,\lambda}^{(z_T,\textbf{0})}(Q^\nu_T)$, $z_T \to z_2$, and
\begin{align}\label{cor-proof-eq:optimal-sequence}
\liminf_{T\to\infty}\frac{1}{T^{d-1}}E_1(X_T,Q_T^\nu)=\Phi_{\mathrm{vac}}(z_2,\nu)\,.
\end{align}
Setting $\tilde{X}_T:=X_T-R_T\tau_T+R_T\tau_1$, $\tilde{z}_T=(R_T,\tau_1,1)$ and $y_T:=-R_T\tau_T+R_T\tau_1$ it holds $\tilde{X}_T \in \mathrm{Adm}_{1,\lambda}^{(\tilde{z}_T,\textbf{0})}(Q_T^\nu(y_T))$ and $\tilde{z}_T\to z_1$. By Lemma~\ref{lem:elementary-properties}(ii) it holds
\begin{align}
E_1(\tilde{X}_T,Q^\nu_T(y_T))=E_1(X_T,Q^\nu_T)\,.
\end{align}
Additionally, it holds
\begin{align} \label{ineq:phi-vac1-less-phi-vac2}
\Phi_{\mathrm{vac}}(z_1,\nu)\leq\liminf_{T\to\infty}\frac{1}{T^{d-1}}E_1(\tilde{X}_T,Q^\nu_T(y_T))=\liminf_{T\to\infty}\frac{1}{T^{d-1}}E_1(X_T,Q^\nu_T)=\Phi_{\mathrm{vac}}(z_2,\nu)\,.
\end{align}
Exchanging the roles of $z_1$ and $z_2$ in the above proof yields the claim. Now \eqref{coreq:translation-independence} together with \eqref{ineq:phi-vac1-less-phi-vac2} and the precise construction of $\tilde{z}_T$ show also that $\tau_T$ can be chosen to be fixed for all $T$. 
\end{proof}
\subsection{Independence of cube orientation}\label{subsec:orientation}
 The goal of this subsection is to show that the orientation of the cubes is irrelevant, as long as one face has the normal $\nu$. To make this precise, we need the following definitions: Given $\nu\in\mathbb{S}^{d-1}$ and $R\in SO(d)$ such that $Re_d=\nu$, we define the boundary region of $RQ_T$ in analogy to \eqref{def:discrete-boundary} as
\begin{align}\label{def:discrete-boundar-R}
\begin{split}
&\p_\lambda RQ_T(x_0):= RQ_{T+\lambda}(x_0) \setminus RQ_{T-\lambda}(x_0)\,,\\ & \p^\pm_\lambda RQ_T(x_0) := \partial_\lambda RQ_T(x_0) \cap \{ z \in \mathbb{R}^d \colon \pm \langle z-x_0,\nu\rangle \geq r_{\mathrm{int}}\}\,, \\
&\p^c_\lambda RQ_T(x_0) := \partial_\lambda RQ_T(x_0) \setminus (\p^+_\lambda RQ_T(x_0) \cup \p^-_\lambda RQ_T(x_0))\,,
\end{split}
\end{align}
where $RQ_T(x_0) = x_0 + RQ_T$.
Given $z\in \mathcal{Z}$, $ T >0$, $x_0 \in \mathbb{R}^d$ we say that  $X \in \mathrm{Adm}_{1,\lambda}^{(z^+,z^-)}(RQ_T(x_0))$ if it satisfies the following:
\begin{itemize}
\item[(i)] $E_1(X) <+\infty\,,$
\item[(ii)] $X = \mathcal{L}(z^\pm)$  on $\partial_{ \lambda}^\pm RQ_T(x_0)$
\item[(iii)] $X = \emptyset$  on $\partial_{\lambda}^c RQ_T(x_0)\,.$
\end{itemize}
and  we define the corresponding asymptotic cell formula by
\begin{align}\label{def:Phi-R}
\begin{split}
\Phi_{\mathrm{vac}}^R(\nu,z):=\min\big\{\liminf_{T\to\infty}\frac{1}{T^{d-1}}\inf\{&E_1(X,RQ_T(y_T))\mid \\&X \in \mathrm{Adm}_{1,\lambda}^{(z,\textbf{0})}(RQ_T(y_T))\}\mid\{y_T\}_T \subset \mathbb{R}^d\,, z_T\to z \big\}\,.
\end{split}
\end{align}
With this definition it holds $\Phi_{\mathrm{vac}}^{R_\nu}(\nu,z)=\Phi_{\mathrm{vac}}(\nu,z)$ and for $R=R_\nu$ \eqref{def:discrete-boundar-R}  coincides with \eqref{def:discrete-boundary}. With these definitions in mind, we observe that the conclusions of Section~\ref{sec:reduction-two-lattices} and Subsection~\ref{subsec:centers} still hold for this alternate definition. We now show the following:
\begin{lem}\label{lem:independence-of-orientation}
Let $E_{\mathrm{cell}}$ satisfy {\rm\ref{H1}}--{\rm\ref{H10}}. Let $z\in\mc{Z},\nu\in\mathbb{S}^{d-1}$ be given and let $R_1,R_2\in SO(d)$ be given, such that $R_1e_d=R_2e_d=\nu$. Then, it holds
\begin{align}\label{lem-eq:R1-R-2}
\Phi^{R_1}_{\mathrm{vac}}(z,\nu)=\Phi^{R_2}_{\mathrm{vac}}(z,\nu)\,.
\end{align}
In particular, $\Phi(\textbf{0},z,\nu) = \Phi_{\mathrm{vac}}(z,-\nu)$.
\end{lem}
\begin{proof} Our goal is to prove
\begin{align}\label{ineq:R1-R2}
\Phi^{R_1}_{\mathrm{vac}}(z,\nu)\leq\Phi^{R_2}_{\mathrm{vac}}(z,\nu)\,.
\end{align}
Clearly, by exchanging the roles of $R_1$ and $R_2$, this shows the first part of the Lemma. Let $1\ll S\ll T$ and let $\{y_T\}_T$ be any given sequence. We set  
\begin{align*}
\mathcal{Z}_{S,T} = \{ x \mid x= y_T + R_2 j\,, j \in S\mathbb{Z}^{d-1} \times \{0\}\,, R_2Q_S(x) \subset R_1Q_{T-\lambda}(y_T)\}\,.
\end{align*}
By Lemma~\ref{lem:center-independence} up to a (non relabeled) subsequence in $S$ and a  sequence $z_S\to z$  as $S\to+\infty$ such that for all $k \in \mathcal{Z}_{S,T}$ we find $X_{k,S} \in \mathrm{Adm}_{1,\lambda}^{(z,\textbf{0})}(R_2Q_S(k))$ such that
\begin{align}\label{ineq:almost-optimalXjS}
\frac{1}{S^{d-1}}E_1(X_{k,S},R_1Q_S(k))\leq\Phi_{\mathrm{vac}}^{R_2}(z,\nu)+\eta_S\,,
\end{align}
where $\eta_S$ is a null sequence independent of $x_{z,S}$.  In the following, we construct a competitor for $\Phi_{\mathrm{vac}}^{R_1}(z,\nu)$ on the cube $R_1Q_T(y_T)$ with the use of the competitors defined on the cubes of side-length $S$.   We define
\begin{align*}
\mathcal{Q}_{S,T}:=\bigcup_{j \in \mathcal{Z}_{S,T}}R_2Q_S(z)
\end{align*}
and
\begin{align}\label{def:XT-orientation}
X_T:=\begin{cases}
X_{k,S}&\text{in }R_2Q_S(k)\,, k \in \mathcal{Z}_{S,T}\,,\\
\mc{L}(z_T)&\text{in }\{x\mid\langle x-y_T,\nu\rangle\geq r_{\mathrm{int}}\}\setminus \mathcal{Q}_{S,T}\,, \\
\emptyset &\text{otherwise.}
\end{cases}
\end{align}
By definition $X_T \in \mathrm{Adm}_{1,\lambda}^{(z,\textbf{0})}(R_1Q_T)$ and therefore
\begin{align}\label{ineq:XT-adm-orientation}
\inf\{ E_1(X,R_1Q_T(y_T))\mid X \in \mathrm{Adm}_{1,\lambda}^{(z,\textbf{0})}(R_1Q_T(y_T))\} \leq E_1(X_T,R_1Q_T(y_T))\,.
\end{align}
We claim that
\begin{align}\label{ineq:XT-XS-orientation}
E_1(X_T,R_1Q_T(y_T))\leq T^{d-1}(\Phi_{\mathrm{\mathrm{vac}}}^{R_2}(z,\nu)+\eta_S)+CS^{2}T^{d-2}\,.
\end{align}
Once this is proven the inequality \eqref{ineq:R1-R2} follows by using \eqref{ineq:XT-adm-orientation}, dividing by $T^{d-1}$, sending $T\to +\infty$, and finally sending $S\to +\infty$. We now prove \eqref{ineq:XT-XS-orientation}. First note that, as $E_1(X_T,R_1Q_T(y_T))<+\infty$, using Lemma~\ref{lem:elementary-properties}(iv), we have
\begin{align}\label{eq:decomposition-XTR1}
E_1(X_T,R_1Q_T(y_T))=\sum_{k\in \mathcal{Z}_{S,T}} E_1(X_T,R_2Q_S(k))+E_1(X_T,R_1Q_T(y_T)\setminus \mathcal{Q}_{S,T})\,.
\end{align}
Now for each $k \in \mathcal{Z}_{S,T}$ as $ X_T \cap \overline{(R_2Q_S(k))_{r_\mathrm{int}}} = X_{k,S} \cap \overline{(R_2Q_S(k))_{r_\mathrm{int}}}$, due to Lemma~\ref{lem:elementary-properties}(x), we have
\begin{align*}
E_1(X_T,R_2Q_S(k))= E_1(X_{k,S},R_2Q_S(k))\,.
\end{align*}
Thus, using \eqref{ineq:almost-optimalXjS} and the fact that $\#\mathcal{Z}_{S,T} \leq \frac{T^{d-1}}{S^{d-1}}$, we obtain
\begin{align*}
\sum_{k\in \mathcal{Z}_{S,T}} E_1(X_T,R_2Q_S(k)) =\sum_{k\in \mathcal{Z}_{S,T}}E_1(X_{k,S},R_2Q_S(k)) \leq T^{d-1}\left(\Phi_{\mathrm{vac}}^{R_2}(z,\nu) +\eta_S\right)\,.
\end{align*}
Therefore, in order to show \eqref{ineq:XT-XS-orientation}, recalling \eqref{eq:decomposition-XTR1}, it suffices to show 
\begin{align}\label{ineq:XTR1outsidecubes}
E_1(X_T,R_1Q_T(y_T)\setminus \mathcal{Q}_{S,T}) \leq CS^{2}T^{d-2}\,.
\end{align}
In order to obtain this estimate, we define
\begin{align*}
A_T:=(R_1Q_T(y_T)\setminus R_1Q_{T-2S}(y_T))\cap\{y \in \mathbb{R}^d\mid\vert\langle\nu,y-y_T\rangle\vert\leq S\}.
\end{align*}
Observe that for all $x \in R_1Q_T(y_T)\setminus (\mathcal{Q}_{S,T} \cup A_T)$ we have $X_T \cap \overline{B}_{r_\mathrm{crys}}(x) = \mathcal{L}(z) \cap \overline{B}_{r_\mathrm{crys}}(x)$ and therefore, due to \ref{H3}, $E_{\mathrm{cell}}(x,X_T)=0$. Now for $S$ and thus $T$ large enough it holds $\vert (A_T)_1\vert\leq CT^{d-2}S^{2}$, and therefore, using Lemma~\ref{lem:elementary-properties}(v) and \ref{H2}, we obtain \eqref{ineq:XTR1outsidecubes}. Now, $\Phi(\textbf{0},z,\nu) = \Phi_{\mathrm{vac}}(z,-\nu)$ follows from the fact that 
$R_{-\nu}Je_d=R_{\nu}e_d=\nu$  for $J \in SO(d)$ given by $J=\mathrm{diag}(-1,1,\ldots,1,-1)$ and a small approximation argument, since our cubes are half-open. This concludes the proof.
\end{proof}
\begin{remark}\label{rem:limit} The construction in the proof of Lemma~\ref{lem:independence-of-orientation} shows that the $\liminf$ in the definition of $\Phi_{\mathrm{vac}}$ (and similarly $\varphi_{\mathrm{vac}}$) is actually a limit. 
\end{remark}

\subsection{Converging to fixed boundary values}\label{sec:7.3}
We are now in a position to show that we can replace converging boundary values by fixed ones in the definition of $\Phi_{\mathrm{vac}}$. Namely, we show the following lemma.
\begin{lem}\label{lem:fixed-bdryvalues}
Let $E_{\mathrm{cell}}$ satisfy {\rm\ref{H1}}--{\rm\ref{H10}}. For all $z\in\mc{Z}$ and $\nu\in\mathbb{S}^{d-1}$ and every sequence $\{x_T\}_T$ it holds
\begin{align*}
\Phi_{\mathrm{vac}}(z,\nu)=\varphi_{\mathrm{vac}}(z,\nu)\,.
\end{align*}
\end{lem}
\begin{proof}Recalling \eqref{eq:propdensity}, \eqref{def:phivac}, and Lemma~\ref{lem:elementary-properties}(vii), we have that 
\begin{align*}
\varphi_{\mathrm{vac}}(z,\nu) = \liminf_{T\to\infty}\frac{1}{T^{d-1}}\inf\{&E_1(X,Q_T^\nu(x_T))\mid X \in \mathrm{Adm}_{1,\lambda}^{(z,\textbf{0})}(Q^\nu_T(x_T))\}\,.
\end{align*}
 Obviously, we have that
\begin{align*}
\Phi_{\mathrm{vac}}(z,\nu)\leq\liminf_{T\to\infty}\frac{1}{T^{d-1}}\inf\{&E_1(X,Q_T^\nu(x_T))\mid X \in \mathrm{Adm}_{1,\lambda}^{(z,\textbf{0})}(Q^\nu_T(x_T))\} =\varphi_{\mathrm{vac}}(z,\nu)\,.
\end{align*}
In order to conclude the proof, it remains to show
\begin{align}\label{ineq:Phi-vac-fixed-cond}
\Phi_{\mathrm{vac}}(z,\nu)\geq\liminf_{T\to\infty}\frac{1}{T^{d-1}}\inf\{&E_1(X,Q_T^\nu(x_T))\mid X \in \mathrm{Adm}_{1,\lambda}^{(z,\textbf{0})}(Q^\nu_T(x_T))\}\,.
\end{align}
To this end, let $z=(R,\tau,1)\in\mc{Z}$ and $\nu\in\mathbb{S}^{d-1}$ be given. By Lemma~\ref{lem:center-independence}, we can assume that $x_T=0$ for all $T$. Now let $\{z_T\},\{X_T\}$ be an optimal sequence for $\Phi_{\mathrm{vac}}(z,\nu)$.  By Lemma~\ref{lem:center-independence} and Corollary~\ref{cor:translation-independence}, we can assume that $z_T=(R_T,\tau,1)$ and up to a (not relabeled) subsequence we have
\begin{align}\label{ineq:almost-optimal-phivac}
\frac{1}{T^{d-1}}E_1(X_T,Q^\nu_T)\leq\Phi_{\mathrm{vac}}(z,\nu)+\eta_T\,,
\end{align}
where $\eta_T\to 0$ as $T \to +\infty$. We define $X_T^{\mathrm{rot}}=M_TX_T$, where $M_T=RR_T^{-1}$ with $M_T\to\id$.  Here, $M_T\to\id$ in $SO(d)$ follows as the quotient map is a local diffeomorphism, see above \eqref{def:isometryspace}. Then, as $X_T \in \mathrm{Adm}_{1,\lambda}^{(z_T,\textbf{0})}(Q_T^\nu)$ we have $X_T^{\mathrm{rot}} \in \mathrm{Adm}_{1,\lambda}^{(z,\textbf{0})}(M_TQ_T^\nu)$
Using Lemma~\ref{lem:elementary-properties}(iv) we have
\begin{align} \label{ineq:XT-Xrot}
\begin{split}
E_1(X_T,Q_T^\nu)=E_1(X_T^{\mathrm{rot}},M_TQ_T^\nu)\geq\inf\{E_1(X,M_TQ_T^\nu)\mid X \in \mathrm{Adm}_{1,\lambda}^{(z,\textbf{0})}(M_TQ_T^\nu)\}\,.
\end{split}
\end{align}
Notice that the asymptotic cell problem is a problem with fixed boundary values, but it is defined on a sequence of rotating cubes. Now, using \eqref{ineq:almost-optimal-phivac} and \eqref{ineq:XT-Xrot}, it suffices to show
\begin{align}\label{ineq:fixed-cubes-rotated-cubes}
\begin{split}
\liminf_{T\to\infty}\frac{1}{T^{d-1}}\inf\{E_1(X,M_TQ_T^\nu)&\mid X \in \mathrm{Adm}_{1,\lambda}^{(z,\textbf{0})}(M_TQ_T^\nu)\}\\&
\geq\liminf_{T\to\infty}\frac{1}{T^{d-1}}\inf\{E_1(X,Q^\nu_T)\mid X \in \mathrm{Adm}_{1,\lambda}^{(z,\textbf{0})}(Q_T^\nu) \}\,.
\end{split}
\end{align} 
As $M_T \to \mathrm{Id}$, for every $\delta>0$ there exists $T(\delta)$ such that for $T\geq T(\delta)$ it holds $\vert M_T\nu-\nu\vert\leq\delta$ and $\vert M_TR_\nu e_i-R_\nu e_i\vert\leq\delta$ for $i=1,\dots,d-1$. This implies
$M_TQ_{T}^\nu\subset Q^\nu_{(1+2\sqrt{d}\delta)T}$. Indeed, if $x\in M_TQ_T^\nu$ we have $\vert x\vert\leq\frac{\sqrt{d}}{2}{T}$ and therefore
\begin{align*}
\vert\langle\nu,x\rangle\vert\leq\vert\langle M_T\nu,x\rangle\vert+\vert \langle M_T\nu-\nu,x\rangle\vert\leq\frac{T}{2}+\vert M_T\nu-\nu\vert\,\vert x\vert\leq(1+\sqrt{d}\delta)\frac{T}{2}\,.
\end{align*}
Similarly, one can argue to estimate $|\langle R_\nu e_i, x\rangle|$ for $i=1,\ldots,d-1$, which implies the aforementioned inclusion. We now set $T_{\delta}:=(1+2\max(\sqrt{d},\lambda)\delta)T$ so that $M_TQ_T^\nu\subset Q^\nu_{T_\delta}$.  Moreover, for $T$ large enough, we have 
\begin{align}\label{eq:intersection-bdry}
\p_{\lambda} Q^\nu_{T_\delta}\cap M_TQ_{T+\lambda}^\nu=\emptyset\,.
\end{align}
For each $T\geq T(\delta)$ we choose a configuration $X_T \in \mathrm{Adm}_{1,\lambda}^{(z,\textbf{0})}(M_T Q^\nu_T)$, $X_T \subset \mathcal{L}(z)$ (this is possible due to  Lemma~\ref{lem:separation-2-lattices} and Lemma~\ref{lem:independence-of-bdry-layer}), and
\begin{align}\label{eq:XToptforMT}
E_1(X_T,M_TQ_T^\nu)\leq\inf\{E_1(X,M_TQ_T^\nu)\mid  X\in \mathrm{Adm}_{1,\lambda}^{(z,\textbf{0})}(M_T Q^\nu_T)\}+\frac1T\,.
\end{align}
 We define
\begin{align}\label{def:XTdelta}
X_{T_\delta}:=\begin{cases}
X_T&\text{in }M_TQ_{T+\lambda}^\nu\,,\\
\mc{L}(z)&\text{in } \{x\mid\langle\nu,x\rangle\geq r_{\mathrm{int}}\}\setminus M_TQ_{T+\lambda}^\nu\,,\\
\emptyset&\text{else}.
\end{cases}
\end{align}
Due to \eqref{eq:intersection-bdry} we have that $X_{T_\delta} \in \mathrm{Adm}_{1,\lambda}^{(z,\textbf{0})}(Q_{T_\delta}^\nu)$ and, as $X_{T_\delta} \subset \mathcal{L}(z)$, $E_1(X_{T_\delta},Q^\nu_{T_\delta})<+\infty$. We claim that
\begin{align}\label{ineq:final-estimate-fixedbdry}
E_1(X_{T_\delta},Q^\nu_{T_\delta})\leq E_1(X_T,M_TQ_T^\nu)+C\delta T^{d-1}
\end{align}
for some dimensional constant $C>0$. We postpone the proof of \eqref{ineq:final-estimate-fixedbdry} and show how we can conclude once it is proven. Noting that $T\leq T_\delta \leq CT$ and  $X_{T_\delta} \in \mathrm{Adm}_{1,\lambda}^{(z,\textbf{0})}(Q_{T_\delta}^\nu)$ we have
\begin{align*}
\liminf_{T_\delta \to +\infty}  \frac{1}{T_\delta^{d-1}} \inf\{E_1(X,Q_{T_\delta}^\nu)\mid  X\in \mathrm{Adm}_{1,\lambda}^{(z,\textbf{0})}(Q^\nu_{T_\delta})\}&\leq \liminf_{T_\delta \to +\infty}  \frac{1}{T_\delta^{d-1}} E_1(X_{T_\delta},Q^\nu_{T_\delta}) \\&\leq \liminf_{T \to +\infty}  \frac{1}{T^{d-1}} E_1(X_T,M_TQ_T^\nu) +\delta\,.
\end{align*} 
Recalling \eqref{eq:XToptforMT} concludes the proof by letting $\delta \to 0$. It now remains to prove \eqref{ineq:final-estimate-fixedbdry}. To this end, we set $A_T=\{x\mid\vert\langle x,\nu_T\rangle\vert\leq 2r_{\mathrm{int}}\} \cap Q_{T_\delta}^\nu \setminus M_TQ^\nu_{T-2r_{\mathrm{int}}}$. By Lemma~\ref{lem:elementary-properties}(iv) we can write
\begin{align}\label{eq:decompXTdelta}
E_1(X_{T_\delta},Q^\nu_{T_\delta})=E_1(X_{T_\delta},M_TQ_T^\nu \setminus A_T)+E_1(X_{T_\delta},(Q^\nu_{T_\delta}\cup A_T)\setminus M_TQ_T^\nu)\,.
\end{align} 
As $ X_{T_\delta} \cap  \overline{(M_TQ_T^\nu \setminus A_T)_{r_{\mathrm{int}}}} = X_{T} \cap  \overline{(M_TQ_T^\nu \setminus A_T)_{r_{\mathrm{int}}}}$, by Lemma~\ref{lem:elementary-properties}(iii),(x) we have
\begin{align}\label{ineq:XTdelta-inside}
E_1(X_{T_\delta},M_TQ_T^\nu \setminus A_T) = E_1(X_{T},M_TQ_T^\nu \setminus A_T) \leq E_1(X_{T},M_TQ_T^\nu)\,.
\end{align} 
Furthermore, note that for all $x \in X_{T_\delta} \cap Q^\nu_{T_\delta}\setminus (M_TQ_T^\nu \cup A_T)$ we have $X_{T_\delta} \cap \overline{B}_{r_{\mathrm{crys}}}(x) = \mathcal{L}(z) \cap \overline{B}_{r_{\mathrm{crys}}}(x)$ and therefore, due to \ref{H3}, we have $E_\mathrm{cell}(x,X_{T_\delta})=0$. Additionally, $\mathcal{L}^d((A_T)_1) \leq C\delta T^{d-1}$ and therefore, using \ref{H2}, we have
\begin{align*}
E_1(X_{T_\delta},(Q^\nu_{T_\delta}\cup A_T)\setminus M_TQ_T^\nu) \leq C\#(A_T\cap X_{T_\delta}) \leq C\delta T^{d-1}\,.
\end{align*} 
This, together with \eqref{eq:decompXTdelta} and \eqref{ineq:XTdelta-inside}, shows \eqref{ineq:final-estimate-fixedbdry}. This concludes the proof.
\end{proof}
\begin{remark}\label{rem:continuity} The same proof as the proof of Lemma~\ref{lem:fixed-bdryvalues} shows that for every $z\in \mathcal{Z}\setminus \{{\bf 0}\}$ the function $\nu \mapsto \varphi_{\mathrm{vac}}(z,\nu)$ is continuous.
\end{remark}
\section[Cell Formula III]{Cell Formula III: Relation of converging and fixed boundary values}\label{sec:converging-fixed-bdryvalues}
\subsection{Converging to fixed boundary values}\label{subsec:8.1}
In this section we show that for $z^+,z^-\in\mc{Z}\setminus\{\textbf{0}\}, z^+\neq z^-$ the energy density of a transition from $z^+$ to $z^-$ is given as the sum of the energies of the transitions from the lattice $\mathcal{L}(z^+)$ to vacuum and from vacuum to $\mathcal{L}(z^-)$. In order to prove this, we will use the results of Section~\ref{sec:reduction-two-lattices} and Section~\ref{sec:vacuum-energy}. The main result of this section is
\begin{lem}\label{lem:Phi=varphi}
Let $E_{\mathrm{cell}}$ satisfy {\rm\ref{H1}}--{\rm\ref{H10}}. Let $\nu\in\mathbb{S}^{d-1}$ and let $z^+,z^-\in\mc{Z}\setminus\{\bf{0}\}$ such that $z^+\neq z^-$. Then it holds
\begin{align}\label{lem-eq:Phi=varphi}
\Phi(z^+,z^-,\nu)=\Phi_{\mathrm{vac}}(z^+,\nu)+\Phi_{\mathrm{vac}}(z^-,-\nu)=\varphi_{\mathrm{vac}}(z^+,\nu)+\varphi_{\mathrm{vac}}(z^-,-\nu)= \varphi(z^+,z^-,\nu)\,.
\end{align}
\end{lem}
\begin{proof} We claim that the statement of the Lemma follows once we show 
\begin{align}\label{ineq:phi-vacgreaterphi}
\varphi_{\mathrm{vac}}(z^+,\nu)+\varphi_{\mathrm{vac}}(z^-,-\nu)\geq \varphi(z^+,z^-,\nu)\,.
\end{align}
First of all note that, as $z_T^\pm =z^\pm$ can be chosen in the definition of $\Phi$ we have
\begin{align*}
\varphi(z^+,z^-,\nu) \geq \Phi(z^+,z^-,\nu)\,. 
\end{align*}
Additionally, due to Lemma~\ref{lem:separation-2-lattices}, Lemma~\ref{lem:independence-of-bdry-layer},Lemma~\ref{lem:independence-of-orientation}, Lemma~\ref{lem:fixed-bdryvalues}, and \eqref{ineq:phi-vacgreaterphi} we have
\begin{align*}
\Phi(z^+,z^-,\nu) &\geq \Phi(z^+,{\bf 0},\nu) +\Phi({\bf 0},z^-,\nu)  = \Phi_{\mathrm{vac}}(z^+,\nu) +\Phi_{\mathrm{vac}}(z^-,-\nu) \\&=\varphi_{\mathrm{vac}}(z^+,\nu) +\varphi_{\mathrm{vac}}(z^-,-\nu) \geq \varphi(z^+,z^-,\nu)\,.
\end{align*}
This shows that all the inequalities are actually equalities and concludes the proof once \eqref{ineq:phi-vacgreaterphi} is proven. We prove this in the following. Let $1\ll S\ll T$.  Let 
\begin{align}
\mathcal{Z}_{S,T}^\nu : = \{ x \mid x=R_\nu k \,, k \in S\mathbb{Z}^{d-1} \times \{0\}\,, Q^\nu_S(x) \subset Q_{T-\lambda}^\nu\}\,.
\end{align} 
Using  Lemma~\ref{lem:center-independence}, Lemma~\ref{lem:fixed-bdryvalues}, Lemma~\ref{lem:separation-2-lattices}, and Lemma~\ref{lem:independence-of-bdry-layer} for every $j \in \mathcal{Z}_{S,T}^\nu$ we define $x_{j,S}^\pm := j \pm S \nu$ and we choose a sequence of configurations $X_{j,S}^+ \in \mathrm{Adm}_{1,\lambda}^{(z^+,\textbf{0})}(Q^\nu_S(x_{j,S}^+))$ such that $X_{j,S}^+ \subset \mathcal{L}(z^+)$ and
\begin{align}\label{ineq:XjSplus}
E_1(X_{j,S}^+,Q_S^\nu(x_{j,S}^+))\leq S^{d-1}(\varphi_{\mathrm{vac}}(z^+,\nu)+\eta_S)
\end{align}
and $ X_{j,S}^- \in \mathrm{Adm}_{1,\lambda}^{(\textbf{0},z^-)}(Q^{\nu}_S(x_{j,S}^-))$ such that $X_{j,S}^- \subset \mathcal{L}(z^-)$ and
\begin{align}\label{ineq:XjSminus}
E_1(X_{j,S}^-,Q_S^\nu(x^-_{j,S}))\leq S^{d-1}(\varphi_{\mathrm{vac}}(z^-,-\nu)+\eta_S)\,,
\end{align}
where $\eta_S \to 0$ as $S\to +\infty$. By construction we have that for any $j_1,j_2\in \mathcal{Z}^\nu_{S,T}$ it holds 
\begin{align}\label{ineq:distance-cubes}
\dist(Q_S^\nu(x_{j_1,S}^+),Q_S^\nu(x_{j_2,S}^-))\geq 2r_{\mathrm{int}}\,.
\end{align}
We set 
\begin{align*}
\mathcal{Q}_{S,T}:= \bigcup_{j\in \mathcal{Z}^\nu_{S,T}}\left(Q^\nu_S(x_{j,S}^+)\cup Q^\nu_S(x_{j,S}^-) \right)
\end{align*}
and define, see Figure~\ref{Fig:defXT},
\begin{align}\label{def:XT}
X_T:=\begin{cases}
X_{j,S}^\pm&\text{in }Q^\nu_S(x_{j,S}^\pm),\\
\mc{L}(z^\pm)&\text{in }(\{x \in \mathbb{R}^d\mid\pm\langle x,\nu\rangle\geq S \}\setminus \mathcal{Q}_{S,T})\cup\p^\pm_{\lambda} Q^\nu_T\,,\\
\emptyset &\text{otherwise.}
\end{cases}
\end{align}
\begin{figure}
\begin{tikzpicture}[scale=.6]

\tikzset{>={Latex[width=1mm,length=1mm]}};

\fill[gray!50](-5.4,-5.4) rectangle++(10.8,4.8);
\fill[gray!80](-5.4,0.6) rectangle++(10.8,4.8);
\fill[white](-4.5,1.4)--(4.5,1.4)--(4.5,-1.4)--(-4.5,-1.4)--cycle;
\draw[thick](-5,-5)--(-5,5)--(5,5)--(5,-5)--cycle;
\draw[dashed](-5,0)--(6,0);
\draw[->](6,0)--(6,1) node[right]{$\nu$};
\draw(-4.5,0.4)--(4.5,0.4);
\draw(-4.5,1.4)--(4.5,1.4);
\draw(-4.5,-0.4)--(4.5,-0.4);
\draw(-4.5,-1.4)--(4.5,-1.4);
\draw[dashed](-4.6,-0.9)--(4.6,-0.9);
\draw[dashed](-4.6,0.9)--(4.6,0.9);

\foreach \i in {-4,...,5}{
\draw(-0.5+1*\i,0.4)--(-0.5+1*\i,1.4);
\draw(-0.5+1*\i,-0.4)--(-0.5+1*\i,-1.4);
}
\draw[<->](-5.5,-5)--(-5.5,0) node[left]{$T$}--(-5.5,5);
\draw[<->](-5.1,-0.6)--(-5.1,0.2) node[right]{\small $\sim r_{\mathrm{int}}$}--(-5.1,0.6);

\draw[<->](0.5,0.3)--(1.0,0.3) node[above]{$S$}--(1.5,0.3);
\draw(0,3) node[above]{$\mathcal{L}(z^+)$};
\draw(0,1.3)--(0.2,1.8) node[above]{$X_{j,S}^+$};
\draw(0,-3) node[below]{$\mathcal{L}(z^-)$};
\draw(0,-1.3)--(-0.2,-1.8) node[below]{$X_{j,S}^-$};

\end{tikzpicture}
\caption{The definition of $X_T$}
\label{Fig:defXT}
\end{figure}
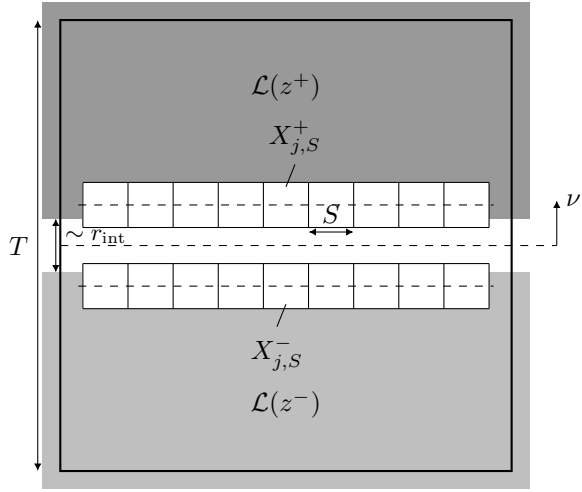
Note that by \eqref{ineq:distance-cubes}, the properties of $X_{j,S}^\pm$ we have that $ X_T \subset \mathcal{L}(z^+) \cup \mathcal{L}(z^-)$, $X_T \in \mathrm{Adm}_{1,\lambda}^{(z^+,z^-)}(Q_T^\nu)$, and $E_1(X_T,Q_T^\nu)<+\infty$. We now show  
\begin{align}\label{ineq:varphi-vac}
E_1(X_T,Q_T^\nu)\leq  T^{d-1}(\varphi_{\mathrm{vac}}(z^+,\nu)+\varphi_{\mathrm{vac}}(z^-,-\nu)+2\eta_S+CS^{-1}) + CST^{d-2}\,.
\end{align}
Clearly, once this is proven \eqref{ineq:phi-vacgreaterphi} follows by dividing by $T^{d-1}$ and sending first $T \to +\infty$ and then $S \to +\infty$. In order to prove \eqref{ineq:varphi-vac}, we notice that by Lemma~\ref{lem:elementary-properties}(iv) we have
\begin{align}\label{eq:energy-decomp-XT}
E_1(X_T,Q^\nu_T)=E_1(X_T,Q^\nu_T\setminus \mathcal{Q}_{S,T})+\sum_{j\in \mathcal{Z}^\nu_{S,T}}\left( E_1(X_T,Q^\nu_S(x_{j,S}^+))+E_1(X_T,Q^\nu_S(x_{j,S}^-))\right)\,.
\end{align}
Note that for all $j \in \mathcal{Z}^\nu_{S,T}$ and all $x \in X_T \cap Q^\nu_S(x_{j,S}^\pm)$ we have $X_T \cap \overline{B}_{r_\mathrm{int}}(x) = X_{j,S}^\pm \cap \overline{B}_{r_\mathrm{int}}(x)$, whenever $x\notin A_j^{\pm}:=(Q^\nu_S(x_{j,S}^\pm)\setminus Q^\nu_{S-2r_{\mathrm{int}}}(x_{j,S}^\pm))\cap\{y\mid\vert\langle y-x_{j,S}^\pm,\nu\rangle\vert\leq 2r_{\mathrm{int}}\}$. However, we have $\#(X_T\cap A_j^\pm)\leq\mc{L}^{d}((A_j^\pm)_1)\leq CS^{d-2}$ by Lemma~\ref{lem:elementary-properties}(v) and therefore, using \eqref{ineq:XjSplus} and \eqref{ineq:XjSminus} respectively, we obtain
\begin{align*}
E_1(X_T,Q^\nu_S(x_{j,S}^+)) \leq E_1(X_{j,S}^+,Q_S^\nu(x_{j,S}^+))+CS^{d-2} \leq  S^{d-1}(\varphi_{\mathrm{vac}}(z^+,\nu)+\eta_S+CS^{-1})
\end{align*}
and 
\begin{align*}
E_1(X_T,Q^\nu_S(x_{j,S}^-)) \leq E_1(X_{j,S}^-,Q_S^\nu(x_{j,S}^-))+CS^{d-2}\leq  S^{d-1}(\varphi_{\mathrm{vac}}(z^-,-\nu)+\eta_S+CS^{-1})\,.
\end{align*}
Note that $\#\mathcal{Z}^\nu_{S,T} \leq \frac{T^{d-1}}{S^{d-1}}$ and therefore \eqref{ineq:varphi-vac} follows from \eqref{eq:energy-decomp-XT} once we prove that 
\begin{align}\label{ineq:XT-outsideQST}
E_1(X_T,Q^\nu_T\setminus \mathcal{Q}_{S,T}) \leq CS^2T^{d-2}\,.
\end{align}
To prove this estimate, we set  
\begin{align*}
A_T:=(Q^\nu_T\setminus Q^\nu_{T-2S})\cap\{y\mid\vert\langle y,\nu\rangle\vert\leq 2S\}\,.
\end{align*}
We note that by \eqref{def:XT} for all $x \in Q^\nu_T \setminus (\mathcal{Q}_{S,T} \cup A_T) $ we have $X_T \cap \overline{B}_{r_\mathrm{crys}}(x) = \mathcal{L}(z^+) \cap \overline{B}_{r_\mathrm{crys}}(x)$ or $X_T \cap \overline{B}_{r_\mathrm{crys}}(x) = \mathcal{L}(z^-) \cap \overline{B}_{r_\mathrm{crys}}(x)$ and therefore $E_{\mathrm{cell}}(x,X_T) =0$. As $ \mathcal{L}^d((A_T)_1)\leq CS^{2}T^{d-2}$ using Lemma~\ref{lem:elementary-properties}(v) and \ref{H2} there holds
\begin{align*}
E_1(X_T,Q_T^\nu\setminus \mathcal{Q}_{S,T})\leq C\#(A_T \cap X_T)  \leq C \mathcal{L}^d((A_T)_1) \leq  CS^{2}T^{d-2}\,.
\end{align*}
This is \eqref{ineq:XT-outsideQST} and concludes the proof.
\end{proof}

We are now in position to prove Proposition~\ref{prop:relationpsiphi}.

\begin{proof}[Proof of Proposition~\ref{prop:relationpsiphi}] This is a consequence of Lemma~\ref{lem:psi-phi} and Lemma~\ref{lem:Phi=varphi} together with Lemma~\ref{lem:elementary-properties}(vii).
\end{proof}

\subsection{Properties of the energy density}
This subsection is dedicated to the proof of Theorem~\ref{thm:properties-of-phi}. We observe that Theorem~\ref{thm:properties-of-phi}(i) is a consequence of Lemma~\ref{lem:independence-of-orientation} and Lemma~\ref{lem:fixed-bdryvalues}. Theorem~\ref{thm:properties-of-phi}(ii) is a consequence of Lemma~\ref{lem:Phi=varphi}, Theorem~\ref{thm:properties-of-phi}(v) is a consequence of Lemma~\ref{cor:translation-independence}, and Theorem~\ref{thm:properties-of-phi}(vi) is a consequence of Lemma~\ref{lem:elementary-properties}(ii) and Lemma~\ref{lem:independence-of-orientation}. It therefore remains to prove Theorem~\ref{thm:properties-of-phi}(iii) and (iv). 
\begin{proof}[Proof of Theorem~\ref{thm:properties-of-phi}{\rm (iii)}] It is a classical fact that the (pos.~$1$-homogenous) function $\nu \mapsto \varphi_{\mathrm{vac}}(z,\nu)$ is necessarily convex for $L^1_{\mathrm{loc}}$-lower semicontinous functionals defined on partitions, see for example \cite[Theorem~5.11(ii)]{Ambrosio-Fusco-Pallara:2000}. The functional is necessarily lower-semicontinuous with respect to the $L^1_{\mathrm{loc}}$-convergence, it being a $\Gamma$-limit (w.r.t~the strong $L^1_{\mathrm{loc}}$-convergence) of a sequence of functionals, see Theorem~\ref{thm:Gamma-convergence}.
\end{proof}
\begin{proof}[Proof of Theorem~\ref{thm:properties-of-phi}{\rm (iv)}] Fix $z \in \mathcal{Z} \setminus \{{\bf 0}\}$ and $T>0$. We define
\begin{align}\label{def:plane-min}
X_T:=\begin{cases}
\mc{L}(z)&\text{in }\{x\mid\langle x,\nu\rangle\geq r_{\mathrm{int}}\}\,,\\
\emptyset&\text{otherwise.}
\end{cases}
\end{align}
We then have $X_T \in \mathrm{Adm}_{1,\lambda}^{(z,\textbf{0})}(Q^\nu_T)$ and therefore
\begin{align*}
\inf\EEE\{E_1(X,Q_T^\nu)\mid X\in \mathrm{Adm}_{1,\lambda}^{(z,\textbf{0})}(Q^\nu_T)\}\leq E_1(X_T,Q_T^\nu)\,.
\end{align*}
In order to conclude the proof, it suffices to show that
\begin{align}\label{ineq:boundedness}
E_1(X_T,Q_T^\nu)\leq CT^{d-1}
\end{align}
for a universal constant $C$ independent of $\nu$ and $z$. To this end, note that for all atoms in $X_T\cap Q_T^\nu\cap\{x\mid|\langle x,\nu\rangle|\geq 2r_{\mathrm{int}}\}$ there holds $X_T \cap \overline{B}_{r_{\mathrm{crys}}}(x) = \mathcal{L}(z) \cap \overline{B}_{r_{\mathrm{crys}}}(x)$ and therefore, due to \ref{H3}, we have $E_{\mathrm{cell}}(x,X)=0$. On the other hand, $\mathcal{L}^d((Q_T^\nu\cap\{x\mid|\langle x,\nu\rangle|\leq 2r_{\mathrm{int}}\})_1) \leq CT^{d-1}$ and therefore, due to \ref{H2} and Lemma~\ref{lem:elementary-properties}(v), we have
\begin{align*}
E_1(X_T,Q_T^\nu) \leq  C\#\left(Q_T^\nu\cap\{x\mid|\langle x,\nu\rangle|\leq 2r_{\mathrm{int}}\}\right) \leq \mathcal{L}^d((Q_T^\nu\cap\{x\mid|\langle x,\nu\rangle|\leq 2r_{\mathrm{int}}\})_1)\leq CT^{d-1}\,.
\end{align*}
This shows \eqref{ineq:boundedness} and concludes the proof due to Theorem~\ref{thm:properties-of-phi}(iii).
\end{proof}
\appendix

\section{Examples} \label{sec:examples}
In this section, we give examples of interactions that satisfy assumptions \ref{H1}-\ref{H10}.  
\subsection{Integer lattice $\Z^d$}
We set $\mc{L}=\Z^d$. Note that $\mc{L}$ clearly fulfills the conditions \ref{L1}-\ref{L4} with $R_V=\frac{\sqrt{d}}{2}, S_V=1$.  Next, we set
\begin{align}\label{def:unscaled-energy-Z2}
E(X) = \sum_{x\in X} E_{\mathrm{cell}}(x,X)\,,
\end{align}
where, setting $\mathcal{N}_1(x)= \{y \in X\setminus \{x\} \mid |x-y|\leq 1\}$, we define
\begin{align}\label{def:single-interaction}
E_{\mathrm{cell}}(x,X):=\begin{cases}
\frac{1}{2}\left(2d-\#\mathcal{N}_1(x)) +\sum_{y,z\in\mc{N}(x), y\neq z}V_3(\theta_{yxz})\right)&\text{if } \dist(\{x\},X\setminus\{x\})\geq1\,,\\
+\infty&\text{otherwise.}
\end{cases}
\end{align}
Here $\theta_{yxz}$ is the angle spanned by the vectors $y-x$ and $z-x$ in anti-clockwise orientation and $V_3:[0,2\pi]\to\R$ is an angle potential given by
\begin{align}\label{def:angle-potential}
V_3(\theta):=\begin{cases}
0& \text{if }\theta=0\mod\frac{\pi}{2}\,,\\
C_{V_3}&\text{otherwise.}
\end{cases}
\end{align}
where $C_{V_3}>0$ is a large enough constant. Here, the atoms are modeled to interact via the sticky-disc potential, see e.g.~\cite{FriedrichKreutzSchmidt,HeitmannRadin:80}. The $2d$ is a normalization constant corresponding to the number of nearest neighbors in a perfect lattice. 

Also by choosing $C_{V_3}>0$ large enough the cell energy readily fulfills \ref{H1}-\ref{H7} for $d\leq4$, with $r_{\mathrm{int}}=r_{\mathrm{crys}}=1$, some $C$ large enough and some $c$ small enough depending on $C_{V_3}$. Furthermore, for this type of energy, we have $d_{\mathrm{unique}}=2d-3$.

Now for $d=2$ the energy also fulfills \ref{H8}-\ref{H10}, if $C_{V_3}>8$. This can be seen as follows:\newline
\ref{H8} follows readily by non-negativity of the cell energy and $d_{\mathrm{unique}}=1$. \\
Now assume that there is a $x\in X$ such that $\forall z\in\mc{Z}$ it holds $\{x\}\cup\mc{N}(x)\nsubseteq\mc{L}(z)$, then this implies that there are $y,w\in\mc{N}(x)$ such that $\theta_{yxw}\neq0\mod\frac{\pi}{2}$. Without loss of generality, we can assume finite energy for $X$ as otherwise, this is trivial. Then it holds $\#\mc{N}(x)\leq 6$. So removing $x$ yields an increase of cell energy by $1$ in each of its neighbors. However, as $C_{V_3}>8$ and $4-\#\mc{N}(x)\geq-2$ this implies that $E_{\mathrm{cell}}(x,X)>6$, which yields the validity of \ref{H9}.\\
To show the validity of \ref{H10}, assume you have two neighbors $x,y\in X$, where $X$ is a finite energy configuration. As we can associate a lattice to $x$ and $y$ this implies $\max(\#\mc{N}(x),\#\mc{N}(y))\leq 4$ and all angles are multiples of $\frac{\pi}{2}$. In particular, by this condition we have $x\in\mc{L}(z(y))$. We claim that this already implies $z(x)=z(y)$ and so \ref{H10} is an empty condition. Let $z(y)=(R_y,\tau_y,1),z(x)=(R_x,\tau_x,1)$. Then by assumption $\vert x-y\vert=1$, it holds 
\begin{align*}
y=R_x(\tilde{y}+\tau_x)=R_y(\bar{y}+\tau_y)\,,\\
x=R_x(\tilde{x}+\tau_x)=R_y(\bar{x}+\tau_y)\,.
\end{align*}
for some $\tilde{x},\bar{x},\tilde{y},\bar{y}\in\Z^2$. Furthermore, this implies $\tilde{w}=\tilde{y}-\tilde{x}\in\{\pm e_1,\pm e_2\},\bar{w}=\bar{y}-\bar{x}\in\{\pm e_1,\pm e_2\}$
But now this shows
\begin{align*}
R_x\tilde{w}=y-x=R_y\bar{w}\,,
\end{align*}
which implies as $\tilde{w},\bar{w}\in\{\pm e_1,\pm e_2\}$ that $R_y^{-1}R_x$ is a symmetry of $\Z^2$, so by construction of $\mc{Z}$ this shows $R_y=R_x$. But then the image of $x$ under $\mc{L}(z(x))$ and $\mc{L}(z(y))$ differs only by  a vector in $\Z^2$, which implies $\tau_x=\tau_y$ and so finally $z(x)=z(y)$.
Notice that this argumentation only works in dimension $2$. \\
Indeed, for $d=3$, i.e. $\Z^3$ \ref{H10} does not hold. Here it holds $d_{unique}=3$. Now setting $X=\{0,e_3,2e_3,e_3+e_1,e_3-e_1,Re_1,-Re_1\}$, where $R$ is a $\frac{\pi}{4}$-rotation in the $xy$-plane, then we can associate to $e_3$ $\Z^3$ as its lattice and to $0$ $R\Z^3$, but removing either $e_3$ or $0$ increases the energy. A similar argument can be adapted for higher dimensions. Via classical slicing techniques, it can be shown that
\begin{align*}
\varphi_{\mathrm{vac}}(\nu) = \frac{1}{2}\|\nu\|_1\,.
\end{align*}
\subsection{Honeycomb lattice}
We set $d=2$ and
\begin{align}\label{def:Honeycomb}
\mc{L}=\{\lambda_1v_1+\lambda_2e_2+\delta e_1\mid \lambda_1,\lambda_2\in\Z,\delta\in\{\pm1\}\}\,,
\end{align}
where $v_1=\begin{pmatrix}0 &\sqrt{3}\end{pmatrix}^T,v_2=\begin{pmatrix}\frac32&\frac{\sqrt{3}}{2}\end{pmatrix}^T$. Clearly,  $\mathcal{L}$ fulfills \ref{L1}-\ref{L4}. We define
\begin{align}\label{\theequation}
 E_{\mathrm{cell}}(x,X):=\begin{cases}
\frac{1}{2}\left(3-\#\mc{N}(x)\right)+\frac{1}{2}\sum_{y,z\in\mc{N}(x), y\neq z}V_3(\theta_{yxz})&\text{if } \dist(\{x\},X\setminus\{x\})\geq1,\\
+\infty&\text{else.}
\end{cases} 
\end{align}
Here $\theta_{yxz}$ is as before and $V_3:[0,2\pi]\to\R$ is again an angle potential given by
\begin{align}\label{def:angularpotential-honeycomb}
V_3(\theta):=\begin{cases}
0& \text{if }\theta=0\mod\frac{2\pi}{3}\,,\\
C_{V_3}&\text{otherwise,}
\end{cases}
\end{align}
where $C_{V_3}>0$ is a large enough constant. Using similar arguments as for $\Z^2$, one can again show the validity of \ref{H1}-\ref{H10}. Using \cite[Proposition~2.6]{Chambolle-Kreutz}, it can be shown that
\begin{align*}
\varphi_{\mathrm{vac}}(\nu) = \frac{1}{6}\min\left\{|\sqrt{3}\nu_1+\nu_2|+|\sqrt{3}\nu_1-\nu_2|,|\sqrt{3}\nu_1+\nu_2|+2|\nu_2|,|\sqrt{3}\nu_1-\nu_2|+2|\nu_2|\right\}\,.
\end{align*}

\end{document}